\documentclass{vldb}
\usepackage{times,amsmath,epsfig}
\usepackage[show]{chato-notes}
\usepackage{graphicx,algorithmic,algorithm,multirow,subfigure,color}
\title{From Group Recommendations to Group Formation} 
\numberofauthors{1}
 \author{
 \alignauthor
  Senjuti Basu Roy$^{\dag}$, Laks V. S. Lakshmanan $^{\diamond}$, Rui Liu$^{\dag}$.  \\
 \affaddr{
  $^{\dag}$ University of Washington Tacoma,
 $^{\diamond}$ University of British Columbia
 }
 {\email{
     \{senjutib,ruiliu\}@uw.edu,laks@cs.ubc.ca}
 }
 }

\newcommand{\eat}[1]{} 
\newcommand{\ratings}{\mbox {$\mathcal{R}$}} 

\newtheorem{theorem}{Theorem}
\newtheorem{lemma}{Lemma}
\newtheorem{definition}{Definition}
\newtheorem{example}{Example}
\begin{document}

%\conferenceinfo{KDD}{'14 NY, USA}

%\title{Group Formation}

\newcommand\Mark[1]{\textsuperscript#1}

%\author{
     
%}

\newcommand{\hide}[1]{}

\maketitle
\begin{abstract}
There has been significant recent interest in the area of group recommendations, where, given groups of users of a recommender system, one wants to recommend top-$k$ items to a group that maximize the satisfaction of the group members, according to a chosen semantics of group satisfaction. Examples semantics of satisfaction of a recommended itemset to a group include the so-called \emph{least misery} (LM) and \emph{aggregate voting} (AV). We consider the complementary problem of \emph{how to form groups} such that the users in the formed groups are most satisfied with the suggested top-$k$ recommendations. We assume that the recommendations will be generated according to one of the two group recommendation semantics -- LM or AV. Rather than assuming groups are given, or rely on ad hoc group formation dynamics, our framework allows a \emph{strategic} approach for forming  groups of users in order to maximize satisfaction. We show that the problem is NP-hard to solve optimally under both semantics. Furthermore, we develop two efficient algorithms for group formation under LM and show that they achieve  {\em bounded absolute error}. We develop  efficient heuristic algorithms for group formation under AV. We validate our results and demonstrate the scalability and effectiveness of our group formation algorithms on two large real data sets.
\end{abstract}

%Under this assumption, the question is how best to assign users to groups such that the resulting group recommendations maximize the satisfaction of the groups.

\section{Introduction}
There is a proliferation of {\em group recommender systems} that cope with the challenge
of addressing recommendations for groups of users.  YuTV~\cite{tvrec} is a TV program recommender for groups of viewers, %to find relevant recommendations
%for the groups. 
LET'S BROWSE~\cite{lieberman1999let} recommends web pages to a group of two or more people who are browsing the web together, and  FlyTrap~\cite{crossen2002flytrap} recommends  music to be played in a public room. What all these recommender systems have in common is that they assume that the groups are \emph{ad hoc}, are {\em formed organically and are provided as inputs}, and  focus on designing the most appropriate {\em recommendation semantics} and effective algorithms. {\sl In all these systems, all members of a group are recommended a common list of items to consume.} Indeed, designing semantics to recommend items to ad-hoc groups has been a subject of recent research~\cite{reco1, polylens,DBLP:journals/pvldb/Amer-YahiaRCDY09,senot2010analysis}, and several algorithms have been  developed for recommending items personalized to \emph{given} groups of users  \cite{DBLP:conf/recsys/BerkovskyF10,DBLP:journals/pvldb/Amer-YahiaRCDY09}.

We, on the other hand, study the {\em flip problem} and address the question, {\em if group recommender systems follow one of these existing popular semantics, how best can we form groups to maximize user satisfaction}. In fact, we pose {\em this as an optimization problem to form groups in a principled manner, such that, after the groups are formed and the users inside the group are recommended  an itemset to consume together, following the existing semantics and group recommendation algorithms, they are as satisfied as possible}. 

%Unlike {\em individual user recommendations} that attempt to maximize the satisfaction of an individual, group recommendation semantics are designed to capture {\em the satisfaction of a group of users, by aggregating the satisfaction of the group members}.
%Given a user population and an underlying group recommendation semantics, we seek to {\em strategically form groups, such that, after the groups are formed and the users inside the same group are recommended  with an itemset to consume together based on the given group recommendation semantics are satisfied as highly as possible}. also stress that we are not proposing any group rec semantics ourselves just adopting popular existing  ones. 
%that is, we address the question, if group recs follow one of these semantics, how best can we form groups to maximize user satisfaction. 
%the RS will keep recommending itemset to group members based on those semantics. 
%For the application scenarios described above, intuitively, this will require grouping those individuals together who are likely to have similar satisfaction in watching the same TV program or listening to music, such that, the recommended items designed for this group (based on a given group recommendation semantics) maximizes the satisfaction of its individual users. 
%While this intuition is easy to understand, incorporating the group recommendation semantics explicitly to govern the group formation phase to maximize user satisfaction remains to be of significant challenge. 

{\bf Applications:} Our {\em strategic} group formation is potentially of interest to all group recommender system applications, as long as they use certain recommendation semantics. Instead of ad-hoc group formation~\cite{tvrec,lieberman1999let,crossen2002flytrap,Pizzutilo:2005:GMP:1366349.1366378,DBLP:conf/recsys/BerkovskyF10,DBLP:journals/pvldb/Amer-YahiaRCDY09}, or grouping individuals based on similarity in  preferences~\cite{DBLP:conf/er/NtoutsiSNK12}, or meta-data (e.g., socio-demographic factors~\cite{tvrec,lieberman1999let,crossen2002flytrap}), we explicitly embed the underlying group recommendation semantics in the group formation phase, which may dramatically improve user satisfaction. For example, two users with similar socio-demographic attributes may still have very distinct preferences for watching TV or listening to music; therefore, a meta-data based group formation strategy may place them in the same group, and thus a group recommendation semantics may end up recommending items to them that are not satisfactory to both. 
\eat{Similarly, two users with similar preferences on most of the items but high disagreement  only on a very few preferred items, may be considered {\em similar and placed in the same group}, while the recommended itemset may not satisfy both of them. Both of these scenarios indicate that the group recommender system may fail to ensure optimal recommendations to the groups, if the underlying group recommendation semantics are not explicitly considered in the group formation phase.} 
We attempt to bridge that gap and propose systematic investigation of the group formation problem. More concretely, our focus in this work is to formalize {\em how to form a set of user groups for a given population, such that, the aggregate satisfaction of all groups w.r.t. their recommended top-$k$ item lists, generated according to existing popular group recommendation semantics using existing group recommendation algorithms, is maximized}. %As such our work can be deployed in an existing group recommender system in a non-intrusive way. 

Our work is also orthogonal to existing {\em market based strategies~\cite{groupon1,groupon2, groupon3} on daily deals sites}, such as Groupon and LivingSocial. Most of these works focus on {\em recommending deals} (i.e., items or bundles of items) to users. They rely on incentivizing formation of groups via price discounting. The utility of such recommendation strategies  is to {\em maximize revenue}, whereas, our group formation problem is purely designed to {\em maximize user satisfaction}. We elaborate on this interesting but orthogonal research direction further in the related work section.

Travel planning for user groups is a popular group recommendation application \cite{DBLP:journals/eswa/GarciaSO11}, where several hundreds of travelers can register their individual preferences to visit certain points of interest (POIs) in a city. A travel agency %is unlikely to provide each of the registered customers individual travel plans that he/she would be most satisfied with. Instead, the agency 
may decide to support, say 25 different user groups. Given these groups, they  accordingly design %be able to only accommodate  
25 different plans, where each plan consists of a list of $5-10$ different POIs tailored to each group. Consequently, the registered customers are to be partitioned to form $25$ different groups and each group will be recommended a  plan of $k$ items ($5\le k\le 10$), based on a ``standard'' group recommendation semantics. 

%Not only travel planning~\cite{DBLP:journals/eswa/GarciaSO11}, 
Other emerging applications, such as, recommending news~\cite{Pizzutilo:2005:GMP:1366349.1366378}, music~\cite{crossen2002flytrap}, book, restaurants~\cite{DBLP:conf/er/NtoutsiSNK12}, and TV programs~\cite{tvrec} to groups, make use of similar settings. The size of the user population, the number of groups, the length of the recommended item list, or the most appropriate group recommendation semantics may be application dependent and are best decided by the domain experts; for example, an online news agency may create hundreds of segments of their large reader-base (with several thousands of users) to serve the top-$10$ news, whereas a TV recommender system may only form a few groups to serve the most appropriate $3$ programs to a family. What ties all these applications together  is the applicability of the same underlying settings.

For all these scenarios,  we study, {\em if group recommender systems follow a given semantics, how best can we can partition the user-base to form a  pre-defined number of groups, such that the recommended top-$k$ items to the groups maximize user satisfaction}. We exploit existing popular group recommendation semantics and do not propose a new one. Clearly, our problem is a {\em non-intrusive addition} to existing operational recommender systems and has clear practical impact in all these applications. 

%Existing research primarily investigates how to design new group recommendation semantics, assuming {\em that the groups are formed organically and provided as inputs, or similar users are clustered together~\cite{DBLP:conf/er/NtoutsiSNK12} to form groups}.

%Many emerging applications naturally require to group individuals and recommend the most appropriate set of items for each group to consume together. For example, a travel agency desires forming multiple traveler groups for purposes of suggesting a tour plan personalized to each of these groups such that the individuals inside a group are most satisfied with the recommended tour plan. As a second example, a book club may like to form  groups of members having common interests in reading at the level of genre or even at the level of specific books so that recommendations to each group are likely to be enjoyed by most, if not all members. As a third example, we might want to form  groups in online social networks consisting of individuals who hold common interests in certain activities.

{\bf Contributions:} Our first contribution is in proposing a formalism to create groups in the context of an existing group recommender system, that naturally fits many emerging applications. In particular, we study the group formation problem under two popular group recommendation semantics, namely \emph{least misery} (LM) and \emph{aggregate voting} (AV)~\cite{reco1, polylens,DBLP:journals/pvldb/Amer-YahiaRCDY09,senot2010analysis,tvrec}. Given an item and a group, LM sets the  preference rating of the item for the group to be the preference rating of the least happy member of that group for that item. On the other hand, AV sets  the preference rating of the item for the group to be the aggregated (sum) preference rating of that item over the members of the group. In this paper, {\sl given a user population of a recommender system, a set of items, and a number $\ell$, we seek to partition the users into at most $\ell$ non-overlapping groups, such that when each group is recommended a top-$k$ list of items under a  group recommendation semantics (LM or AV), the aggregate (sum) satisfaction of the created groups is maximized}. Given a group, its satisfaction with a recommended top-$k$ item list could  be measured in many ways: by considering the group preference of the most preferred item, the least preferred item (i.e., $k$-th recommended item), or the sum of group preferences over all $k$ items. These alternatives are discussed in Section~\ref{sec:model}.

Our second contribution is computational. We provide an in-depth analysis of the group formation problem and prove that finding an optimal set of groups is NP-hard under both group recommendation semantics, LM and AV. We propose two efficient approximation algorithms for group formation under LM with provable theoretical guarantees. In particular, our proposed algorithms {\tt GRD-LM} have   {\em absolute error}~\cite{kann1992approximability} guarantees. Additionally, we also describe efficient heuristic algorithms  {\tt GRD-AV} for forming groups under AV semantics and analyze the complexity of both algorithms. We also propose an integer programming based optimal solution for both LM and AV semantics (referred to Section~\ref{sec:optimal} in appendix) which will not scale, but can be used as a reference with which to calibrate scalable algorithms w.r.t. the quality of the solution, on small data sets.

%Furthermore, we prove that the group formation problem under AV is even harder and can not be approximated blah..  . Consequently, we propose approximation algorithm {\tt GrpFrm-LM}..... 

%\note[Laks]{The last two lines contradict each other. Anyway, they all need editing depending on what we end up proposing and proving.}  

Finally, we conduct a detailed empirical evaluation by using two large scale real world data sets (Yahoo! Music and MovieLens) to demonstrate the effectiveness as well as the scalability of our proposed solutions. We also compare our proposed solutions with intuitive baseline algorithms both qualitatively and w.r.t. various performance metrics. Our experimental results indicate that our proposed algorithms successfully form groups with high satisfaction scores w.r.t. the top-$k$ recommendations made to the groups. Additionally, we demonstrate that the proposed solutions are highly scalable and terminate within a couple of minutes in most cases. Furthermore, we conduct a user study involving users from Amazon Mechanical Turk. Our results demonstrate that our proposed formalism is indeed effective for forming multiple groups, with the individuals being highly satisfied w.r.t. the suggested top-$k$ group recommendations. 
Based on our experimental analysis, our group formation algorithms consistently  outperform the baseline algorithms both qualitatively and on  various performance metrics.
%\note[Laks]{Is that true?}

In summary, we make the following contributions:
\begin{itemize}
\item We initiate the study of how to form groups in the context of recommender systems, considering popular group recommendation semantics. We formalize the task as an optimization problem, with the objective to form groups, such that the aggregated group satisfaction  w.r.t. the suggested group recommendation is maximized (Section~\ref{sec:model}).

\item We provide an in-depth analysis of the problem and prove that finding an optimal set of groups is NP-hard under both LM and AV semantics (Section~\ref{hardness}). We present several simple and efficient  algorithms for group formation (Sections~\ref{sec:apprxalgo} and~\ref{sec:av}). We show that our algorithms for LM semantics under both Min and Sum aggregation, achieve a bounded absolute error w.r.t. the optimal solutions (Section~\ref{sec:apprxalgo}). We also work out a clean integer programming based formulation of the optimal solution under both LM and AV semantics (Appendix~\ref{sec:optimal}). 

\item We conduct a comprehensive experimental study (Section~\ref{exp}) on Yahoo! Music and MovieLens data sets and show that our algorithms are effective in achieving high aggregate satisfaction scores for user groups compared to the optimum, lead to relatively balanced group sizes, high average group satisfaction, and scale very well w.r.t. number of users, items, items recommended, and number of groups allowed. Our user study results demonstrate the effectiveness of our proposed solutions. 
\eat{
validates the effectiveness as well as the scalability aspects of our proposed algorithms.} 
\end{itemize}

Section~\ref{disc} presents extensions of the proposed work. Section~\ref{sec:related} presents related work. Finally, we conclude the paper  in Section~\ref{conc} and outline future research directions. 

\eat{The rest of the paper is organized as follows: Section~\ref{sec:model} contains the preliminaries, our data model, and problem definitions. We present a detailed complexity analysis  of the problems in Section~\ref{hardness}. Sections~\ref{sec:apprxalgo} and \ref{sec:av} include the efficient greedy algorithms and theoretical analyses for group formation under LM and AV, considering Min and Sum aggregation. We present two integer programming based optimal algorithms in Section~\ref{sec:optimal}. Section~\ref{disc} presents discussions of some extensions. Our experimental analysis is presented in Section~\ref{exp} and Section~\ref{sec:related} presents  related work. Finally, we conclude the paper  in Section~\ref{conc} and give future research directions.}

%It also shows the constant absolute error guarantee for our greedy group formation algorithm for LM.

%\note[Senjuti]{Things to discuss with Laks:\\
%1. Where do we add the argument that objective function is maximized when no of groups is $x$ \\
%2. Min and max aggregation and how we should present the IP and other algorithms
%3. User study} 

\section{Preliminaries \& Problem Defintion}\label{sec:model}
In this section, we first discuss the preliminaries and describe our data model. We also present two running examples that are used throughout the paper. Finally, we formalize the group formation problem in sub-section~\ref{sec:probdef}. 

\subsection{Data Model} 
We assume an item-set  $\mathcal{I} = \{i_1,i_2, \ldots, i_m\}$ containing $m$ items and a user set $\mathcal{U}=\{u_1,u_2,,\ldots,u_n\}$ with $n$ users. A group $g$ corresponds to a subset of users, i.e., $g \subseteq \mathcal{U}$. In this paper, we consider recommender systems with explicit feedback, which means users' feedback on items is in the form of an explicit rating $sc(u,i) \in \ratings$, where $\ratings$ is typically a discrete set of positive integers, e.g., $\ratings = \{1, ..., 5\}$, with $r^{min}$ and $r^{max}$ being the minimum and maximum possible ratings respective (e.g., $r^{min}$ may be $0$ and $r^{max}$ may be $5$). Without causing confusion, we also use $sc(u,i)$ to denote the  rating of item $i$ predicted for user $u$ by the recommender system.\footnote{Predicted ratings may be real numbers.} Thus, in general, $sc(u,i)$ denotes \emph{user $u$'s preference for item} $i$, whether user provided or system predicted. We sometimes also refer to $sc(u,i)$ as the \emph{relevance of an item for a user}. The recommended top-$k$ item list  for a group $g$ is denoted $\mathcal{I}^k_g$, where $\mathcal{I}^k_g \subseteq \mathcal{I}$ and $|\mathcal{I}^k_g| = k$. Furthermore, we denote the $k$-th item score for group $g$ as $sc(g, i^k)$, where $i^k$ denotes the $k$-th item (i.e., the worst item) in the top-$k$ item list $\mathcal{I}^k_g$ recommended to $g$. We note that $sc(g,i^k)$ is a  quantity that is defined according to a chosen group satisfaction semantics such as LM or AV, as explained in the next subsection.

\begin{example}\label{ex1}
{\em 
Imagine that the user set \\
$\mathcal{U}= \{u_1,u_2,u_3,u_4,u_5,u_6\}$ contains $6$ members and the itemset  $\mathcal{I} = \{i_1,i_2, i_3\}$ has $3$-items. The user's preference for the itemset is given (or predicted) as in Table~\ref{tab:ex1}. Imagine that the user set needs to be partitioned into at most $3$ groups ($\ell \leq 3$). \qed }  
%Also imagine that $k=1$ item is to be suggested to each group.
\end{example}

\eat{
\begin{tabular}{l*{6}{c}r}
User-item Ratings & $u_1$ & $u_2$ & $u_3$ & $u_4$ & $u_5$  & $u_6$  \\
\hline
$i_1$ 		& 5 & $5$ & 2 & 2 & 1 & 1   \\
$i_2$       & 2 & 2 & 5 & $5$ &  1 & 2   \\
$i_3$       & 1 & 1 & 1 & 1 &  1 & $5$  \\
\end{tabular}

\begin{tabular}{l*{6}{c}r}
User-item Ratings & $u_1$ & $u_2$ & $u_3$ & $u_4$ & $u_5$  & $u_6$  \\
\hline
$i_1$ 	    & 5 & 4 & 2 & 2 & 1 & 1   \\
$i_2$       & 3 & 3 & 5 & 5 & 1 & 2   \\
$i_3$       & 1 & 4 & 1 & 1 & 1 & 4  \\
\end{tabular}

\begin{tabular}{l*{6}{c}r}
User-item Ratings & $u_1$ & $u_2$ & $u_3$ & $u_4$ & $u_5$  & $u_6$  \\
\hline
$i_1$ 	    & 3 & 4 & 2 & 2 & 1 & 1   \\
$i_2$       & 3 & 3 & 5 & 5 & 1 & 2   \\
$i_3$       & 1 & 5 & 1 & 1 & 1 & 5  \\
\end{tabular}}

\begin{table}
\begin{tabular}{l*{6}{c}r}
User-item Ratings & $u_1$ & $u_2$ & $u_3$ & $u_4$ & $u_5$  & $u_6$  \\
\hline
$i_1$ 	    & 1 & 2 & 2 & 2 & 3 & 1  \\
$i_2$       & 4 & 3 & 5 & 5 & 1 & 2   \\
$i_3$       & 3 & 5 & 1 & 1 & 1 & 5  \\
\end{tabular}
\caption{User Item Preference Rating for Example~\ref{ex1}\label{tab:ex1}}
\end{table}

\begin{example}\label{ex2}
{\em 
Imagine that the same user set and with the same itemset has now different ratings, as presented in Table~\ref{tab:ex2}. Let us assume that the user set needs to be partitioned into at most $2$ groups ($\ell \leq 2$). \qed }  
%Also imagine that $k=1$ item is to be suggested to each group.
\end{example}

\begin{table}
\begin{tabular}{l*{6}{c}r}
User-item Ratings & $u_1$ & $u_2$ & $u_3$ & $u_4$ & $u_5$  & $u_6$  \\
\hline
$i_1$ 	    & 3 & 1 & 2 & 2 & 1 & 3   \\
$i_2$       & 1 & 4 & 5 & 5 & 2 & 2   \\
$i_3$       & 4 & 3 & 1 & 1 & 3 & 1  \\
\end{tabular}
\caption{User Item Preference Rating for Example~\ref{ex2}\label{tab:ex2}}
\end{table}

%\note[Laks]{This example looks incomplete. What is the point being made?; Senjuti: I completed the example and used it in the algorithm section} 

%$sc(u_1,i1)=5, sc(u_1,i2)=2, sc(u_1,i3)=1$, $sc(u_2,i1)=5-\epsilon, sc(u_2,i2)=2, sc(u_2,i3)=1$, $sc(u_3,i1)=2, sc(u_3,i2)=5, sc(u_3,i3)=1$, $sc(u_4,i1)=2, sc(u_4,i2)=5-\epsilon, sc(u_4,i3)=1$, $sc(u_5,i1)=1, sc(u_5,i2)=1, sc(u_5,i3)=1$,$sc(u_6,i1)=1, sc(u_6,i2)=2, sc(u_6,i3)=5-\epsilon$.

\subsection{Group Recommendation Semantics}

A group recommendation semantics spells out a numeric measure of just how satisfied a group is with an item recommended to it. 
Two popular semantics of group recommendation that have been employed in the literature on group recommendations are: (i) aggregate voting and (ii) least misery. Given a group $g$ and item $i$, the aggregated voting score of item $i$ for the group is the sum of the preference ratings of item $i$ for each member $u \in g$. On the other hand, the least misery score of item $i$ for group $g$ is the minimum preference rating of item $i$ across all  members of $g$. 

%Intuitively, aggregate voting defines the satisfaction measure as the sum of the scores of the $k$-th recommended item for the group's members. On the other hand, least misery defines this measure as the score of the $k$-th recommended item for the least happy user. 

%\note[Laks]{Where is this notation used? 
%$\mathcal{F}_{AV}$: senjuti: used in algorithm section} 

%\note[Laks]{Do we want to define AV or LM first?} 

\begin{definition}[Least Misery Semantics $\mathcal{F}_{LM}$]\cite{reco1, polylens,DBLP:journals/pvldb/Amer-YahiaRCDY09,senot2010analysis} The group  satisfaction score of item $i$ for group $g$ is the least score of $i$ across  members of $g$, i.e., $sc(g,i) := min_{u \in g} sc(u,i)$. 
\end{definition}

\begin{definition}[Aggregated Voting  Semantics $\mathcal{F}_{AV}$]\cite{reco1, polylens,DBLP:journals/pvldb/Amer-YahiaRCDY09,senot2010analysis} 
The group satisfaction score of item $i$ for group $g$ is the aggregated score of $i$ for all members of $g$, i.e., $sc(g,i) :=\sum_{u \in g} sc(u,i)$. %Let $i^k \in \mathcal{I}^k_g$ be the $k$-th item. 
%Then the group satisfaction measure of $\mathcal{I}^k_g$ for group $g$ is given by $sc(g, i^k) := \sum_{u\in g} sc(u,i^k)$. 
\end{definition}

%Let $i^k \in \mathcal{I}^k_g$ be the $k$-th item. Then the group satisfaction measure of $\mathcal{I}^k_g$ for group $g$ is given by $sc(g, i^k) := min_{u\in g} sc(u,i^k)$. 

\subsection{Group Satisfaction Aggregation}
\label{sec:agg} 	
Given a group $g$ and a list of $k$ recommended items $\mathcal{I}^k_g$,  there are multiple ways of aggregating the scores of the $k$ items in order to define the group $g$'s satisfaction with the recommended list $\mathcal{I}^k_g$. Some of the natural alternatives are described below.

\begin{itemize}
\item Max-aggregation: Satisfaction of the group is measured as the score of the very top item in the list, i.e., $g^s(\mathcal{I}^k_g) := sc(g,i^1)$.
%Max_{\forall i \in \mathcal{I}^k_g} sc(g,i)$.

\item Min-aggregation: Satisfaction of the group is measured as the score of the $k$-th item in the recommended list, i.e., $g^s(\mathcal{I}^k_g) := sc(g,i^k)$. 

\item Sum-aggregation: Satisfaction of the group is measured as the sum of scores of all items in the list, i.e., satisfaction of group $g$ is, $g^s(\mathcal{I}^k_g) := \Sigma_{i \in \mathcal{I}^k_g} sc(g,i)$. It is also possible to design a Weighted Sum aggregation function, where each of the $k$ items is assigned a differential weight correlated with its position $i$. We present a brief discussion of this extension in Section~\ref{disc}. 
%A detailed exploration of the Weighted Sum aggregation is deferred to future work.

\end{itemize}
%Observe that, based on these different aggregation function, the score of the 

%Min_{\forall i \in \mathcal{I}^k_g} sc(g,i)$. 

%\note[Laks]{Not sure if we want to pitch this paper around diversification. I understand your reason. However, ``diversity'' may set different expectations from the readers.} 

Notice that when $k=1$,  Max, Min, and Sum-aggregation coincide. 

%{\bf NEED CHANGE: Even though we do not focus on optimizing for Sum-aggregation, we note that a grouping that is optimized for Min-aggregation is compatible with the objective of reaching a high Sum-aggregation score for the grouping formed. We explore this intuition further in the experiments section, where we measure the average satisfaction score of groups for the groupings generated using our proposed algorithms. 
%In this paper, we primarily focus on Max and Min-aggregation.}

\subsection{Problem Definition}\label{sec:probdef} 
 {\bf  Recommendation Aware Group Formation (GF):} Given items $\{i_1,i_2,\ldots i_m\}$  and users $\{u_1,u_2,\ldots u_n\}$,  a group recommendation semantics %$\mathcal{F}_{AV}$ or $\mathcal{F}_{LM}$, 
LM or AV, two integers $k$ and  $\ell$, create a set of at most $\ell$ non-overlapping groups, where each group $g$ is associated with a top-$k$ itemset $\mathcal{I}^k_g$ in accordance with semantics LM or AV, s.t.:
\begin{itemize}
\item The aggregated group satisfaction of the created groups is maximized; i.e., maximize $Obj =\sum_{j=1}^{\ell}g^s_j(\mathcal{I}^k_{g_j})$.
\end{itemize} 
 
%Recall that the satisfaction score of a given group $g$ with a recommended item list $\mathcal{I}^k_g$ could be measured either by capturing the score of the max (top item), or the min ($k$-th or bottom item), or aggregation (sum of all item scores) of $g$ based on an underlying recommendation semantics. 

%For brevity, we only consider the max and the min aggregation semantics in the rest of the paper. A principled exploration of the sum aggregation paper is deferred to future work.

%\note[Laks]{I did not check the notations in the table below. I hid the table, not serving much %purpose} 
\eat{
\begin{table}
\centering
\caption{Notations and Interpretations}
\begin{tabular} {|c|p{5cm}|}
\hline
  {\bf Notation} & {\bf Interpretation} \\
   \hline 
   $sc(u,i)$ & relevance of item $i$ for user $u$ \\
   \hline 
  $sc(g,i)$ & relevance of item $i$ for group $g$ \\
   \hline 
  $\ell$ &  maximum total number of groups \\
  \hline 
  $y_{jp}$ & a decision variable to denote if item $j$ is the $k$-th item of group $p$ \\
  \hline
  $w_{ip}$ & a decision variable to denote if item $j$ is in the $(k-1)$-th itemset of group $p$ \\
  \hline
  $u_{ip}$ & a decision variable to denote if user  $u_i$ is part of group $p$ \\
  \hline
%  $z_{ip}$ & a decision variable which is same as $w_{ip}$\\
%\hline
\end{tabular}\label{tab:notations}
\end{table}}

%\section{Hardness Results}
%In this section, we discuss the complexity of the group formation problem.
\section{Complexity Analysis}\label{hardness}
In this section, we show that the recommendation-aware group formation (GF) problem is NP-hard. Our hardness reduction is from Exact Cover by 3-Sets (X3C), known to be NP-hard~\cite{DBLP:books/fm/GareyJ79}. Since a direct reduction is involved, we first prove a helper lemma, which shows that a restricted version of Boolean Expected Component Sum (ECS\footnote{Expected Component Sum is also NP-hard \cite{DBLP:books/fm/GareyJ79}.}), called Perfect ECS (PECS for short), is NP-hard. We then reduce PECS to GF. 
\eat{
: (1) We consider a simpler version of GF (referred to as Simpler GF or SGF) for the complexity analysis. SGF is akin to GF, except that it does not have to form the groups such that no two groups have the same recommended top-$k$ item set. To show that SGF is NP-hard, we take a special case of Boolean Expected Component Sum (ECS)  (Boolean ECS is shown to be NP-hard~\cite{DBLP:books/fm/GareyJ79}) for the NP-hardness proof. We reduce this special case  of Boolean ECS (referred to as Boolean Perfect Expected Component Sum or Boolean PECS in brief) to SGF to prove that SGF is NP-hard.  2) We prove that Boolean PECS is also NP-hard; 3) Finally, we argue that our original version of GF must also be NP-hard.
} 

An instance of PECS consists of a collection $V$ of $m$-dimensional boolean vectors, i.e., $V \subset \{0,1\}^m$, and a number $K$. The question is whether there exists a disjoint partition of $V$ into $K$ subsets $V_1, ..., V_K$ such that $\sum_{i=1}^K max_{1\le j\le m}( \sum_{\vec{v}\in V_i} \vec{v}[j] ) = |V|$. We have: 

%Next, we prove that PECS is actually NP-hard.
\begin{lemma}
PECS is NP-complete.
\end{lemma}

\begin{proof}
The membership of PECS in NP is straightforward. 

To prove hardness, we reduce a known NP-Complete problem, namely, Exact Cover by 3-sets (X3C)~\cite{DBLP:books/fm/GareyJ79} to PECS. 

X3C: An instance $\mathcal{I}$ of X3C consists of a ground set $\mathcal{X}$ with $3q$ elements and a collection $\mathcal{C} = \{S_1, ..., S_m\}$ of 3-element subsets of $\mathcal{X}$. The problem  is to find if there exists a subset $\mathcal{C}' \subseteq \mathcal{C}$, such that $\mathcal{C}'$ is an exact cover of $\mathcal{X}$, i.e., each element of $\mathcal{X}$ occurs in exactly one subset in $\mathcal{C}'$. 

Given an instance $\mathcal{I}$ of X3C, we create an instance $\mathcal{J}$ of PECS as follows. Transform each element $x_i \in \mathcal{X}$ into a boolean vector $\vec{v}_i\in\{0,1\}^m$ by setting $\vec{v}_i[j] = 1$ if $x_i\in S_j$ and $\vec{v}_i[j] = 0$ otherwise. Thus, $V = \{\vec{v}_1, ..., \vec{v}_{3q}\}$, where vector $\vec{v}_i$ corresponds to ground element $x_i\in\mathcal{X}$. By construction, subset $S_j\in \mathcal{C}$ corresponds to dimension $j$ of the vectors. Notice that at most three vectors have a $1$ in any given dimension. Thus, instance $\mathcal{J}$ consists of the vectors $V$ and the number $K := q$. 

We claim that $\mathcal{I}$ is a YES-instance of X3C iff $\mathcal{J}$ is a YES-instance of PECS. 

($\Rightarrow$): Suppose $\mathcal{C}' = \{S_{j_1}, ..., S_{j_q}\} \subseteq \mathcal{C}$ is an exact (disjoint) cover of $\mathcal{X}$. Then consider the partition $[V_1, ..., V_q]$ of $V$, where $\vec{v}_i \in V_k$ iff $x_i \in S_{j_k}$. Notice that each $V_k$ consists of exactly three vectors. Since $\mathcal{C}'$ is an exact cover of $\mathcal{X}$, each element $x\in \mathcal{X}$ appears in exactly one subset $S_k \in \mathcal{C}'$. Thus, $\max_{1\le \ell\le m} \sum_{\vec{v}\in T_i} \vec{v}[\ell] = 3$ and hence $\sum_{k=1}^q(\max_{1\le \ell\le m} \sum_{\vec{v}\in T_i} \vec{v}[\ell]) = 3q$, showing $\mathcal{J}$ is a YES-instance. 

($\Leftarrow$): Let $\pi$ be a partition of $V$ such that\\ $\sum_{T\in \pi}(\max_{1\le \ell\le m} \sum_{\vec{v}\in V_i} \vec{v}[\ell]) = 3q$, witnessing the fact that $\mathcal{J}$ is a YES-instance. Observe that any block $T \in \pi$ with $> 3$ vectors in it cannot contribute more than $3$ to the sum above. As well, any block $T \in \pi$ with $< 3$ vectors will surely contribute less than $3$ to the sum above. Since the overall sum is $3q$, it follows that every block must have exactly three vectors in it. For a block $T\in \pi$, let $dim(T) = argmax\{\sum_{\vec{v}\in T} \vec{v}[j] \mid 1\le j\le m\}$, i.e., $dim(T)$ is the dimension which maximizes the component sum. Then consider the collection $\mathcal{C}' = \{S_{dim(T)} \in \mathcal{C} \mid T\in \pi\}$. It is easy to verify that $|\mathcal{C}'| = q$ and that every element $x\in\mathcal{X}$ appears in exactly one set $S\in \mathcal{C}'$. 

\eat{ 
 the ground set $\mathcal{X}$ represents the collection of vectors $C$ in PECS, i.e., $|C|=3q$. Subset $i$ (which is a triplet) in collection $\mathcal{C}$ of X3C represents a dimension $i$ in PECS. In particular, each element $x$  in the triplet $i$ gets  a $1$ for dimension $i$, whereas, all other elements get a $0$ in dimension $i$. Notice, each dimension this way has exactly $3$ and only $3$ 1's, and the remaining $3q-3$ elements get $0$'s. Finally, we set the number of partitions $K$ of PECS to $q$. Given this instance, the 
objective is to create $K$ partitions such that the objective function value of PECS is $|C|$, such that there exists a solution of X3C of size $q$, if and only if, a solution to our instance of Boolean PECS exists. } 

\end{proof}

%We consider a special case of ECS with Boolean vectors(which is shown to be NP-hard~\cite{DBLP:books/fm/GareyJ79}) to a simpler version of GF (SGF) that relaxes the last constraint of every two groups in the created partition have to have different top-$k$ recommended item set. 

\begin{theorem}
The Group Formation Problem is NP-hard under both the least misery and aggregated voting semantics. 
\end{theorem}

\begin{proof} 
We prove hardness of the restricted version of GF, where $k=1$ item is to be recommended to each group such that sum of satisfaction measures of each group is maximized, from which the theorem follows. We prove this hardness by reduction from PECS. Given an instance $\mathcal{I}$ of PECS, consisting of a set of boolean vectors $V \subset \{0,1\}^m$ and an integer $K$, we create an instance $\mathcal{J}$ of GF as follows. Each vector $\vec{v} \in V$ corresponds to a user's preference over the $m$ items, where preferences are binary. The decision version of GF asks whether there exists a disjoint partition of $V$ into $K$ groups $g_1, ..., g_K$ such that $\sum_{j=1}^K \max_{1\le\ell\le m} (min_{\vec{v} \in g_j} \vec{v}[\ell]) \ge K$. We claim that this is true iff $\mathcal{I}$ is a YES-instance. 

($\Rightarrow$): Suppose there are $g_1, ..., g_K$ such that \\ 
$\sum_{j=1}^K \max_{1\le\ell\le m} (min_{\vec{v} \in g_j} \vec{v}[\ell]) \ge K$. This sum can never be $> K$. If we replace the $min$ by a $sum$ in the objective function above, it is easy to see that the value will be exactly $|V|$, showing $\mathcal{I}$ is a YES-instance. 

($\Leftarrow$): Suppose there is a disjoint partition of $V$ into $K$ subsets showing that $\mathcal{I}$ is a YES-instance. Again, replacing the innermost summation by a $min$ in the objective function of PECS, will result in a value of $K$, showing $\mathcal{J}$ is a YES-instance. 

Hardness under the aggregated voting semantics follows trivially from the construction above.

\end{proof}

\eat{
\note[Laks]{The following lemma needs to be properly proved.} 

\begin{lemma}
DGF in general is also NP-hard.
\end{lemma}

\begin{proof}
Since the simpler version of GF, namely SGF is NP-hard, it is easy to notice that the additional constraint of GF, i.e., no two top-$k$ recommended sets can be the same, makes the problem only more difficult.  Therefore, GF is also NP-hard.
\end{proof}
}

\section{Approximation Algorithms: LM}\label{sec:apprxalgo}
%Given that the group formation problem is NP-hard under both LM and AV group recommendation %semantics, 
In this section, we investigate efficient algorithms for group formation based on LM. Notice that when $k=1$, the Max, Min, and Sum aggregation (see Section~\ref{sec:agg}) coincide. When $k>1$, basing the LM score on the bottom item (i.e., the $k$-th item in the top-$k$ list) corresponds to Min aggregation while basing it on the top item corresponds to Max aggregation, and the entire top-$k$ set corresponds to Sum aggregation. Unless otherwise stated, we henceforth focus on Min and Sum aggregation.

%generic approximation algorithm and then further illustrate how it could be instantiated under different group recommendation semantics under different group satisfaction aggregation function.

We propose two greedy algorithms, where both of them have respective {\em absolute error guarantees}. Algorithm {\tt GRD-LM-MIN} is designed for LM considering Min aggregation and Algorithm {\tt GRD-LM-SUM} is for Sum aggregation. 

%that also has a provable guarantee on {\em absolute error} and exploits a similar solution framework. 

For simplicity of exposition, we use our running example~\ref{ex1} and ~\ref{ex2} from Section~\ref{sec:model}. We interleave our exposition of the algorithm with an illustration of how it works on these examples for $k=1$ as well as $k=2$ (i.e., top-$1$ or top-$2$ items are recommended).

 %We also note that $k=1$ corresponds to the Max-Aggregation for any general value of top-$k$ itemset. 
%For our illustration of the algorithm, we use Example~\ref{ex1}. 
\subsection{Min Aggregation}
Our proposed algorithm operates in a top-down fashion and gradually forms the groups. Intuitively, the algorithm consists of the following three high level steps. \\
{\bf Step 1 - forming a set of intermediate groups:} It begins with the user set $\mathcal{U}$ and leverages a preference list ($\mathcal{L}$) of items for each user $u$, sorted in non-increasing order of item ratings. In our running example, for user $u_2$ in Example~\ref{ex1}, $\mathcal{L}^{u_2} = \langle i_3,5; i_2,3; ; i_1,2 \rangle$. After that, the algorithm creates a set of intermediate groups. Each group $g$ {\em contains a set of users who have the same top-$k$ item sequence}, as well as the {\em same preference rating for the bottom item} across all users in the group. E.g., the group $\{u_2, u_6\}$ shares the same top-$1$ item ($i_3$) and the same rating for it ($5$). Thus, for $k=1$, this is a valid group. On the other hand, for $k=2$, even though $u_2$ and $u_6$ share the same top-$2$ sequence of items ($i_3; i_2$), they have distinct ratings for the bottom item, namely $3$ and $2$, and so they cannot be in the same group for $k=2$. 

Assuming Min-aggregation, the interesting observation is that, for each group, it is a good strategy to form these groups on the common top-$k$ item sequence, as long as the group members (users)  match on the preference rating of the bottom item. This is because the objective function (LM score) is based on the ratings of the bottom (i.e., $k$-th) item among group members. On the other hand, a subtle point is that it is {\em not} a good strategy to consider just the bottom item only (even though the aggregation is on that item) instead of the entire top-$k$ sequence. The reason is that the bottom item recommended to a group containing a user $u$ may differ from $u$'s personal bottom item. We next illustrate this with an example. 

\begin{example}\label{ex3} 
{\em 
Consider a group consisting of two users $u_1, u_2$ whose individual ratings over items $i_1, ..., i_3$ are respectively $u_1 = (5,4,1)$ and $u_2 = (1,4,5)$. For $k=2$, the second best (i.e., bottom) item for either user in isolation is $i_2$, and yet under LM semantics, it can be easily verified that the top-$2$ item list recommended to the group $\{u_1, u_2\}$ is $(i_2; i_{*})$, where $i_{*}$ could be any one of the remaining items. Notice that the bottom item recommended to the group is different from the indidual bottom  preference of every group member, even though they all shared the same item $i_2$ as their bottom preference, with identical ratings ($4$). The reason this happened is because $i_2$ ended up having the \emph{highest} LM score for this group, among \emph{all items}. When $i_2$ is moved to the top position, no matter which other item is chosen as the top-$2$ (bottom) item for the group, its LM score is just $1$ in this example. This shows that forming a group solely based on shared bottom item and score can lead to a group with a poor LM score, when $k>1$. \qed } 
\end{example} 

To generalize this observation, our algorithm needs to store the top-$k$ common sequence as well as the rating of the item on which the group satisfaction (i.e., the LM score of that group) is aggregated.
\eat{ 
because, depending on the members in a group the top-$k$ itemset changes (for example, top-$2$ item for $(u_2,u_3)$ is $i_1$, whereas, $u_2$ herself does not have $i_1$ as her top-$1$ or top-$2$ item). 
} 
%\note[senjuti]{need to start here}.
 Next, we describe how one can execute this first step above efficiently.

For every user $u$, we create a sequence $\langle i^u_1,i^u_2,.. i^u_k:sc(u,i^k)\rangle$ comprising her top-$k$ ranked items (in sequence) followed by the preference rating of the $k$-th (bottom) item, $sc(u,i^k)$. We use a hash map to hash each user $u$ using $(\langle i^u_1,i^u_2,.. i^u_k:sc(u,i^k)\rangle)$ as the key and the user id $u$ as the value. 
%Notice that the value reflects the LM score associated with the $k$-th item. 
Then, we create a heap $\mathcal{H}$ to store the LM scores $sc(u,i^k)$ for various users.\footnote{It is sufficient to store the value once per intermediate group formed.} This data structure enables us to efficiently retrieve the highest LM score, needed for Step 2 below. 
In our Example~\ref{ex1}, when $k=1$, we hash  user $u_2$ with key $\langle i_3:5 \rangle$ and value $u_2$. We add the entry $sc(u,i^k)$ to the heap, i.e., $\mathcal{H}.insert(sc(u,i^k))$. Notice that, users with same keys get hashed together and the associated value gets updated with their union after each such operation in the hash map. For example, $\{u_3,u_4\}$ gets hashed together, with key $\langle i_2:5 \rangle$ and value $\{u_3,u_4\}$. Finally, we preserve the association between the hash keys and the corresponding LM scores in another data structure. 
%\note[Laks]{Do we have another data structure besides the heap?!} 
This operation generates the set of intermediate groups, where users with same keys belong to the same group. 

For our running example (Example~\ref{ex1}), when $k=1$, we form the following set of intermediate user groups: $\{u_2, u_6\}$ on item $i_3$,  $\{u_3,u_4\}$ on item $i_2$, and two singleton groups for $\{u_1\}$ and $\{u_5\}$, since they do not share a common top-$1$ item. 
%Notice that generalizing this step to any $k$ over Min-aggregation is no harder. 
When $k=2$, $\{u_3,u_4\}$ will be grouped together  with key $\langle i_2, i_1:2 \rangle$ and value $\{u_3,u_4\}$. This step creates the  following intermediate  groups:  $\{u_3,u_4\}$, and four singleton groups, ($\{u_1\}, \{u_2\}, \{u_5\}, \{u_6\}$). Observe that the sets of intermediate groups generated for $k=1$ and for $k=2$ are different.

{\bf Step 2 - greedy selection of $\ell-1$ groups:} Recall that the group formation problem requires that we form at most $\ell$ groups of users. Observe that the objective function value $Obj$ (see Section~\ref{sec:probdef}) is maximized when all $\ell$ groups are formed. Accordingly, in this step, the algorithm runs in $\ell-1$ iterations. Continuing with Example~\ref{ex1}, suppose $\ell=3$, then this means that this step of the algorithm runs for $2$ iterations. In iteration $i$, it  retrieves the {\em maximum element}  from the heap $\mathcal{H}$ (i.e., highest associated  LM score), extracts the corresponding key and uses that to output the user group from the hash map. After that, it deletes that $\langle key,value \rangle$  entry from the hash map and deletes the corresponding LM score from the heap. 

%After that, it selects that intermediate user group (for our running example, it is either ($u_1, u_2$) or ($u_3, u_4$), and ties are broken arbitrarily) as a separate group which gives rise to the {\em highest increase in the objective function value $Obj$}. More specifically, this splits $\mathcal{U}$ and creates two groups, one with the newly selected user group and the consists of the remaining users. 

When $k=1$, in Example~\ref{ex1}, iteration $1$ outputs the group $\{u_3,u_4\}$ with score $5$ and iteration $2$ outputs $\{u_2, u_6\}$ with score $5$. When $k=2$, for the same example instance, iteration $1$ outputs  the group $\{u_1\}$ with score $3$ and  iteration $2$ outputs $\{u_2\}$ with score $3$. 

{\bf Step 3 - forming the $\ell$-th group:}
Finally, the last group $g_{\ell}$ is formed by considering all the remaining users from the hash map and a top-$k$ LM score is assigned to this group. For our running example (Example~\ref{ex1}), when $k=1$, this group is $\{u_1,u_5\}$ with LM score of $1$. When $k=2$, this last group is  $\{u_3,u_4,u_5,u_6\}$ with LM score of $1$. Our algorithm terminates after this iteration. When $k=1$, the final set of groups are $\{u_3,u_4\}$, $\{u_2, u_6\}$, $\{u_1,u_5\}$ and the corresponding value $Obj$ of the objective function is $5+5+1 = 11$. When $k=2$, the final set of groups are $\{u_1\}$, $\{u_2\}$, $\{u_3,u_4,u_5,u_6\}$. The corresponding value of $Obj$ is $3+3+1=7$. The pseudo-code of the algorithm is presented in Algorithm~\ref{alg:lm}.

%{\bf Generalizing {\tt GrpFrm-LM} for Min-aggregation on any $k$:} 

\begin{algorithm}[t]
\caption{Algorithm {\tt GRD-LM-MIN}}
\label{alg:lm}
\begin{algorithmic}[1]
\begin{small}
\REQUIRE Preference lists $\mathcal{L}^u \forall u \in \mathcal{U}$, $k$, group recommendation semantics $\mathcal{F}_{LM}$, $\ell$; hash map $h$ and heap $\mathcal{H}$;
%\STATE Retrieve preference rating list of each user $u$ sorted in decreasing order of preference;
\STATE for every user $u$, generate top-$k$ preference list $\mathcal{L}^{u^k}$ from $\mathcal{L}^u$;
\STATE for every user $u$, generate top-$k$ item sequence and the $k$-th item score $sc(u,i^k)$, i.e., $\langle i^u_1,i^u_2,.. i^u_k:sc(u,i^k)\rangle$;
\STATE For every distinct $\langle i^u_1,i^u_2,.. i^u_k:sc(u,i^k)\rangle$, create $h.key=\langle i^u_1,i^u_2,.. i^u_k:sc(u,i^k)\rangle$, $h.value=\{u\}$, where each user $u$ has same $sc(u,i^k)$;
\STATE Perform $\mathcal{H}.insert = sc(u,i^k)$, associated with every unique key;
\STATE $\mathcal{G} = \{\text{ values present in h} \}$;
\STATE $Obj = 0$;
\STATE $j=0$;
\WHILE {($j < (\ell-1)$)}

%\STATE $g^s(\mathcal{I}^k_g) = Minimum_{\forall u \in g} sc(u,i), \forall g \in \mathcal{G}$, where $i$ is the Max/Min item in the top-$k$ itemset;
\STATE Retrieve the maximum LM score $SC$ from $\mathcal{H}(Maximum)$;
%\STATE Retrieve $g \in \mathcal{G}$ with the highest $h1.value$.
\STATE Retrieve the sequence key $\mathcal{S}$ associated with $SC$;
\STATE Retrieve the value for the given key, $g = h.key(\mathcal{S})$;
\STATE $Obj = Obj + SC$;
\STATE Remove $\mathcal{S}$ from $h$;
\STATE Remove $SC$ from $\mathcal{H}$;
\STATE $j=j+1$;
\ENDWHILE
\STATE Form group $g_{\ell}$ with all the remaining users in $\mathcal{U}$;
\STATE $g_{\ell}^s(\mathcal{I}^k_{g_{\ell}}) = \min_{\forall u \in g_{\ell}} sc(u,i^k)$;
\STATE $Obj = Obj + g_{\ell}^s(\mathcal{I}^k_{g_{\ell}})$;
\RETURN $g_1,g_2,\ldots g_\ell$ groups and $Obj$;
\end{small}
\end{algorithmic}
\end{algorithm}

In general, the grouping formed by Algorithm {\tt GRD-LM-MIN} is sub-optimal. For  Example~\ref{ex1}, we saw that the objective function value achieved by the grouping found by {\tt GRD-LM-MIN} is $11$. It may be verified that the optimal grouping for Example~\ref{ex1} is $\{u_1,u_3,u_4\}$, $\{u_2, u_6\}$, $\{u_5\}$ with an overall $Obj$ value of $4 + 5 + 3 = 12$. 

\subsection{Sum Aggregation}
The greedy algorithm for Sum aggregation, {\tt GRD-LM-SUM}, exploits a similar framework, except that it primarily differs in {\em Step-1, i.e., in the intermediate groups formation step} of the previous algorithm. Notice that {\tt GRD-LM-MIN} forms these intermediate groups by bundling those users who have the same top-$k$ sequence, as well as the same preference rating for the $k$-th item. Obviously, this strategy falls short for {\tt GRD-LM-SUM}, where, we are interested to aggregate satisfaction over the entire $k$-itemset. Therefore,  {\tt GRD-LM-SUM} forms these intermediate user groups by  hashing users who have not only the {\em same top-$k$ item sequence, but also the same score for each item}. 

More concretely, for every user $u$, we create a sequence $\langle i^u_1:sc(u,i^1),i^u_2:sc(u,i^2),.. i^u_k:sc(u,i^k)\rangle$ comprising her top-$k$ ranked items and scores. We use a hash map to hash each user $u$ using $\langle i^u_1:sc(u,i^1),i^u_2:sc(u,i^2),.. i^u_k:sc(u,i^k)\rangle$  as the key and the user id $u$ as the value. Then, in the heap $\mathcal{H}$, we store the aggregated LM scores $sc(u,\sum_{j=1}^k i^j)$ for various users - i.e., we perform, $\mathcal{H}.insert(sc(u,\sum_{j=1}^k i^j))$.
In our Example~\ref{ex1}, when $k=2$, this way, users  $u_3$ and $u_4$ are hashed together with key $\langle i_2:5, i_1:2 \rangle$ and value $\{u_3,u_4\}$. The other four users get hashed in $4$ individual buckets, as their top-$2$ item rating sequences do not match. We insert a score of $5+2=7$ in the heap for users $u_3$ and $u_4$, and the respective sum scores of top-$2$ items for the other four users. 

{\em Steps 2 and 3} of {\tt GRD-LM-SUM} are akin to that of {\tt GRD-LM-MIN}, except for the obvious difference, that the group satisfaction score is now aggregated over all $k$-items. We omit the details for brevity. Using Example~\ref{ex1}, this algorithm may form the following three groups at the end - $\{u_3,u_4\}$, $\{u_1, u_5, u_6\}$, $\{u_2\}$ with the total objective function value as $(5+2)+(1+1)+(5+3)=17$. Appendix~\ref{addex} gives an example where the grouping produced by {\tt GRD-LM-SUM} is suboptimal. 

\subsection{Analysis}\label{sec:analyseslm}
In this subsection, we analyze the performance of the greedy algorithms for LM, compared to %. In particular, we focus on the quality of the generated groups compared with that of the 
optimal grouping. We also analyze the running time complexity of the algorithms. 

{\bf Approximation Analysis:} Our analysis focuses on the absolute error of the greedy algorithms, compared to their optimal solutions. 

%We recall the definition of absolute error. 

\begin{definition}
[Absolute error~\cite{kann1992approximability}] 
{\em 
Let $\Pi$ be a problem, $I$ an instance of $\Pi$, $A(I)$ be the solution provided by algorithm $A$, $f(A(I))$ be the value of the objective function for that solution. Let $OPT(I)$ be the value of the objective function for an optimal solution on instance $I$. We say Algorithm $A$ solves the problem with a guaranteed absolute error of $\eta$, provided for every instance $I$ of $\Pi$, $|f(A(I)) - OPT(I)| <= \eta$. \qed 
} 
\end{definition}

%As before, we first present the analysis for {\tt GRD-LM-MIN}. 

\begin{theorem}
Algorithm {\tt GRD-LM-MIN} solves the group formation problem under LM semantics using min-aggregation with a  guaranteed absolute error of at most $r^{max}$, where $r^{max} \in \mathcal{R}$ is the maximum value in the rating scale.
\end{theorem}

\begin{proof}
Let $g_1,g_2,\ldots g_\ell$  (resp., $g'_1,g'_2,\ldots g'_\ell$) be the set of groups formed by Algorithm {\tt GRD-LM-MIN} (resp., the optimal algorithm). Sort the groups w.r.t. the group's LM score and assume that the orders $g_1, ..., g_{\ell}$ and $g_1', ..., g_{\ell}'$ are non-increasing in the group LM scores. 

Recall that $\mathcal{I}^k_g$ denotes the top-$k$ item list recommended to a group $g$ under the group recommendation semantics under consideration, which, for this theorem, is LM. Similarly, $g^s(\mathcal{I}^k_g)$ denotes the LM score of group $g$ w.r.t. the recommended top-$k$ item list $\mathcal{I}^k_g$. Let ${\tt OPT}(i) := \sum_{j=1}^i g'^s_j(\mathcal{I}^k_{g'_j})$ and 
${\tt GRD}(i) := \sum_{j=1}^i g^s_j(\mathcal{I}^k_{g_j})$, i.e., they are the partial sum of LM scores of the first $i$ groups formed by the two algorithms. 
Clearly, ${\tt OPT}(\ell)$ and ${\tt GRD}(\ell)$ are the final objective function values achieved by both algorithms and ${\tt OPT}(\ell) = {\tt OPT}(\ell-1) + g'^s_{\ell}(\mathcal{I}^k_{g'_{\ell}})$ and similarly, 
${\tt GRD}(\ell) = {\tt GRD}(\ell-1) + g^s_{\ell}(\mathcal{I}^k_{g_{\ell}})$. 
%%%%%%%%%%%%%%%%%%%%%%%%%%%%%% 
\eat{ 
denote the objective function value of the optimal algorithm and that of Algorithm GRD-LM respectively over all $\ell$ groups. First, for both algorithms, we rank these created $\ell$ groups based on their respective group satisfaction score in a descending fashion, ${\tt GRD}(\ell) = {\tt GRD}(\ell-1) + g_{\ell}^s(\mathcal{I}^k_{g_{\ell}})$ and ${\tt OPT}(\ell) = {\tt OPT}(\ell-1) + {g'}_{\ell}^s(\mathcal{I}^k_{{g'}_{\ell}})$, i.e., the overall objective function value is the summation of the group satisfaction score of the best $\ell-1$ groups and that of the $\ell$-th (i.e, the last) group. 
} 
%%%%%%%%%%%%%%%%%%%%%%%%%%%%%%% 
The overall proof hinges on the following points.

(1) The aggregated satisfaction score generated by the best $\ell-1$ groups of {\tt GRD}  is always larger or equal to that of {\tt OPT}, i.e., ${\tt GRD}(\ell-1) \geq {\tt OPT}(\ell-1)$. We prove this below with a simple domination argument. 
%\begin{itemize}
%\item 

First, notice that the LM score of $g_1$ cannot be less than that of $g_1'$. By construction, $g_1$ consists of a set of users who are indistinguishable w.r.t. their top-$k$ list and bottom item score, which is what determines the LM score of a group; besides, $g_1$ has the highest such bottom score and hence highest possible LM score among all possible groups. Let $i<\ell$ be the smallest number such that the LM score of $g_i$ is less than that of $g'_i$. Removal of any user from $g_i$ leaves its LM score unchanged while adding a new user to $g_i$ cannot possibly increase its LM score. This contradicts the assumption. We just showed for $1\le i<\ell$, $g_i$ dominates $g_i'$. Thus, ${\tt GRD}(\ell-1) \geq {\tt OPT}(\ell-1)$. 
%%%%%%%%%%%%% 
\eat{ 
The group satisfaction of the $i$-th group (out of the first $\ell-1$ groups)  produced by the greedy algorithm is dominated by the $i$-th group produced by the optimal algorithm. This is indeed true, as the {\em while loop (lines 8-16)} of  {\tt GRD-LM} forms the first $\ell-1$ groups by considering the $\ell-1$ highest distinct top-$k$ item-rating pairs considering all users and items. Therefore, for any of the $i$-th group, $g_{i}^s(\mathcal{I}^k_{g_{i}}) \geq  {g'}_{i}^s(\mathcal{I}^k_{{g'}_{i}})$. 
\item Thus,  ${\tt GRD}(\ell-1) \geq {\tt OPT}(\ell-1)$. 
} 
%%%%%%%%%%%%%
%\end{itemize}

(2) For the $\ell$-th group, notice that $g_{\ell}^s(\mathcal{I}^k_{g_{\ell}}) \geq {g'}_{\ell}^s(\mathcal{I}^k_{{g'}_{\ell}}) - r^{max}$, since the maximum possible rating value is $r^{max}$. 
%%%%%%%%%%%%% 
\eat{ 
This is also true, because the satisfaction score of the very last group $\ell$ is between $r^{min}$ and $r^{max}$. Given any two groups, the difference in their satisfaction is therefore, at most $r^{max}$. Thus, even at the worst case, ${g'}_{\ell}^s(\mathcal{I}^k_{{g'}_{\ell}})  - g_{\ell}^s(\mathcal{I}^k_{g_{\ell}}) \leq r^{max}$. } 
%%%%%%%%%%%%% 
%\end{enumerate}

The theorem follows from (1) and (2). 
\end{proof}

\begin{theorem}\label{th:3}
Algorithm {\tt GRD-LM-SUM} solves the group formation problem under LM semantics using Sum aggregation with a  guaranteed absolute error of at most $k \times r^{max}$, where $r^{max} \in \mathcal{R}$ is the maximum value in the rating scale.
\end{theorem}

\begin{proof}
(Sketch): We omit the details for lack of space but note that this proof uses similar reasoning as above. Akin to {\tt GRD-LM-MIN}, only the $\ell$-th group of {\tt GRD-LM-SUM} is subject to error compared to OPT, where each of the $k$-items can accrue at most  $r^{max}$ error. Therefore, the aggregated error over all $k$ items, i.e., the absolute error of {\tt GRD-LM-SUM}, is upper-bounded by $k \times r^{max}$.
\end{proof}

%\notes[senjuti]{LAKS: would you please let me know what you think about this O(logn) time operation of max element. Max retrieval is constant but the adjustment takes O(logn) operations).

{\bf Running time  complexity:}
We first describe the running time complexity of {\tt GRD-LM-MIN}.
{\em Line 2} of Algorithm {\tt GRD-LM-MIN}   takes $O(nk)$ time overall to produce top-$k$ item list per user. 
%\note[Laks]{For each user, retrieving the top-k items from the preprocesed preference list should take O(k) time so this shd be O(nk), right?} 
{\em Line 3} takes $O(n)$ time to hash all users. Adding LM score to the heap also takes $O(n)$ time overall. The {\em while loop} runs for $(\ell-1)$ iterations and in each iteration the highest LM score is obtained in constant time from the heap, and rebuilding the heap takes $O(\log n)$ time. Therefore, the entire {\em while loop} takes $O(\ell \log n)$ time. Forming the $\ell$-th group ({\em lines 17-18}) can take at most $O(n\log k)$ time. Therefore, the overall complexity of the algorithm is $O(nk+ n+\ell\log n)$ or simply $O(nk+\ell\log n)$. 
Similarly, it can be shown that the running time of {\tt GRD-LM-SUM} is also  $O(nk+\ell\log n)$.
%\begin{theorem}
%The approximation factor of {\tt GrpFrm-LM} based on Min aggregation is blah....
%\end{theorem} each iteration takes $O(\log n)$ time  and 

%In the first iteration, it performs a group by on common top-$k$ itemset and forms a set of intermediate groups. Notice that a stark difference between this algorithm and the approximation algorithm for LM is that, {\tt GrpFrm-AV} merges groups considering only common top-$k$ items but not their respective ratings.The intuition here is two users having the common top-$k$ itemset (irrespective of their individual ratings for those items) are always better to be together for AV semantics, as opposed to splitting them across different groups.

\section{On Approximation Algorithms : AV}\label{sec:av}
Next, we describe algorithms to produce groups with high satisfaction scores under the semantics of aggregate voting (AV) considering both Min and Sum aggregation. Unlike least misery (LM), aggregate voting defines the satisfaction score of a group as the sum of the preference scores of the individual users in the group, for the recommended top-$k$ itemset. 

An insight for forming good groupings under the AV semantics is that users who share the same top-$k$ sequence of items could be grouped together, irrespective of the underlying aggregation function (Min/Max/Sum). Notice that the grouping principle differs from that used by the greedy algorithms for LM, which look for not only common top-$k$ item sequence but also a common rating for the bottom ($k$-th)  item ({\tt GRD-LM-MIN}) or for all $k$ items ({\tt GRD-LM-SUM}). To see why the above grouping principle is intuitive, notice that a group formed in this way preserves the personal top-$k$ list associated with each group member. Secondly, the contribution of this group to the overall satisfaction score of the grouping is the sum of ratings of the bottom item (Min aggregation), or all $k$-items (Sum aggregation). Two users who have the same sequence of top-$k$ item sequence therefore are best grouped together, irrespective of their individual item preference. Thus, grouping on item's score is not a useful operation for AV semantics. 

We devise two algorithms {\tt GRD-AV-MIN} (for Min aggregation) and  {\tt GRD-AV-SUM} (for Sum aggregation) that exploit the same algorithmic framework as that of greedy algorithms for LM.

{\bf Min Aggregation:} {\tt GRD-AV-MIN} also runs in a top-down manner (starting with a single group with all users and forming a set of intermediate groups from there) and consists of three primary steps. Computationally, it has only two major differences with {\tt GRD-LM-MIN}, described next:

(1) Consider {\em Lines 2 and 3} of Algorithm~\ref{alg:lm} which hash every unique {\em top-$k$ item sequence and the bottom item score} in the hash map. By contrast, as explained above, {\tt GRD-AV-MIN}  hashes {\em only the top-$k$ item sequence and not the $k$-th item score}. %The intuitive reason is, users with the same top-$k$ order are natural candidates to be grouped together, irrespective of what their individual item scores are. 
Because of this, {\tt GRD-AV-MIN} is likely to generate fewer unique hash keys (and hence fewer intermediate groups). This observation is corroborated by our experiments, in Section~\ref{quality}.

(2) The other difference is more obvious: the group satisfaction score is computed differently in {\tt GRD-AV-MIN} compared to {\tt GRD-LM-MIN}. What we store in heap $\mathcal{H}$ in {\em line 4} is the aggregated group satisfaction score $\sum_{u} sc(u,i^k)$, where each user $u$ has the same top-$k$ item sequence, and $sc(u,i^k)$ is their respective bottom item score.

%Other than these two differences, 
The remaining operations of Algorithms {\tt GRD-AV-MIN} and \\ {\tt GRD-LM-MIN} are essentially similar. 

%We now illustrate {\tt GRD-AV-MIN} using Example~\ref{ex2}.

Consider Example~\ref{ex2} and assume the groups are to be formed using Min-aggregation function over top-$2$ ($k=2$) recommended itemset under AV. 

{\bf Step-1} of {\tt GRD-AV-MIN} will only group $\{u_3,u_4\}$ together as they have the same top-$2$ item sequence, $h.key=\langle i_2,i_1 \rangle$ and $h.value=\{u_3,u_4\}$. The heap $\mathcal{H}$ will insert $4$ as the corresponding AV score. The remaining four users will form singleton groups. 

{\bf Step-2} of {\tt GRD-AV-MIN} will have only one iteration (as $\ell-1=1$). It will retrieve that element from the heap with the highest AV score on the top-$2$ item for item $i_1$, which is $4$. Consequently, it will produce $\{u_3,u_4\}$ as the first group. The top-$2$ itemset for this group will be $(i_2,i_1)$.

{\bf Step-3} of {\tt GRD-AV-MIN} will form the second group by merging the remaining singleton groups into $\{u_1,u_2,u_5,u_6\}$. The AV score on the top-$2$ item is $9$ considering item $i_2$. This group will be recommended the following top-$2$ itemset, $(i_3,i_2)$. The algorithm terminates after that and achieves the objective function value $13$.

Notice that, {\tt GRD-AV-MIN} may produce sub-optimal answers as well. For Example~\ref{ex2}, the optimal two user groups are instead $\{u_1,u_3,u_4\}$, $\{u_2,u_5,u_6\}$.  In this case, the first group has the same recommended item list as that of the first group of {\tt GRD-AV-MIN}, however, the second group has  $\{i_2,i_3\}$ as the recommended itemset. The overall objective function value is now $14$, which is higher.

{\bf Sum Aggregation:} Operationally, there is no difference between {\tt GRD-AV-MIN} and {\tt GRD-AV-SUM}, except for the obvious difference, that the latter aggregates the group satisfaction score over the entire $k$-itemset (not just on the $k$-th item). Using Example~\ref{ex2} again,  Step-1 of {\tt GRD-AV-SUM} will group $\{u_3,u_4\}$ together, as they have the same top-$2$ item sequence, $h.key=\langle i_2,i_1 \rangle$ and $h.value=\{u_3,u_4\}$, but   will insert $(5+2)+(5+2)=14$ as the corresponding AV score in the heap. Other than that, the rest of the users will form four singleton groups. {\tt GRD-AV-SUM} will result in the same set of user groups as that of {\tt GRD-AV-MIN} but the overall objective function value is $14+20=34$, as the second group  $\{u_1,u_2,u_5,u_6\}$ will now have a satisfaction score of  $(4+3+3+1)+(1+4+2+2)=20$ based on the Sum-aggregation.

\subsection{Analysis}
%\note[Senjuti]{We need to come back here} 
Akin to Section~\ref{sec:analyseslm}, we present both qualitative and runtime analyses of the greedy algorithms for AV. \\

%\note[senjuti]{needs to be decided}
{\bf Qualitative Analysis:}
Unlike {\tt GRD-LM} algorithms, greedy algorithms for AV do not come with any guarantees about the total satisfaction score of the grouping they provide. While at this time the approximability of optimal group formation under AV semantics is open, we conjecture that the problem is MAX-SNP-Hard~\cite{kann1991maximum} and cannot be approximated within a constant factor. We next give an example to bring out the subtleties of AV semantics. The point is that, {\sl by grouping a user $u$ with others such that the resulting top-$k$ order is personally arguably worse for user $u$, can still produce a group with higher group satisfaction score, than if $u$ had been grouped with users with the same top-$k$ item list.} 

%This is in contrast with optimal group formation for LM semantics, which enjoys a constant  absolute error approximation.

\begin{example}\label{ex4} 
{\em  Consider four users $u_1, ..., u_4$ and two items $i_1, i_2$. Let the ratings for the users, respectively be $u_1 = (5,4)$, $u_2 = u_3 = (4,5)$, and $u_4 = (3, 2)$. Let $k=2$. Suppose we wish to form two groups. Considering Min aggregation, grouping based on common top-$2$ item list would produce the groups $\{u_1, u_4\}$ (satisfaction score $4+2 = 6$) and $\{u_2, u_3\}$ (satisfaction score $4+4=8$) for an overall satisfaction score of $14$. However, suppose $u_1$ is grouped together with $u_2, u_3$ and $u_4$ is left alone. The top-$2$ list for the group $\{u_1, u_2, u_3\}$ becomes $(i_2; i_1)$, whereas $i_1$ is $u_1$'s favorite. Yet, the satisfaction scores are $5+4+4 = 13$ for the first group and $2$ for the second, for a total satisfaction score of $15$. Even though $u_1$'s top-$2$ order changed to something sub-optimal for $u_1$, the overall satisfaction has improved! This kind of behavior is impossible under LM semantics. This illustrates that it's tricky to reason about forming groups that even approximate the optimal satisfaction score for AV semantics. \qed } 
\end{example}

% Moreover, we also note that the objective function is neither {\em sub-modular nor super-modular partition function}~\cite{part1,part2} under AV, which indicates that the group formation problem under AV is unlikely to be approximable within a constant factor. The detailed proofs are discussed in the appendix. 

%\note[Laks]{Double-check!} 
{\bf Running time complexity:} Running time of the greedy algorithms of AV is similar to that of the LM algorithms, except that the group satisfaction score needs to iterate over all the users to compute the sum or iterate over all $k$ items for the Sum-aggregation. Therefore, adding AV scores to the heap for {\tt GRD-AV-MIN} and {\tt GRD-AV-SUM} now take $O(nk)$ time. The {\em while loop} will overall take $O(\ell \log n)$ time. The last step will now take $O(nk)$ time. Therefore, both the algorithms accept same time complexity, i.e., $O(nk+ \ell \log n)$.

\vspace{-0.05in}
\section{Discussion}\label{disc}
\vspace{-0.05in}
In this section, we present some extensions to our proposed group formation framework. In particular, we describe how to extend Sum Aggregation to consider differential weights, as briefly discussed in Section~\ref{sec:agg}. Intuitively, {\em  Weighted Sum Aggregation} can assign different weights to the top-$k$ items and not treat them equally. We next describe two natural alternatives. %and defer in-depth empirical analysis to future work.

%The weights can be computed using one of the many methods that reflects effectiveness of ranking or in a static manner, as proposed in the IR literature~\cite{baeza1999modern}.

{\bf Weights at the item list level:} For any group, we can assign a weight to each of the top-$k$ recommended items, where the weights could simply be inversely proportional to the position or its logarithm. This way, the top items will have higher weight than the lower ones. Then instead of Sum, we compute Weighted Sum LM or AV. This extension does not introduce any  complications to our proposed algorithms, as we only need to consider the weights when the overall objective function value is calculated.

{\bf Weights at the user level:} A more interesting scenario is to consider weighted aggregation at the user level.  More specifically, how satisfied an individual user is with the recommended top-$k$ items could be measured using IR techniques, such as, {\em NDCG} (Normalized Discounted Cumulative Gain)~\cite{baeza1999modern}. Using a graded relevance scale, NDCG computes the user satisfaction (i.e., gain) for an item, given its position in the result list. The gain is then aggregated from the top of the item list to the bottom and the gain of each item is discounted at lower ranks. After weighted satisfaction of each user is computed, any group recommendation semantics (such as, LM or AV) could be used to compute the group satisfaction. Our proposed solutions do not require any significant change even here, except for the fact that the user satisfaction will be computed in a weighted fashion  that our objective function must account for. Notice that, except for the $\ell$-th group in our greedy algorithm, all the users in the first $\ell-1$ groups are fully satisfied, i.e., the recommended top-$k$ lists exactly match their individual top-$k$ lists. Only for  users in the $\ell$-th group, dissatisfaction may occur, which does not affect the theoretical guarantees. 

\section{Experimental Evaluations}\label{exp}
We evaluate our proposed algorithms w.r.t. their 
effectiveness and efficiency. We also conduct a small scale  user study on
Amazon Mechanical Turk (AMT) to evaluate effectiveness.

The development and experimentation environment uses Python on a 2.9 GHz Intel Core i7 with 8 GB of memory using OS X 10.9.5 OS. We use IBM CPLEX for solving the IP instances. All numbers  are presented as the average of three runs.

{\bf Datasets:}
(1) Yahoo! Music: This dataset represents a snapshot of the Yahoo! Music community's preferences for various songs. Standard pre-processing for collaborative filtering and rating prediction was applied to prepare this data set. The data has been trimmed so that each user has rated at least 20 songs, and each
song has been rated by at least 20 users. The data has been
randomly partitioned so as to correspond to 10 equally sized sets of users, in order  to
enable cross-validation. We use a subset of this dataset in our experiments. 
The ratings values are on a scale from 1 to 5, 5 being the best. More information about this dataset can be found at Yahoo! Research Alliance Webscope program\footnote{\small http://research.yahoo.com}.  

(2) MovieLens:  We use the MovieLens 10M ratings dataset~\footnote{\em \small
http://movielens.umn.edu}. MovieLens is a collaborative rating
dataset where users provide ratings ranging on a 1--5 scale. Table~\ref{tbl:ml1m} contains the statistics of both these datasets. Additionally, Section~\ref{exp:usrstudy} conducts a user study using {\em Flickr data}.

{\bf Algorithms Compared:} In addition to the greedy algorithms (Sections~\ref{sec:apprxalgo} and \ref{sec:av}), we also developed optimal algorithms for group formation, based on integer programming (IP) (Appendix~\ref{sec:optimal}). 
\eat{We compare our proposed solutions {\tt GRD-LM} and {\tt GRD-AV} described in Section~\ref{sec:apprxalgo} and~\ref{sec:av} with the IP based optimal algorithms {\tt OPT-LM} and {\tt OPT-AV} (referred to Section~\ref{sec:optimal} for the IP formulation), whenever feasible. The optimal algorithms do not complete  in reasonable time beyond $200$ users, $100$ items, and $10$ groups.} 
In addition to these, we also implemented two baseline algorithms ({\tt BaseLine-LM} and {\tt BaseLine-AV}), described below, by adapting prior work~\cite{DBLP:conf/er/NtoutsiSNK12}. All three aggregation functions (Min/Max/Sum) are considered.

\begin{table}[h]
  \centering
    \begin{tabular}{|r|r|r|r|}
    \hline
   dataset name & \# users & \# items \\
    \hline
    Yahoo! Music & 200,000 & 136736 \\
    \hline   
    MovieLens & 71,567  & 10,681   \\
    \hline
    \end{tabular}
\vspace{-0.1in}
\caption{Dataset Descriptions\label{tbl:ml1m}}
\end{table}
%\vspace{-0.2in}

The baseline algorithms work as follows: For every user pair $u,u'$, we measure the Kendall-Tau distance~\cite{romesburg2004cluster} between them based on their individual ranking of items, induced by the ratings they provide. This way, we obtain  $dist(u,u')$ for each $u,u'$. Notice that it is not {\em sufficient to consider only top-$k$ items to compute this ranked distance, because two users may have a very small overlap on their top-$k$ itemset; therefore, we consider all the items to obtain $dist(u,u')$}. After that, we use {\em K-means clustering}~\cite{DBLP:books/mk/HanK2000} to form a set of $\ell$ user groups. Once these groups are formed, for each group, we compute the  top-$k$ item list and respective group satisfaction scores (using Min/Max/Sum aggregation) based on  LM or AV semantics. We aggregate these scores over $\ell$ groups to produce the final objective function value. The maximum number of iterations in the clustering is set to $100$ by default.

{\bf Experimental Analysis Setup:} Wherever applicable, we compare the aforementioned algorithms both qualitatively and quantitatively. For evaluation of quality, we measure the {\em objective function value} (for AV or LM), as well as {\em average group satisfaction score on the top-$k$ item lists across the groups, $\frac{\sum_{x=1}^{\ell}\sum_{j=1}^{k} \overline{sc(g_x,i^j)}}{\ell}$}, where $\overline{sc(g_x,i^j)}$ is the average $j$-th item for group $g_x$. We additionally present the distribution of group sizes to examine whether our solution can give rise to many degenerated groups (i.e., singleton groups). For scalability experiments, we primarily measure the clock time to produce the groups and their respective top-$k$ item list. %Notice that generation of a group's top-$k$ item list is needed to compute the group satisfaction score. 
We typically vary number of users ($n$), number of  items ($m$),  number of  groups ($\ell$), and the number of recommended items ($k$). In user study, we evaluate the effectiveness of our group formation algorithms compared with the baselines.

%\note[senjuti]{can we drop figure-2 and related discussion completely? now we have sum}

%Notice, none of the algorithms we discuss is optimized for the latter score. They are only designed to optimize the sum of group satisfaction scores on the $k$-th item, which in turn only focuses on the ratings of the bottom item. The reason for measuring it is that, it does give an indication of how satisfied a group is overall, considering all the $k$ items recommended to it. Measuring it would test whether and how satisfaction of a group with its bottom item captures satisfaction with the whole top-$k$ item list it is recommended.

{\bf Preview of Experimental Results:} Our key findings are: (1) We find that our proposed group formation algorithms effectively maximize the objective function compared to the optimal algorithms, i.e., the average group satisfaction, for all three aggregation functions (Min/Max/Sum). (2) Our results indicate the practical usefulness of Min aggregation, where {\em the average aggregate group satisfactions over the entire top-$k$ item lists} are presented.  Even though Min aggregation only optimizes on the $k$-th item, our results demonstrate high aggregate user satisfaction over the entire list. (3) We observe that our solution produces groups that are quite balanced in size, i.e., the variation in size is small. This observation establishes that our greedy algorithms are also practically viable. (4) We observe that our proposed algorithms are scalable and form groups efficiently, even when the number of users, items, or groups is large. We also observe that we outperform the baseline algorithms quite consistently in all cases -- both qualitatively and w.r.t. performance (efficiency, satisfaction scores). (5) Our user study results indicate that, with statistical significance, our optimization guided group formation algorithms produce user groups, in which, the users are more satisfied with the top-$k$ recommendations than that of the baseline algorithms. This observation is consistent across the datasets. Sections~\ref{quality},~\ref{scale}, and~\ref{exp:usrstudy} present the quality, scalability, and the user study results, respectively.

\eat{Next, we present the results of our quality experiments followed by the scalability results for both the datasets. Finally, we present the user study result.} 
For lack of space, we only present a subset of results. The results are representative and the omitted ones are similar.

\vspace{-0.1in}
%\pagebreak
\subsection{Quality Experiments}\label{quality}
The IP-based optimal algorithms do not complete in a reasonable time, beyond $200$ users, $100$ items, and $10$ groups. Our default settings are as follows: number of users = $200$, number of items = $100$, number of groups = $10$, $k=5$ and Max-aggregation. We vary \# users, \# items, \# groups, and $k$ in the top-$k$ list.  We  measure two quality metrics: (1) the objective function value, i.e., the total satisfaction score of a grouping under LM or AV semantics, (2) the average group satisfaction score over all the recommended top-$k$ items, (3) present the distribution of group sizes for both LM and AV.

%The first $3$ are done under  Max-aggregation, whereas, $k$ is varied under Min-aggregation.

{\bf Interpretation of Results:} We observe that the {\tt GRD} algorithms outperform the corresponding  baseline algorithms, over both of these datasets. With increasing number of users, the objective function value as well as the average group satisfaction on the recommended top-$k$ itemset decrease for a given value of the number of groups $\ell$, as larger number of users typically add more  variance in user preference. On the other hand, with increasing number of groups, both of these values increase, as there is more room for similar users to be grouped together, thereby improving overall satisfaction. With increasing $k$, both of these values decrease again (except Sum aggregation).
%, as group satisfaction score typically decreases with higher values of $k$. 

\subsubsection{Measuring Objective Function}
For lack of space, we only present the results for Yahoo! Music dataset. Results on MovieLens are similar. 

{\bf Number of users:} We vary the number of users and use the default settings for the rest of the parameters. Figure~\ref{fig:obj-lm-user} depicts the results. With increasing number of users, the objective function value decreases in general, because, more users typically introduce larger variation in the preference and smaller LM score. The results also clearly demonstrate that {\tt GRD-LM-MAX} consistently outperforms {\tt Baseline-LM-MAX} and achieves an objective function value comparable to {\tt OPT-LM-MAX}.

{\bf Number of items:} Figure~\ref{fig:obj-lm-item} shows that with increasing number of items, the objective function value typically increases. This is also intuitive, because, the top items of each group are likely to improve leading to higher LM score. Also, {\tt GRD-LM-MAX} consistently outperforms the {\tt Baseline-LM-MAX}.

{\bf  Number of groups:} When number of groups is increased, the overall objective function value also increases. Notice that the objective function reaches its maximum possible value  when number of groups equals the number of users. Therefore, with more groups, users get the flexibility to be with other users who are ``similar'' to them. Figure~\ref{fig:obj-lm-grp} presents these results. Similar to the above two cases, {\tt GRD-LM-MAX} outperforms the baseline.

{\bf Top-$k$ on Min and Sum aggregation:} 
\eat{The objective function aggregates the LM group satisfaction over the most preferred item in the previous three experiments.} 
In this experiment, we vary $k$ and produce the objective function value over the least preferred (i.e., bottom) item on the recommended  top-$k$ item list. For Min-aggregation, Figure~\ref{fig:obj-lm-k} shows that with increasing $k$, the objective function value decreases across the algorithms. This is natural, because, group satisfaction based on the bottom element typically decreases with increasing value of $k$. 
Figure~\ref{fig:obj-lm-k-sum} shows the Sum aggregation , where the objective function value increases across the algorithms with increasing $k$, although the rate of increase is smaller for higher $k$. These results demonstrate that {\tt GRD-LM} algorithms are highly effective, their respective objective function values are quite comparable with that of {\tt OPT-LM}. 
 
%{\bf Top-$k$ on Sum-aggregation:} Finally, we vary $k$ and produce the objective function value over the aggregated top-$k$ item list. The results demonstrate that algorithm {\tt GRD-LM-MIN}  is still highly effective, its objective function value being quite comparable with that of {\tt OPT-LM-MIN}. 

\eat{
\begin{figure*}
\centering
\begin{minipage}[t]{0.24\textwidth}
%\centering
    \includegraphics[width=\textwidth]{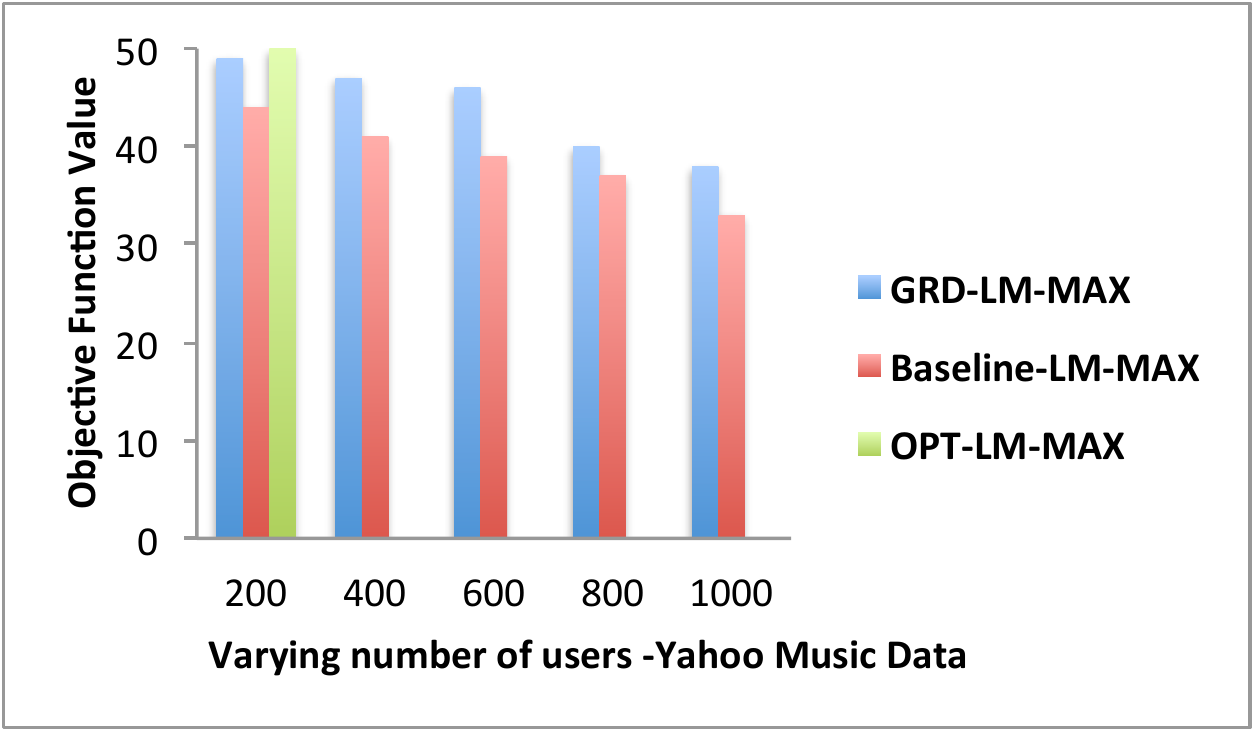}
   %\caption{\small (a)}
    \label{fig:obj-lm-user}
\end{minipage}
%\hspace{5mm}
\begin{minipage}[t]{0.24\textwidth}
%\centering
    \includegraphics[width=\textwidth]{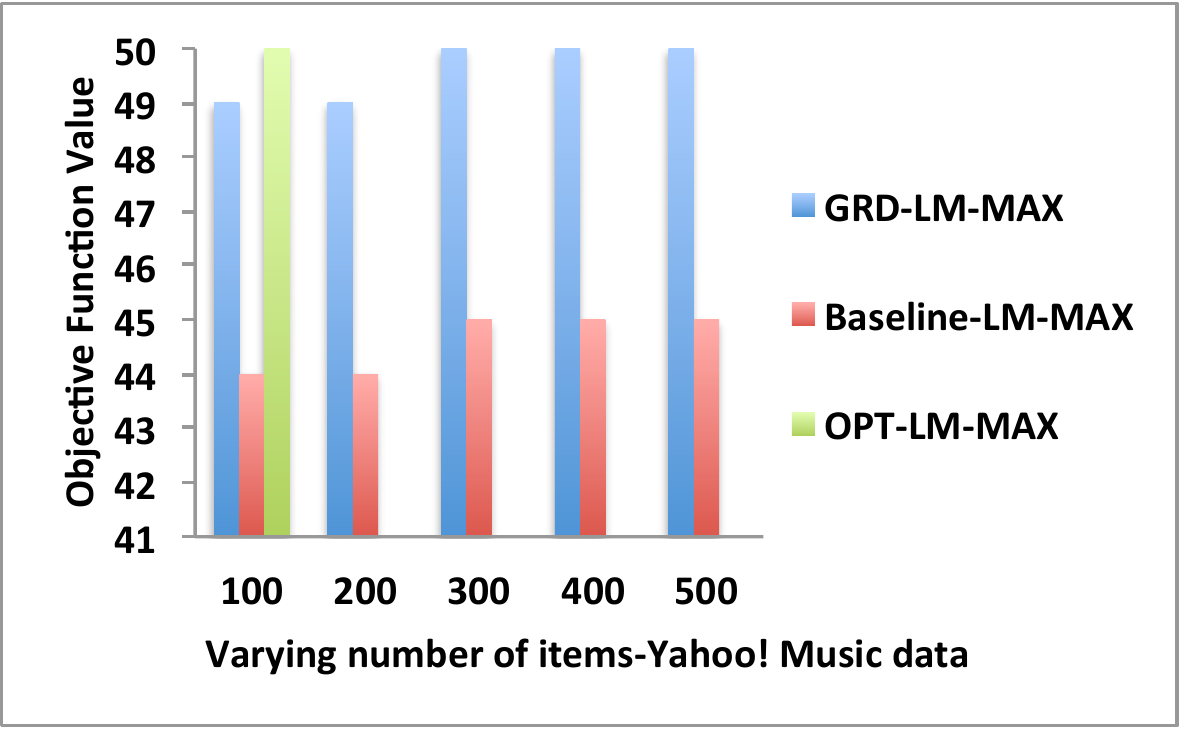}
    %{figures/synthetic/matlab/fig2SimTimeVsThroughput.pdf}
   %\caption{(b)}
    \label{fig:obj-lm-item}
\end{minipage}
%\hspace{5mm}
\begin{minipage}[t]{0.24\textwidth}
\centering
   \includegraphics[width=\textwidth]{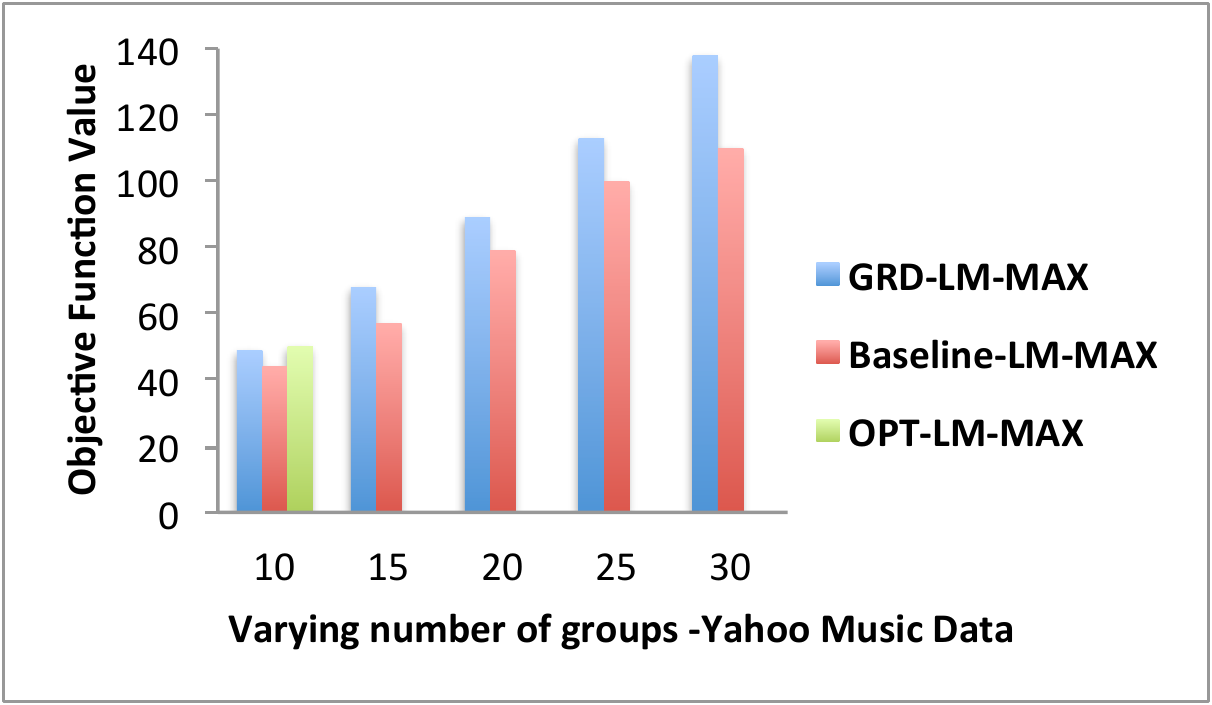}
 %  \caption{(c)}
    \label{fig:obj-lm-grp}
\end{minipage}
%\hspace{5mm}
\begin{minipage}[t]{0.24\textwidth}
%\centering
    \includegraphics[width=\textwidth]{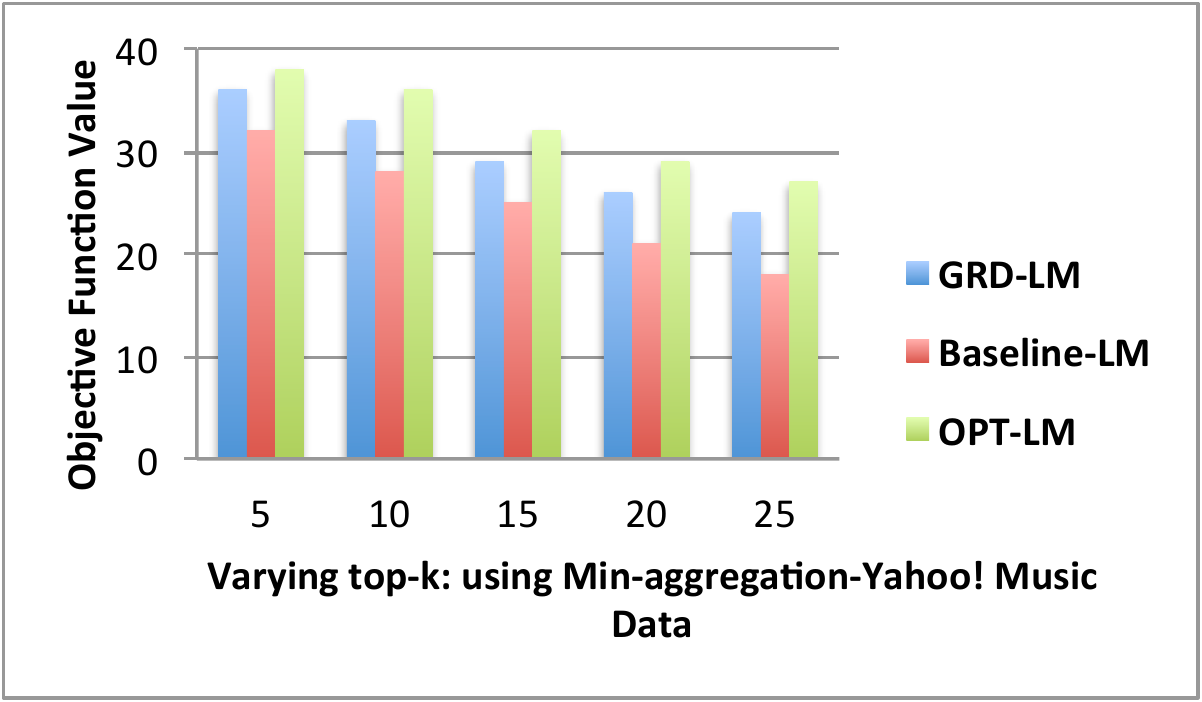}
    %synthetic/fig1WorkloadVsClockTime.pdf
  %  \caption{(d)}
    \label{fig:obj-lm-k}
\end{minipage}
\vspace{-0.1in}
\caption{\small We measure objective function by varying \# users, \# items, and \# groups respectively, one at a time. The default parameters are \# users = $200$, \# items = $100$, \# groups = $10$, $k=5$ and Max-aggregation. }
\end{figure*}}

\subsubsection{Avg Group Satisfaction Over top-$k$ List}
For lack of space, for this set of experiments, we only present the results on MovieLens dataset considering AV semantics. Here, we measure the average user satisfaction over {\em all} the recommended top-$k$ items, i.e., $\frac{\sum_{x=1}^{\ell}\sum_{j=1}^{k} \overline{sc(g_x,i^j)}}{\ell}$, where $\overline{sc(g_x,i^j)}$ is the average AV score on the $j$-th item for group $g_x$ using Min aggregation. While {\tt GRD-AV-MIN}  is not specifically optimized for this measure, our experimental results indicate that the formed groups have very high average satisfaction nevertheless. 
%Thus, it measures how the groups formed by the various algorithms keep the groups happy over all $k$ items recommended to them, on an average.  

{\bf  Number of users:} Figure~\ref{fig:sat-av-user} presents the results where we vary the number of users. Notice that for $10$ user groups, the maximum possible satisfaction per group over the top-$k$ item list could be as high as $25$ when $5$ items are recommended (and the ratings are in the scale of $1-5$). This is indeed true, because, $\sum_{j=1}^{5} \overline{sc(g,i^j)} \leq 25$. Interestingly, {\tt GRD-AV-MIN} consistently produces a score that is close to $25$ and outperforms the baseline algorithm.

{\bf  Number of items:} When number of items is increased, the  group satisfaction score is likely to increase, as the algorithms now have more options to recommend the top-$k$ items from,  for each group. Figure~\ref{fig:sat-av-item} presents these results. The average group satisfaction score of both algorithms improves slightly with more items and again {\tt GRD-AV-MIN} consistently beats the baseline.

{\bf  Number of groups:} The results are shown in Figure~\ref{fig:sat-av-grp}. As expected, we observe that the aggregated group satisfaction over the top-$k$ items improves with the increasing number of groups. As explained before, with increasing number of groups, the algorithms have more flexibility to form groups with users who are highly similar in their top-$k$ item preferences. 

{\bf  Top-$k$ on Min-aggregation:} Finally, we vary top-$k$ and compute the aggregated group satisfaction score over all top-$k$ items (results in Figure~\ref{fig:sat-av-k}). With increasing $k$, the  aggregation is done over more number of items, thus increasing the overall score. As shown in the figure, {\tt GRD-AV-MIN} produces highly comparable results with that of {\tt OPT-AV-MIN} and consistently outperforms the baseline algorithm {\tt Baseline-AV-MIN}.

%\note[Laks]{IMPORTANT: Two MAJOR complaints about the figures. Layout makes it hard to read legend. Consider making 4 figures into one with (a) ... (d). If you do that, you need to change all references to the figures! 

%Second, the legend (or something in the figure should say the values of other parameters. Can't expect reader keep turning pages!} 

\begin{figure*}[ht]
\centering
\subfigure[]{
   \includegraphics[height=2.5cm, width=5.5cm] {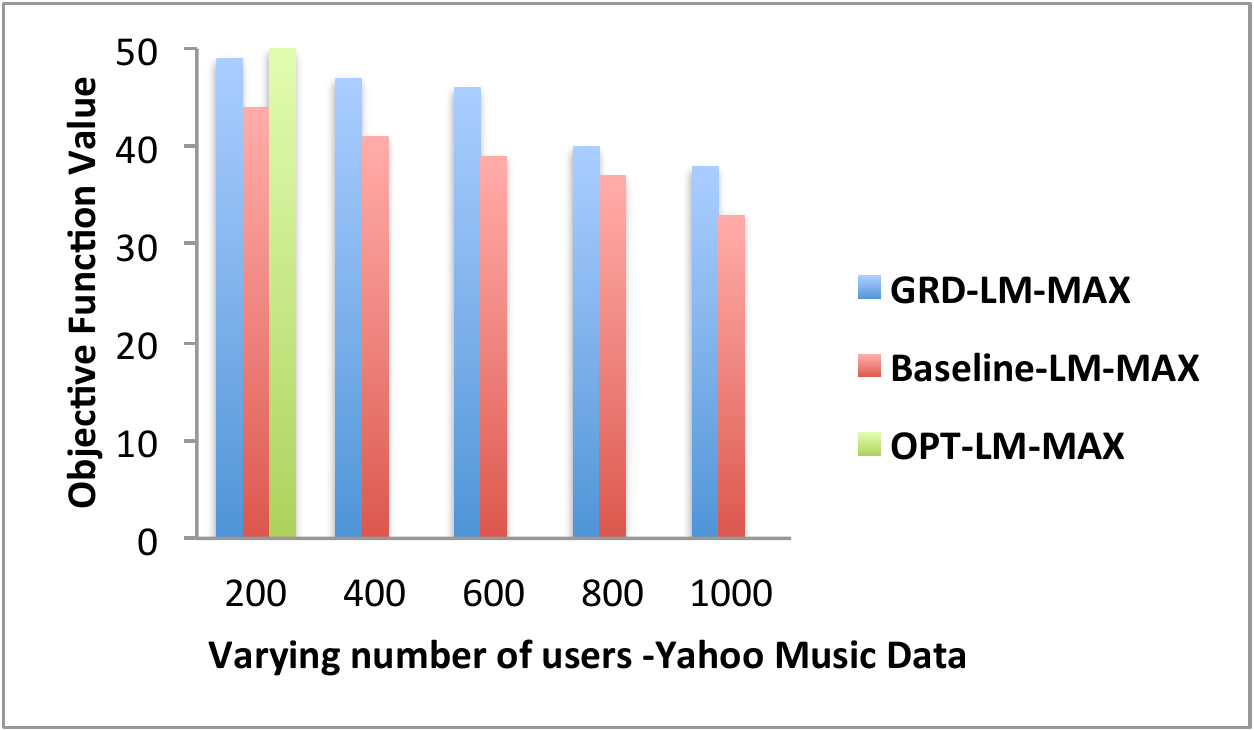}
    \label{fig:obj-lm-user}
 }
 \subfigure[]{
   \includegraphics[height=2.5cm, width=5.5cm] {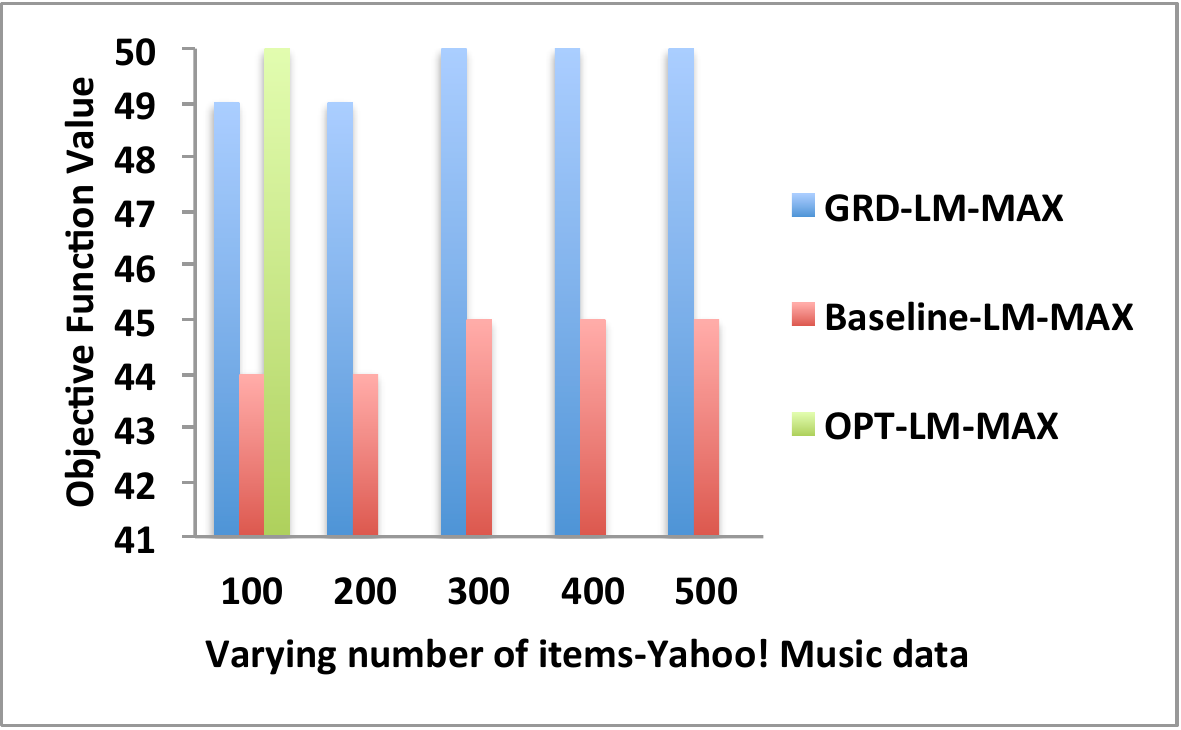}
    %{figures/synthetic/matlab/fig2SimTimeVsThroughput.pdf}
   %\caption{(b)}
    \label{fig:obj-lm-item}
 }
 \subfigure[]{
   \includegraphics[height=2.5cm, width=5.5cm] {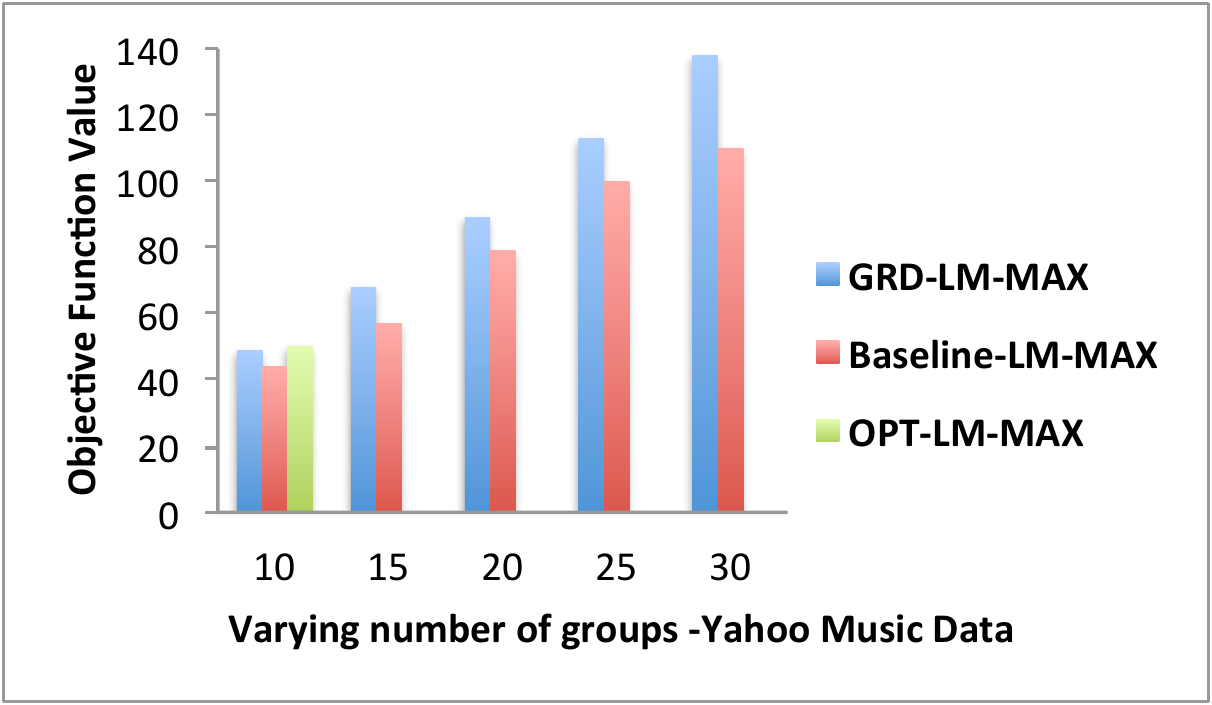}
 %  \caption{(c)}
    \label{fig:obj-lm-grp}
 }
 \label{fig:obj}
\caption{\small We measure the objective function value by varying \# users, \# items, \# groups,  respectively, one at a time. The default parameters are \# users = $200$, \# items = $100$, \# groups = $10$, $k=5$ and Max-aggregation. The underlying dataset is Yahoo! Music.}
\end{figure*}

\begin{figure}[ht]
\centering
\subfigure[]{
   \includegraphics[height=2.5cm, width=4cm] {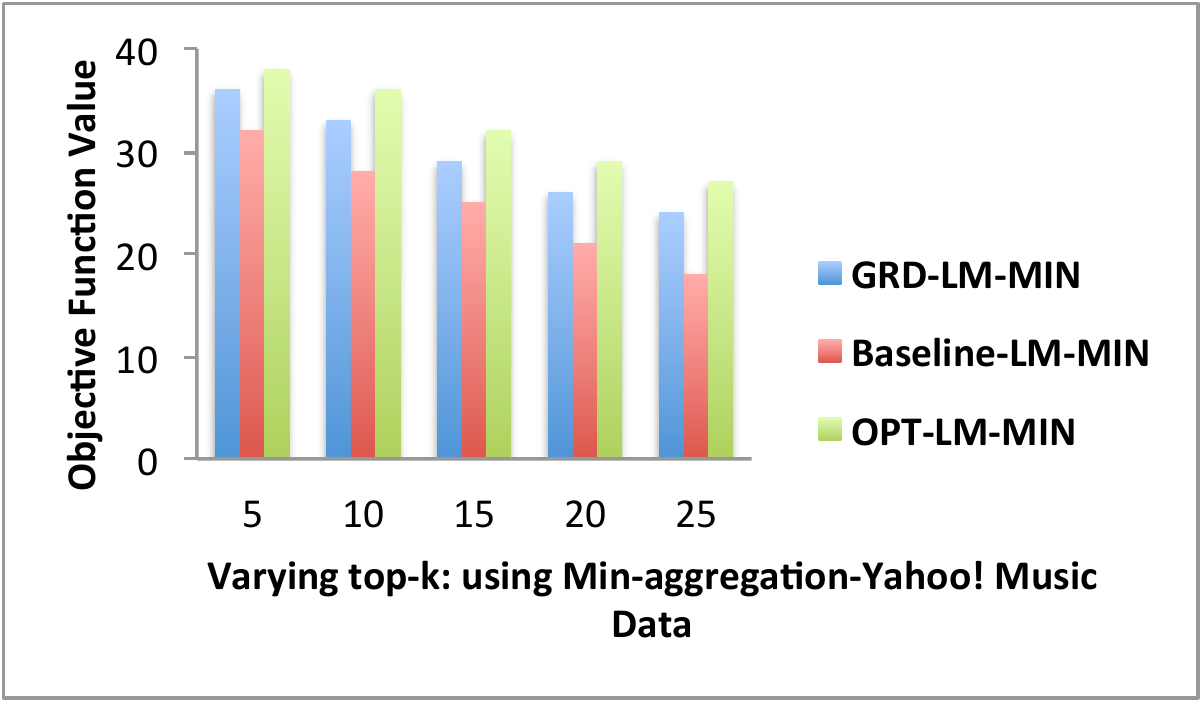}
    %synthetic/fig1WorkloadVsClockTime.pdf
  %  \caption{(d)}
    \label{fig:obj-lm-k}
 }
  \subfigure[]{
   \includegraphics[height=2.5cm, width=4cm] {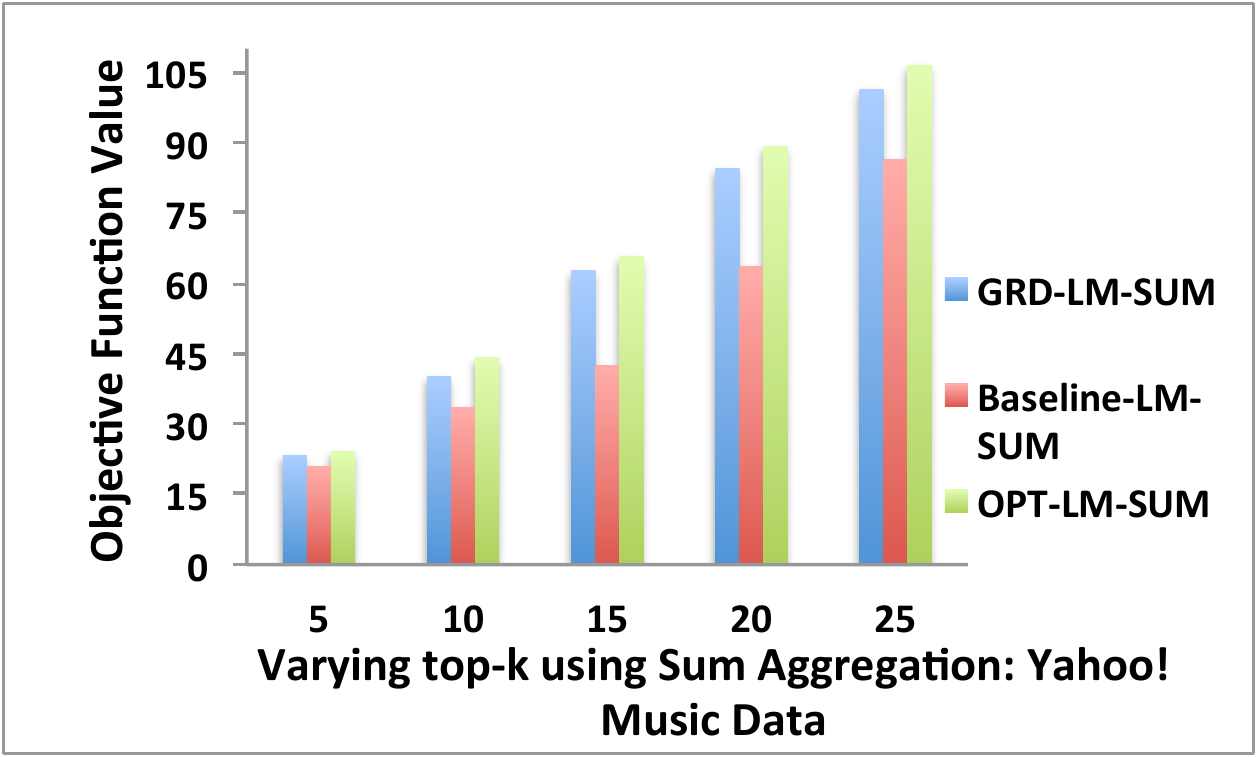}
    %synthetic/fig1WorkloadVsClockTime.pdf
  %  \caption{(d)}
    \label{fig:obj-lm-k-sum}
 }
\label{fig:obj1}
\vspace{-0.1in}
\caption{\small We measure the objective function value by varying top-$k$ for Min and Sum-aggregation using Yahoo! Music. The default parameters are \# users = $200$, \# items = $100$, \# groups = $10$, $k=5$.}
\end{figure}

\begin{figure*}
\centering
\subfigure[]{
   \includegraphics[height=2.5cm, width=4cm]{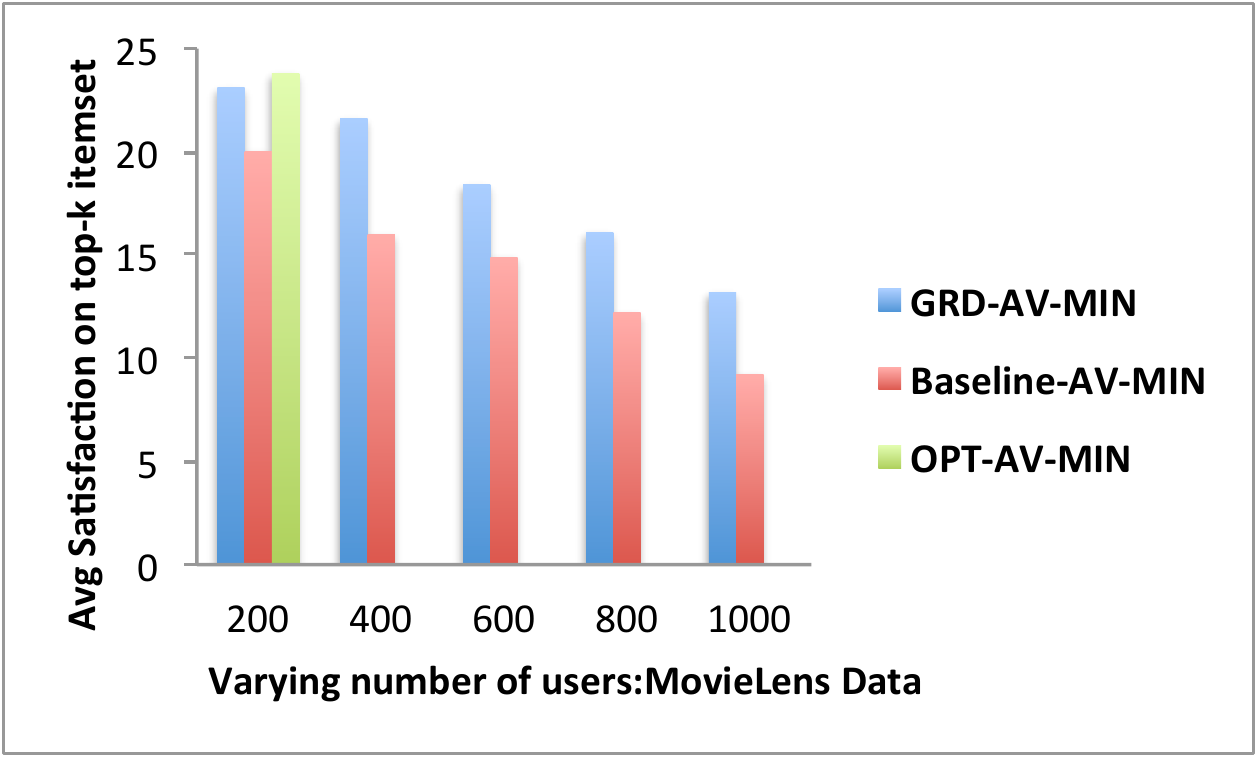}
    %\caption{Average Group Satisfaction on top-$k$ Itemset Varying Number of Users}
    \label{fig:sat-av-user}
    }
\subfigure[]{
   \includegraphics[height=2.5cm, width=4cm]{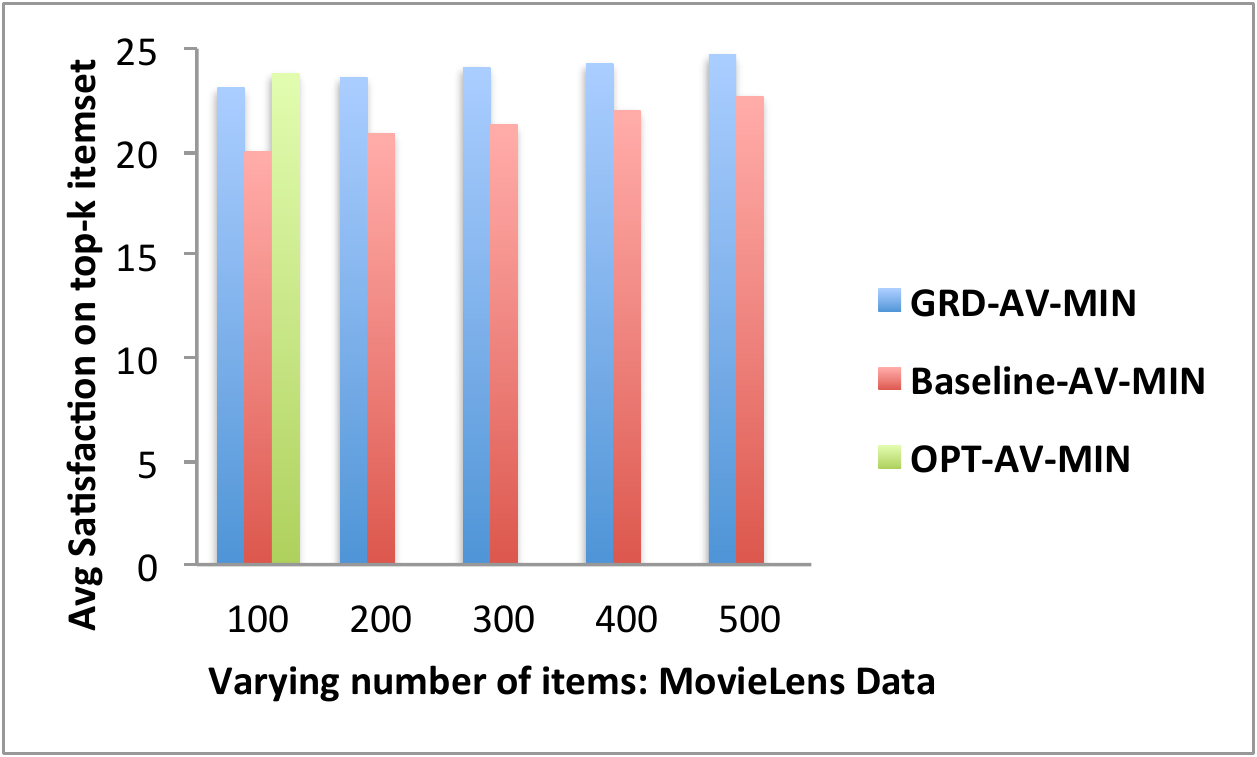}
    %{figures/synthetic/matlab/fig2SimTimeVsThroughput.pdf}
   %\caption{Average Group Satisfaction on top-$k$ Itemset Varying Number of Items}
    \label{fig:sat-av-item}
    }
\subfigure[]{
   \includegraphics[height=2.5cm, width=4cm]{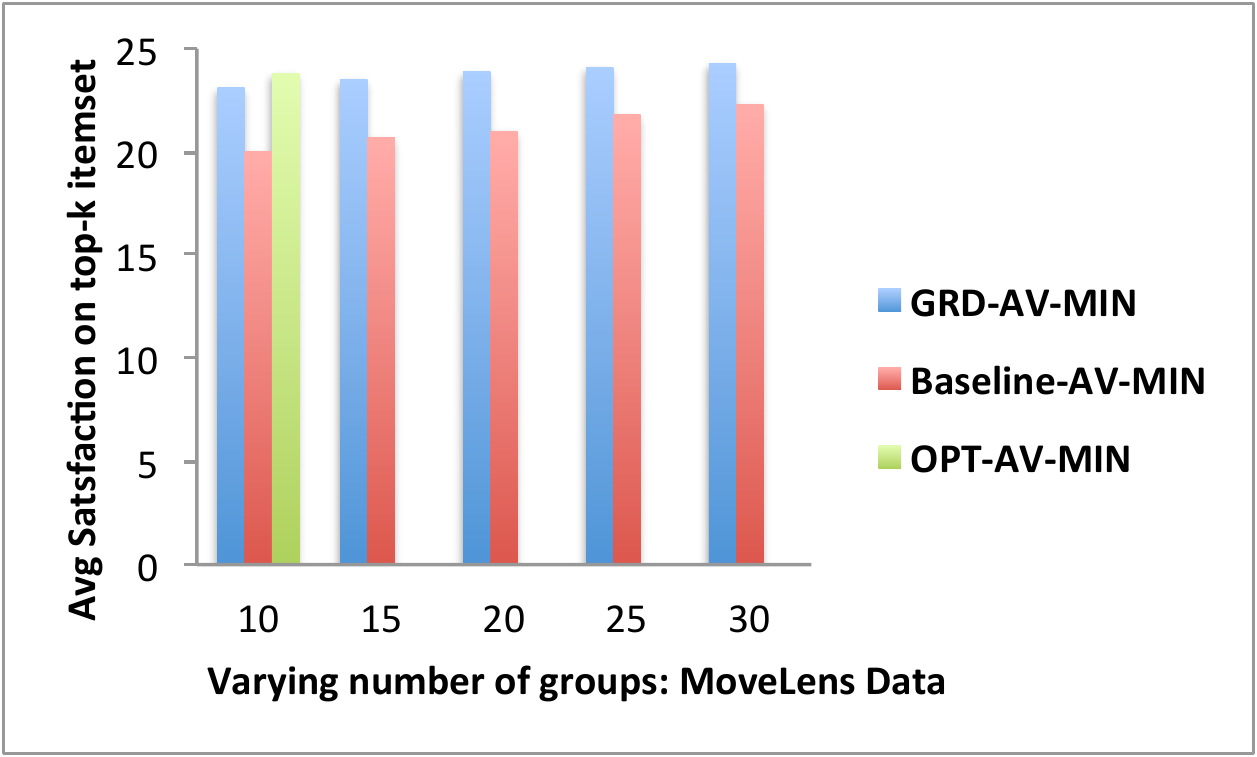}
    %\caption{Average Group Satisfaction on top-$k$ Itemset Varying Number of Groups}
    \label{fig:sat-av-grp}
    }
\subfigure[]{
   \includegraphics[height=2.5cm, width=4cm]{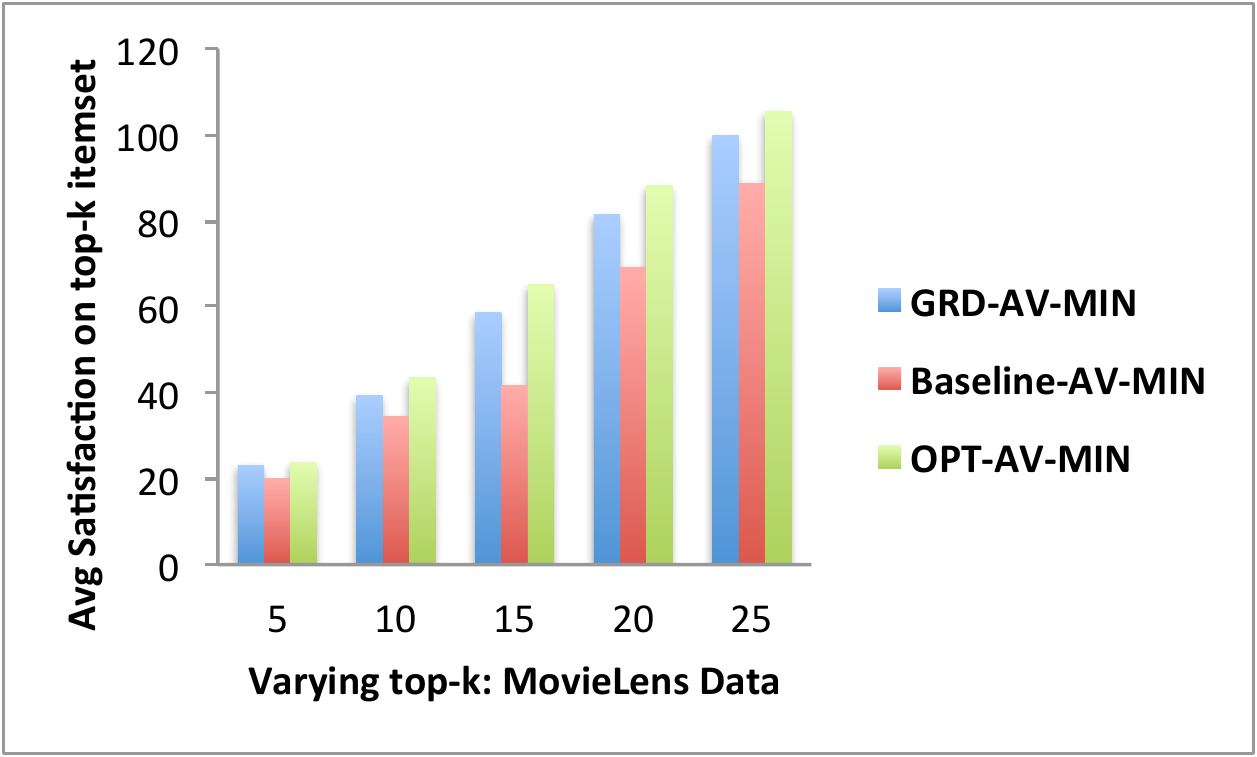}
    %synthetic/fig1WorkloadVsClockTime.pdf
    %\caption{Average Group Satisfaction on top-$k$ Itemset Varying Top-$k$}
    \label{fig:sat-av-k}
    }
    \vspace{-0.1in}
\caption{\small We measure the average group satisfaction score on the top-$k$ item list $\frac{\sum_{x=1}^{\ell}\sum_{j=1}^{k} \overline{sc(g_x,i^j)}}{\ell}$ by varying \# users, \# items, \# groups, and top-$k$, respectively, one at a time. The default parameters are \# users = $200$, \# items = $100$, \# groups = $10$, $k=5$. The underlying dataset is MovieLens.}
\end{figure*}
 %\vspace{-0.4in}

\begin{figure*}
%\centering
\subfigure[]{
   \includegraphics[height=2.5cm, width=5.5cm]{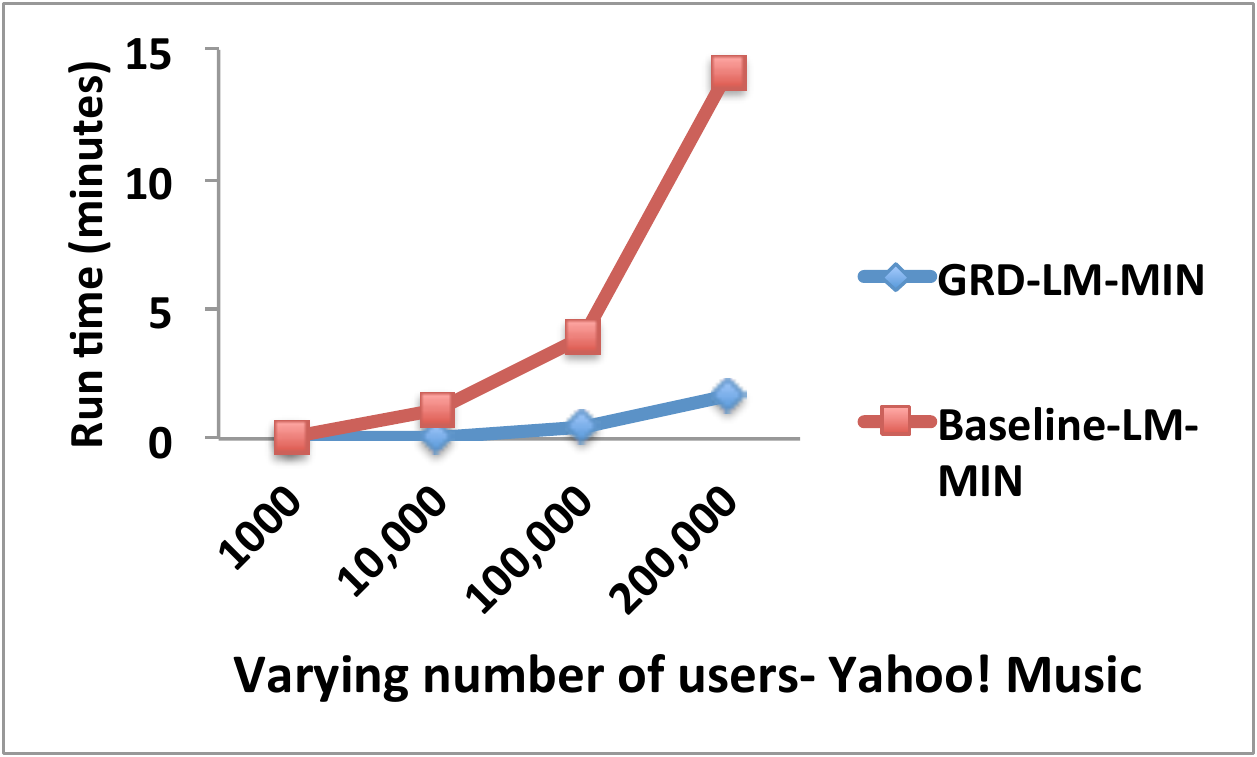}
    %\caption{Run Time in Minutes Varying Number of Users}
    \label{fig:lm-user-scale}
%\end{minipage}
}
%\hspace{5mm}
\subfigure[]{
   \includegraphics[height=2.5cm, width=5.5cm]{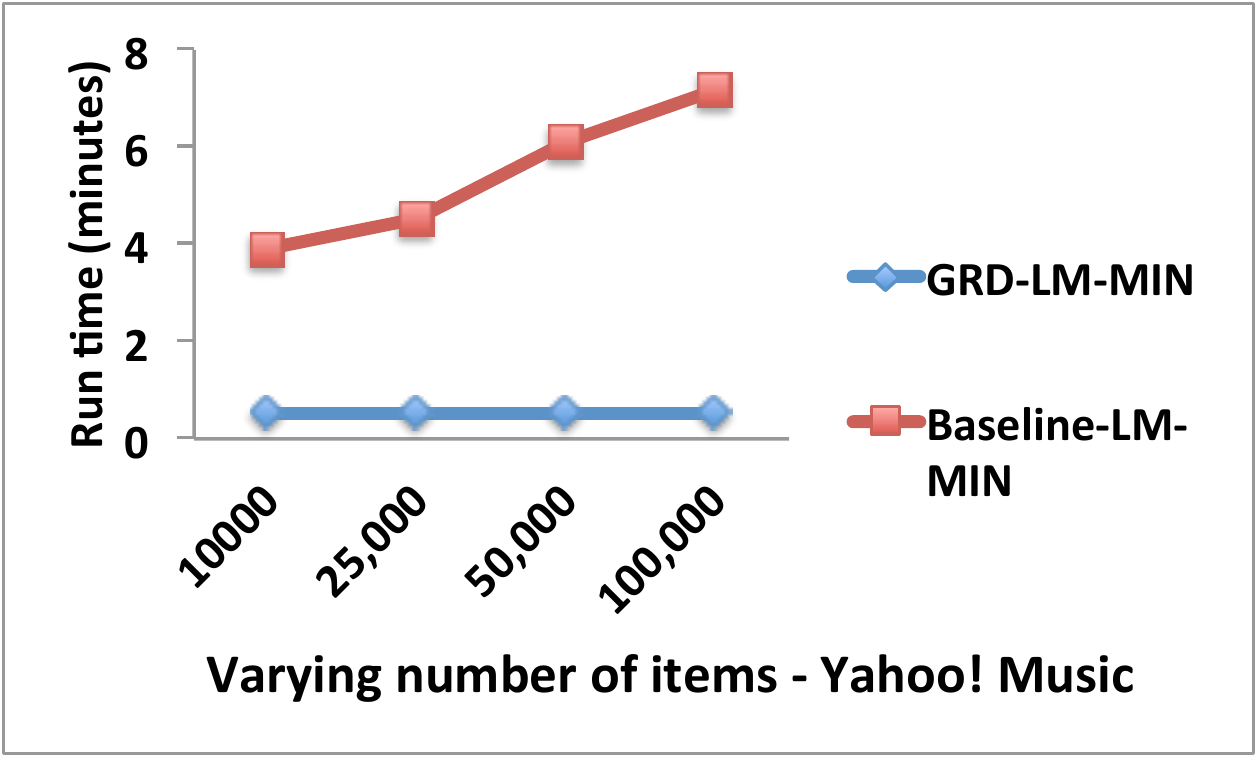}
    %{figures/synthetic/matlab/fig2SimTimeVsThroughput.pdf}
   %\caption{Run Time in Minutes Varying Number of Items}
    \label{fig:lm-item-scale}
%\end{minipage}
}
%\hspace{5mm}
\subfigure[]{
   \includegraphics[height=2.5cm, width=5.5cm]{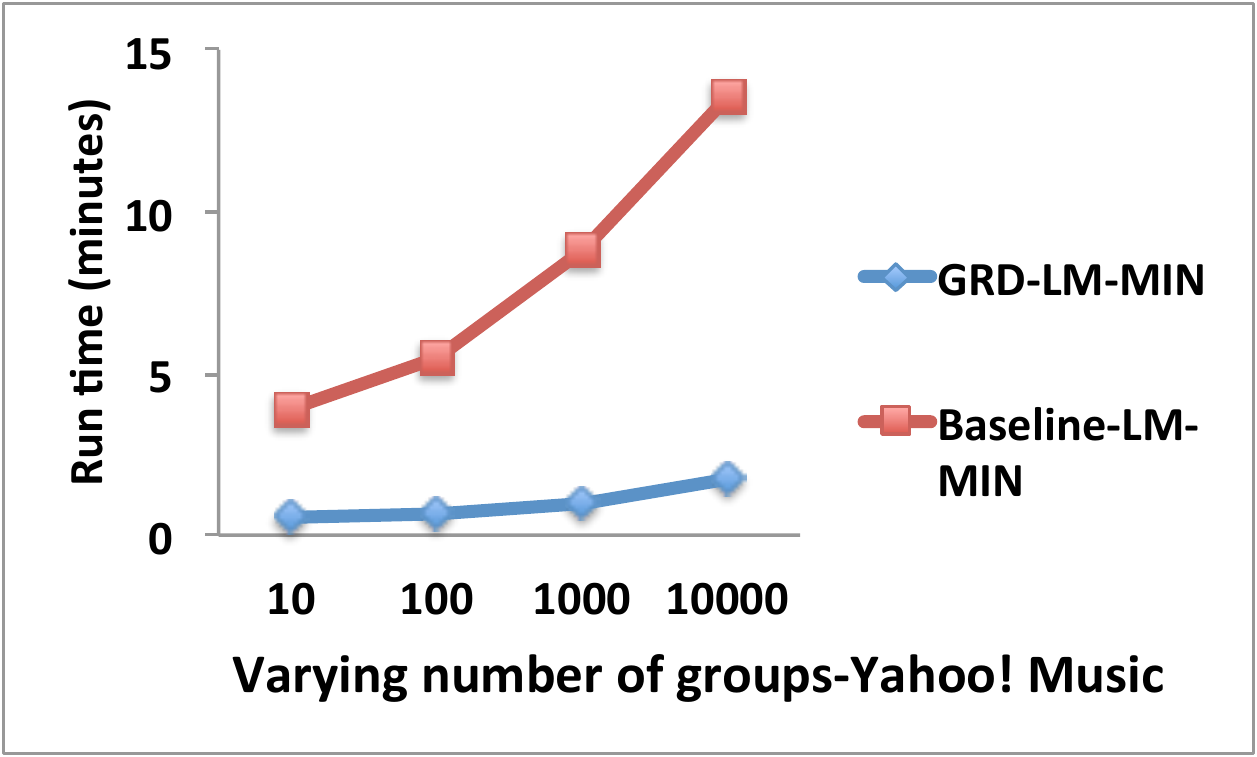}
    %\caption{Run Time in Minutes Varying Number of Groups}
    \label{fig:lm-grp-scale}
%\end{minipage}
}
%\hspace{5mm}
%\subfigure[]{
%   \includegraphics[scale =0.31]{figures/lm-k-scale.pdf}
%    %synthetic/fig1WorkloadVsClockTime.pdf
%    %\caption{Run Time in Minutes Varying Top-$k$}
%    \label{fig:lm-k-scale}
%%\end{minipage}
%}
\vspace{-0.1in}
\caption{\small We present the average running time (measured in minutes) of the group formation algorithms under LM semantics with varying \# users, \# items, \# groups respectively, one at a time. The default parameters are \# users = $100,000$, \# items = $10,000$, \# groups = $10$, $k=5$, considering Min-aggregation. The underlying dataset is Yahoo! Music.}
\end{figure*}
% \vspace{-0.4in}
\begin{figure*}
\subfigure[]{
   \includegraphics[height=2.5cm, width=4cm]{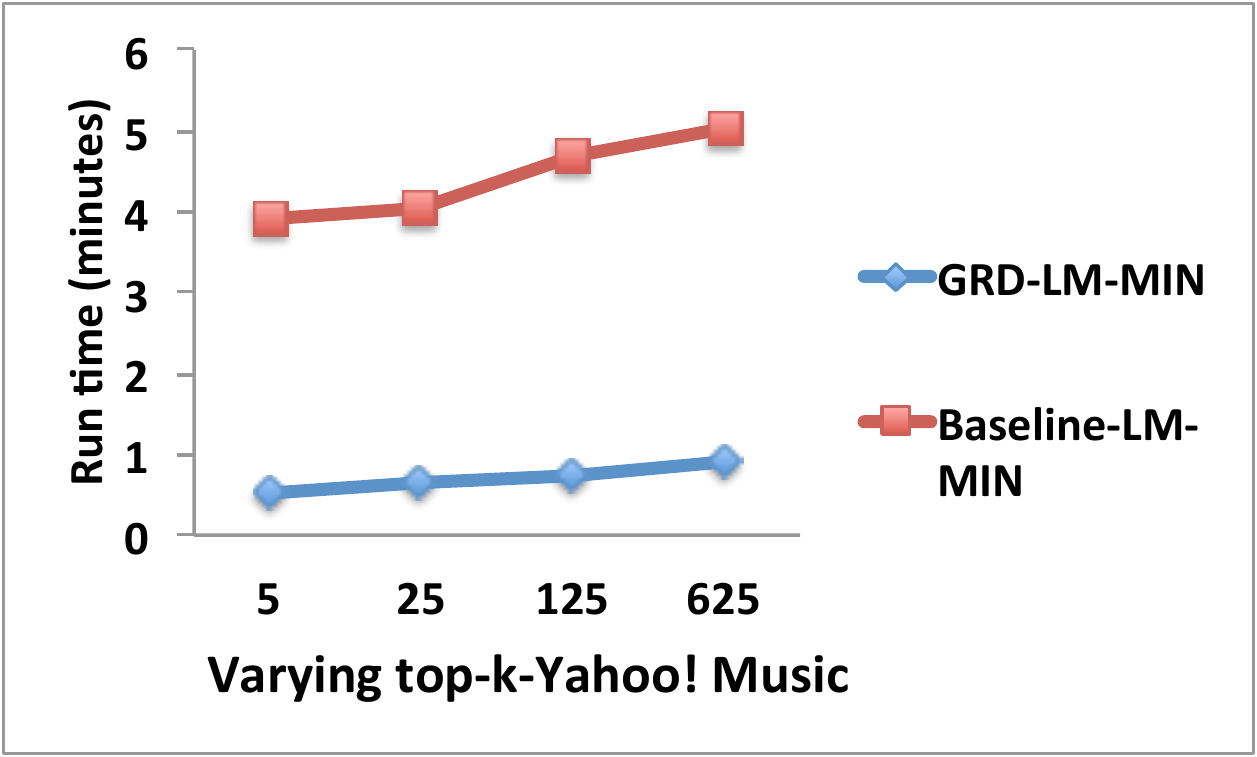}
    %synthetic/fig1WorkloadVsClockTime.pdf
    %\caption{Run Time in Minutes Varying Top-$k$}
    \label{fig:lm-k-scale}
    %\end{minipage}
}
\subfigure[]{
   \includegraphics[height=2.5cm, width=4cm]{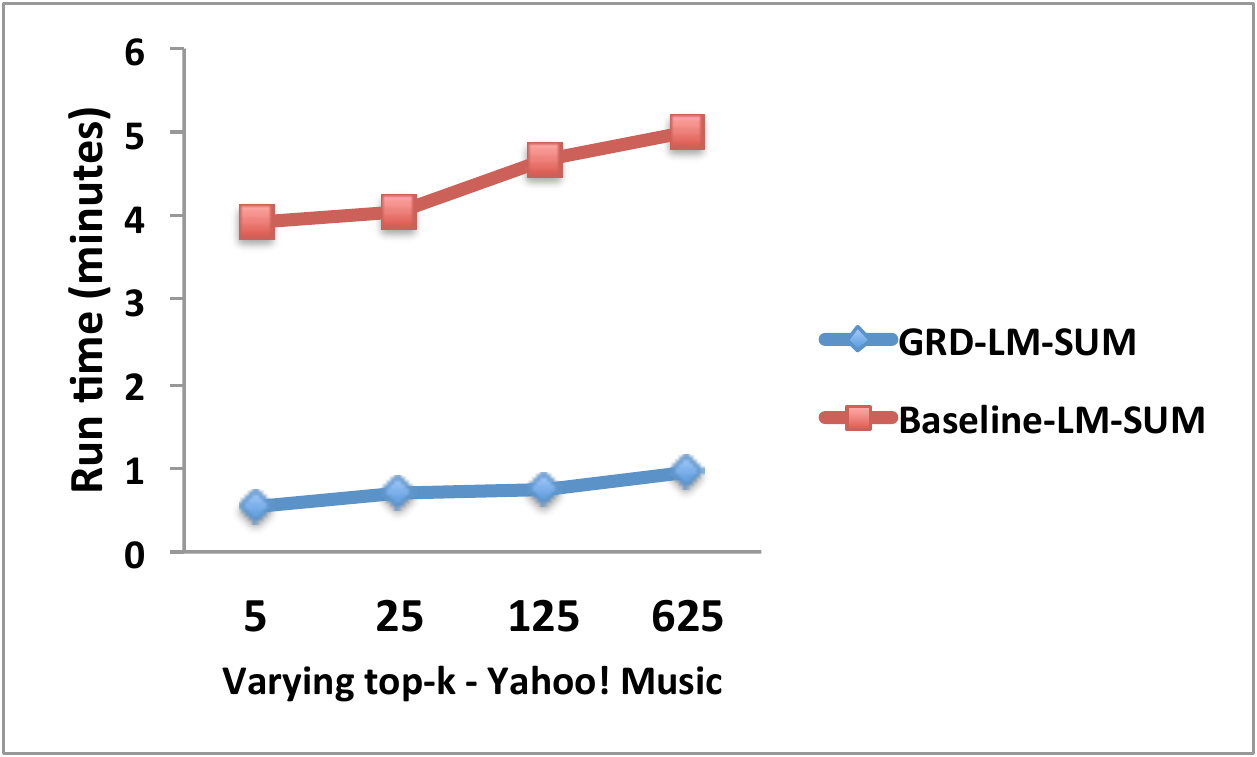}
    %synthetic/fig1WorkloadVsClockTime.pdf
    %\caption{Run Time in Minutes Varying Top-$k$}
    \label{fig:lm-k-scale-sum}
%\end{minipage}
}
\subfigure[]{
   \includegraphics[height=2.5cm, width=4cm]{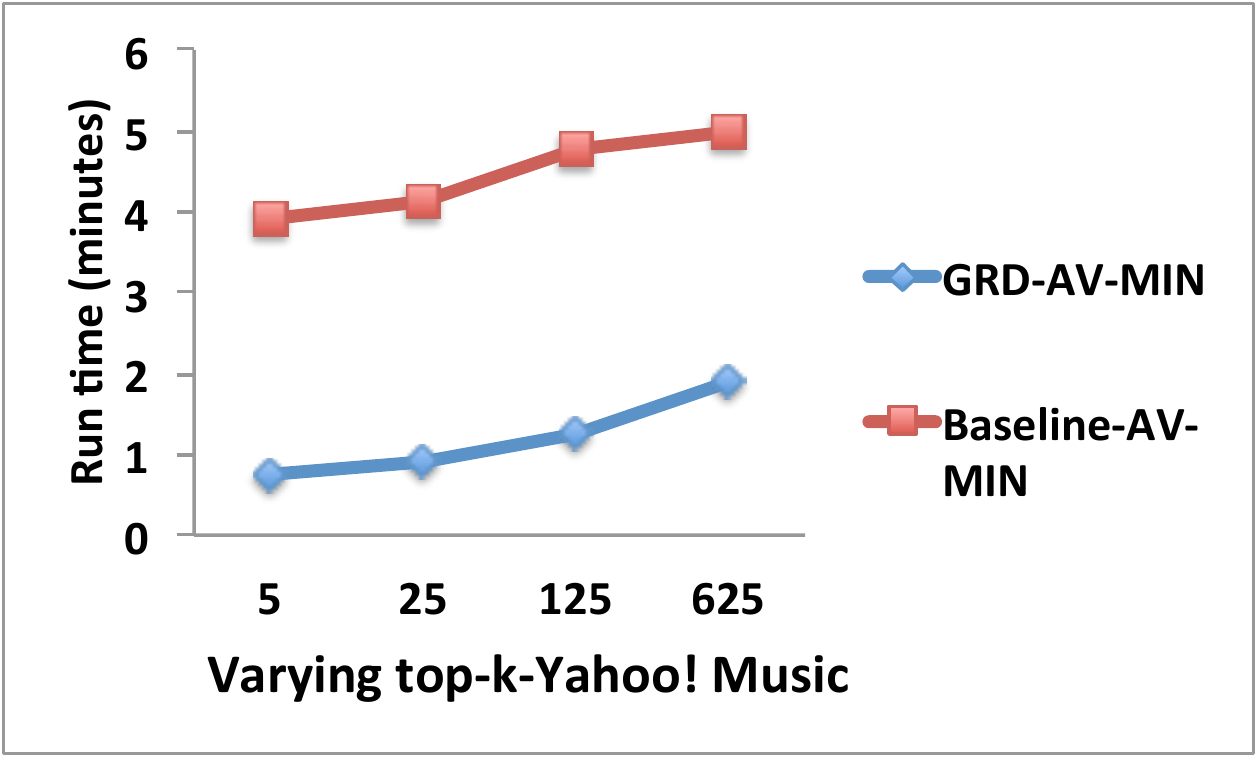}
    %synthetic/fig1WorkloadVsClockTime.pdf
    %\caption{Run Time in Minutes Varying Top-$k$}
    \label{fig:av-k-scale}
%\end{minipage}
}
\subfigure[]{
   \includegraphics[height=2.5cm, width=4cm]{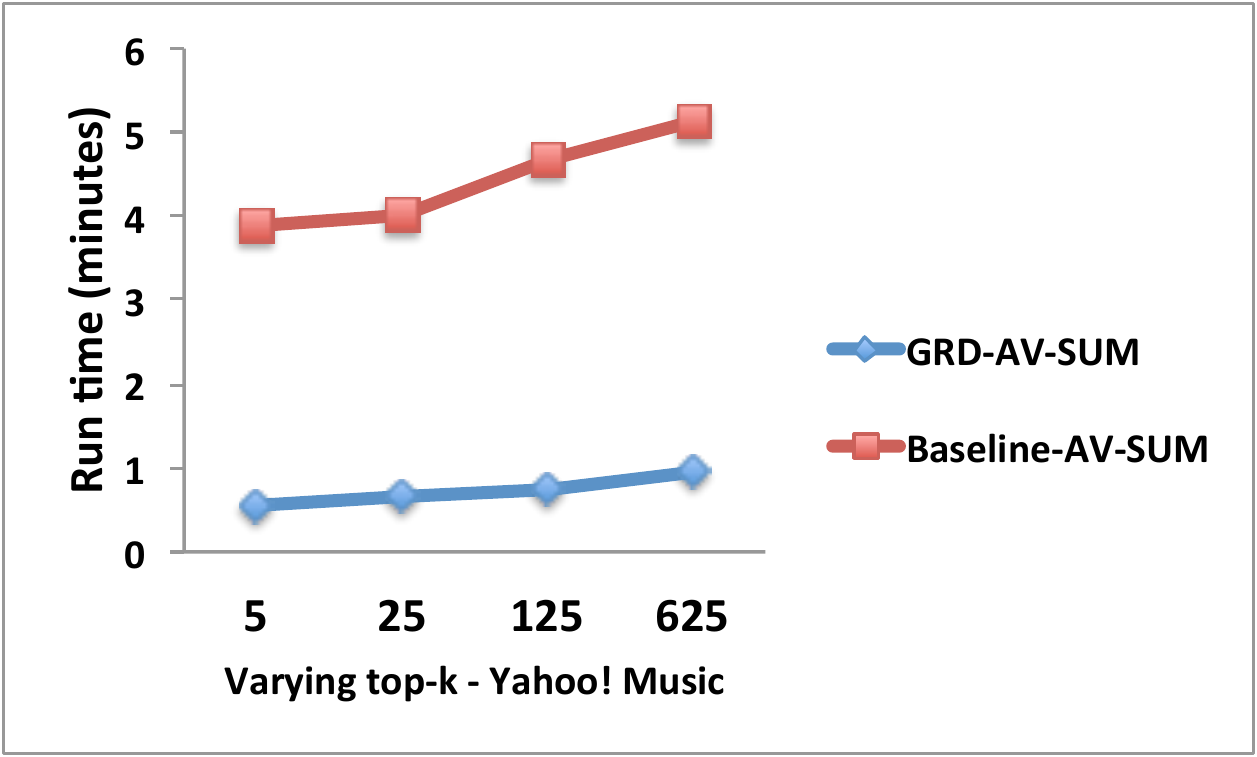}
    %synthetic/fig1WorkloadVsClockTime.pdf
    %\caption{Run Time in Minutes Varying Top-$k$}
    \label{fig:av-k-scale-sum}
%\end{minipage}
}
\vspace{-0.1in}
\caption{\small We present the average running time (measured in minutes) varying top-$k$, using Sum and Min aggregation for LM and AV. The default parameters are \# users = $100,000$, \# items = $10,000$, \# groups = $10$. The underlying dataset is Yahoo! Music.}
\end{figure*}

%\vspace{-0.05in}
\subsubsection{Distribution of Group Sizes}
We randomly select $200$ users and $100$ items with the objective to form $\ell=10$ groups and recommend top-$k$ ($k=5$) items to each group using both datasets. In each sample of $200$ users, we measure the number of users in each of these $10$ groups. We repeat this experiment $3$ times and present the average variation in group size using a $5$-point summary : average minimum size, average 25\% percentile (Q1), average median, average 75\% percentile (Q3), average maximum size. This representation is akin to the box-plot summary~\cite{DBLP:books/mk/HanK2000}. The underlying algorithms are {\tt GRD-LM} and {\tt GRD-AV} considering both Max and Sum-aggregation. These results are summarized in Table~\ref{tbl:dist}. It is evident that the groups that are generated by our algorithms are balanced in general. Unsurprisingly, {\tt GRD-LM-MAX} produces more uniform groups than {\tt GRD-LM-SUM}, as the latter imposes stricter condition on grouping members (needs to match both top-$k$ sequence and ratings). Interestingly, notice that the generated group sizes have smaller average variation under AV than under LM. This is expected, because {\tt GRD-AV} only requires users to have the same top-$k$ item sequence (irrespective of the specific ratings on the bottom item) to belong to the same group. Thus, AV tends to produce relatively larger groups and results in smaller variation in size across the generated groups.

\begin{table}
\centering
\begin{tabular}{ |l|l|l|l| }
\hline
\multicolumn{4}{ |c| }{Distribution of Average Group Size} \\
\hline
Semantics & Quantile  &  {\tt GRD-*-MAX} & {\tt GRD-*-SUM}\\ \hline
\multirow{4}{*}{LM} & {\tt Minimum} & $11.33$ & $8.33$\\
& {\tt Q1} & $15.75$ & $11.5$\\
 & {\tt Median } & $18.5$ & $13.66$\\
 & {\tt Q3} & $23.58$ & $19.33$ \\ 
 & {\tt Maximum} & $31.33$ & $39.33$ \\
 \hline
\multirow{4}{*}{AV} & {\tt Minimum} & $20.33$ & $14.33$\\
& {\tt Q1} & $22.4$ & $19.35$\\
 & {\tt Median } &  $25.4$ & $22.5$\\
 & {\tt Q3} & $28.66$ & $25.95$ \\ 
 & {\tt Maximum} & $30.33$ & $33.75$\\
 \hline
\end{tabular}
\vspace{-0.1in}
\caption{Distribution of Average Group Size\label{tbl:dist}}
\end{table}

%\vspace{-0.15in}

\subsection{Scalability Experiments}\label{scale}
For brevity, we present the results for only the larger dataset, Yahoo! Music, and present a subset of results. As mentioned earlier, {\tt OPT-LM} and {\tt OPT-AV} do not terminate beyond $200$ users, $100$ items, and $10$ groups, and are thus omitted.
%Therefore, in these set of experiments, we measure the run time of {\tt GRD-LM} with {\tt Baseline-LM}, and {\tt GRD-AV} with {\tt Baseline-AV}. 
Our default settings here are as follows: number of users= $100,000$, number of items = $10,000$, number of groups =$10$, $k=5$ and Min-aggregation. Again, we vary \# users, \# items, \# groups, and $k$.

{\bf Interpretation of Results:} The running time of {\tt GRD} is primarily affected by the number of users ($n$), number of groups ($\ell$), and $k$.  Therefore, as it would be seen throughout the results, varying number of items does not impact the computational cost of either {\tt GRD-LM} or {\tt GRD-AV}. Between {\tt GRD-LM} or {\tt GRD-AV}, the latter takes more  time, as it has to aggregate the satisfaction of all the users in a single group to produce the group satisfaction score. Running time of {\tt GRD-LM-MIN} and {\tt GRD-LM-SUM} are observed to be comparable, which corroborates our theoretical analysis.
%Therefore, in general {\tt GRD-AV} takes more time to compute than that of {\tt GRD-LM}. 
For the baseline algorithms ({\tt Baseline}), running time increases with increasing number of users, and number of groups. Additionally, to produce the top-$k$ recommendations once the groups are formed, the last step of these algorithms has to sift through the item-ratings of all users inside every group. Given a group obtained using clustering, the ranked item lists of users may not be aligned. Thus, to form the group's overall top-$k$ list, one may have to consider arbitrarily many items in the individual ranked items lists of the group members. Therefore, the computation time of the baseline algorithms increases with increasing $m$ or $k$. In case of our greedy algorithms, groups are formed by insisting that group members are aligned on the top-$k$ item sequence. Thus, forming the overall top-$k$ list for a group is straightforward in this case, for all groups but the $\ell$-th group formed by the greedy algorithms. For the $\ell$-th group, it sifts through the top-$k$ items per user to generate score.
%\note[Laks]{Pl. chk above!} 

%Since, {\tt Baseline-LM} and  {\tt Baseline-AV} incur similar computational cost, therefore, the difference in runtime between {\tt Baseline-LM} and  {\tt Baseline-AV} is negligible under the same settings. 
\vspace{-0.05in}
%We first present the LM results, followed by the AV results.
\subsubsection{Scalability Experiments : LM}

{\bf  Number of users:} We vary the number of users and measure the clock time of group formation and top-$k$ recommendation time for  {\tt GRD-LM-MIN} and {\tt Baseline-LM-MIN}. (For the record, the optimal algorithms do not complete even after one hour.) Figure~\ref{fig:lm-user-scale} presents the results in minutes. As expected, {\tt GRD-LM-MIN} increases linearly and always terminates within $2$ minutes. These results also exhibit that the clustering based baseline algorithm is non-linear that our algorithm significantly outperforms. 

{\bf  Number of items:} Next, we vary the number of items and measure the clock time. As can be seen from Figure~\ref{fig:lm-item-scale}, the running time of our proposed algorithm is less affected by varying number of items. This result is also consistent with our theoretical analysis. Recall from Section~\ref{sec:analyseslm} that the running time of both the algorithms is $O(nk+\ell\log n)$, which is independent of the number of items $m$. Since {\tt GRD-LM-MIN} leverages the sorted top-$k$ list of items, more items do not necessarily lead to higher computational cost. On the contrary, the clustering based baseline has to produce the top-$k$ itemset for each user group once the groups are formed. As explained earlier, this requires considerable work since top-$k$ lists of cluster members may not be aligned. Figure~\ref{fig:lm-item-scale} clearly demonstrates that {\tt Baseline-LM} is rather sensitive to the increasing number of items. {\tt GRD-LM-MIN} of course beats the baseline. 

{\bf  Number of groups:} These results are presented in Figure~\ref{fig:lm-grp-scale}. When the number of groups is increased, both {\tt Baseline-LM-MIN} and our algorithm take more time. This observation is also consistent with our theoretical analyses, as the running time of both these algorithms depends on the number of groups. However, {\tt GRD-LM-MIN} scales linearly with the increasing number of groups and outperform its baseline counterpart quite consistently.

{\bf  Top-$k$ on Min and Sum aggregation:} We vary top-$k$ for both {\tt GRD-LM-MIN} and {\tt Baseline-LM-MIN} and present the running time in Figure~\ref{fig:lm-k-scale}. While {\tt GRD-LM-MIN} consistently outperforms {\tt Baseline-LM-MIN}, both these algorithms are not very sensitive to increasing $k$. The computation time of the first $\ell-1$ groups are not so much affected by $k$ and only determining the LM score and hence top-$k$ list of the last group (i.e., $\ell$-th group) is affected. The same observation holds for the baseline, as it incurs majority of its computations in forming the clusters that do not depend on $k$. Figure~\ref{fig:lm-k-scale-sum} presents the running time of both {\tt GRD-LM-SUM} and {\tt Baseline-LM-SUM}. {\tt GRD} consistently outperforms {\tt Baseline}, as expected, similar to Min aggregation.

%Only at the very last step, it has to compute the top-$k$ recommended item list per group, where individual member lists may not be aligned.

% it is intuitive to notice that both these algorithms are not very sensitive to increasing $k$ in top-$k$. This makes sense, because {\tt GRD-LM} just has to form the intermediate groups considering a higher top-$k$ set at the beginning. After that, only determining the LM score and hence top-$k$ list of the last group (i.e., $\ell$-th group) is affected by increasing $k$. The same observation holds for the clustering based baseline, as it incurs majority of its computations in forming the clusters that do not depend on $k$. Only at the very last step, it has to compute the top-$k$ recommended item list per group, where individual member lists may not be aligned.  

%\note[Laks]{The story above is a bit confused. Please fix. I tried to, but couldn't figure out why baseline won't be sensitive to k. After all, for every cluster it has to work hard, not just the ``last'' cluster, to form top-k list.} 

\vspace{-0.05in}
\subsubsection{Scalability Experiments : AV}
{\bf  Number of users:}
In this final set of scalability experiments, we again vary number of users and compute the running time of group formation algorithms under  AV semantics. Figure~\ref{fig:av-user-scale} presents the results. These results are similar to those of LM, except that AV takes more time to compute than LM. Then, as expected, the running time of {\tt Baseline-AV} is similar to that of  {\tt Baseline-LM} in Figure~\ref{fig:lm-user-scale}, as the clustering algorithm does not exploit the (AV) semantics in the group formation process. Our proposed greedy algorithm consistently outperforms the baseline algorithm.

{\bf  Number of items:}
We vary the number of items and present the computation time in Figure~\ref{fig:av-item-scale}. {\tt GRD-AV-MIN} takes more time  to terminate compared to that of {\tt GRD-LM-MIN} (Figure~\ref{fig:lm-user-scale}), this slight increase is due to the extra computation that {\tt GRD-AV-MIN} has to perform to aggregate AV score for each group. On the other hand, {\tt GRD-AV-MIN} is not sensitive to the increasing number of items, similarly to {\tt GRD-LM-MIN}. The figure clearly illustrates that {\tt Baseline-AV-MIN} takes more time, as the number of items is increased. As usual, {\tt GRD} outperforms {\tt Baseline}.

{\bf  Number of groups:}
We vary number of groups and observe that both algorithms incur higher processing time with increased number of groups. Figure~\ref{fig:av-grp-scale} presents these results. As expected, {\tt GRD-AV-MIN} scales linearly with the increasing number of groups and consistently outperforms the baseline algorithm.

{\bf  Top-$k$ on Min and Sum aggregation:}
We vary $k$ and measure the computation time in Figures~\ref{fig:av-k-scale} and Figure~\ref{fig:av-k-scale-sum}. With increased values of $k$, running time increases overall. However, {\tt GRD-AV-MIN} takes significantly less time to terminate, compared to its baseline counterpart. Computation times of {\tt Baseline-LM-MIN} and {\tt Baseline-AV-MIN} are similar for the same values of $k$, as these baseline algorithms do not make use of the underlying group recommendation semantics during the group formation process. Sum aggregation results are presented in Figure~\ref{fig:av-k-scale-sum} and the behavior is consistent, as before.

\vspace{-0.05in}
\subsection{User Study}\label{exp:usrstudy}
%Finally, we perform a small scale user study in AMT to validate the effectiveness of our group formation algorithm. Flickr
%data captures user itineraries in the form of photo streams. The
%photos are tagged with corresponding POI names, and the respective
%date/time associated with the photos define the itineraries (such as,
%a set of POIs visited on the same day).

We use publicly available Flickr data to set up the user study in AMT for New York city.  Given a Flickr log of a particular city, each row in that log corresponds to a user itinerary
that is visited in a 12-hour window. From this log, we extract the most popular $10$ POIs. The user study is designed in two phases overall, where Phase 1 is used to create three different sets of users -- similar, dissimilar, and random. Phase 2 is used to assess the performance of the algorithms on each of these sets, under different semantics and aggregation functions. 

{\bf Phase 1: Preference Collection and Group Formation:} First, we set up a HIT (Human Intelligence Task) in AMT, where we ask each AMT user to rate one of these $10$ POIs on a scale of $1-5$, higher rating implying greater  preference. This data is collected from $50$ workers.
From this collected dataset, we create user samples. Sampling is conducted to select a seed user.

{\bf Similar user sample:} We select a subset of $10$ users who have provided very similar ranking on the $10$ POIs. We compute normalized pair-wise similarity, considering each item in the top-$10$ ranked item lists for each user pair and aggregating that over all $10$ items, as follows:
$sim(u,u') = \sum_{j=1}^{10} \frac{sim(u,u',j)}{10}$
\begin{align*}
&  sim(u,u',j) &= 
\begin{cases} 
 1 - \frac{|sc(u,i^j)- sc(u',i^j)|}{5}  \text{ if } i_{j}^u=i_{j}^{u'} \\
0 \text{\qquad otherwise} 
\end{cases}
\end{align*}  
{\bf Dissimilar user sample} We select a different subset of $10$ users who has the smallest aggregate pair-wise similarity.\\
{\bf Random user sample} In our third sample, we select another subset of $10$ users who are chosen randomly from the $50$ users (workers).\\

For brevity, we only report results on LM semantics for these experiments and set the number of groups to be $\ell=3$. We apply {\tt GRD-LM} and {\tt Baseline-LM} (both Sum and Min) to each sample and each algorithm produces three groups.

{\bf Step 2: Group Satisfaction Evaluation}: In this phase, for each user sample (similar, dissimilar, and random), we set up  $6$ HITs in AMT ($3$ for Min and another $3$ for Sum), where each HIT comprises the ($3+3$) groups created by {\tt GRD-LM} and {\tt Baseline-LM}.  In each HIT, we first show the individual user preference ratings for all $10$ users in the sample, over all $10$ items. We do not disclose the underlying group formation algorithm (but refer to them as Method-1 and Method-2) and produce the groups formed  by {\tt GRD-LM} and {\tt Baseline-LM}. We note that our settings mimic the set-up of previous user studies in group recommendation research~\cite{reco1,DBLP:journals/pvldb/Amer-YahiaRCDY09,DBLP:conf/sigmod/RoyICDE14}. We also request the worker to regard herself as one of the individuals in the sample and ask her to rate the following questions (higher is better): (1) Her satisfaction with the formed groups by Method-1; (2) Her satisfaction with the formed groups by Method-2, (3) In an absolute sense, which method she prefers more. Each HIT is undertaken by $10$ unique users, thereby involving $60$ new users in this phase ($30$ for Min and another $30$ for Sum). For each HIT, we average the ratings and present them in Figures~\ref{fig:userstudy1} and~\ref{fig:userstudy2}. Standard error bars are added for statistical significance.

Additionally, we aggregate and compute the percentage of users who prefer  {\tt GRD-LM} versus {\tt Baseline-LM} in Figure~\ref{tbl:userstudy}.

{\bf Interpretation of Results:} We make the following key observations from the results presented in Figures~\ref{tbl:userstudy},~\ref{fig:userstudy1},~\ref{fig:userstudy2}. First and foremost, our proposed  algorithm {\tt GRD-LM} gives rise to higher satisfaction compared to the baseline, in all cases. In fact, this difference in satisfaction seems to be higher when the user population is dissimilar in its individual preferences. During our post-analysis, we see, indeed the clustering based baseline algorithm becomes ineffective, when the individual user preferences are dissimilar from each other. For the same reason, the difference in the average satisfaction of our greedy algorithm from the baseline algorithm is the highest for dissimilar users and smallest for similar users. A random user population consists of both similar and dissimilar users, hence the effectiveness falls in the middle. These results clearly demonstrate that our proposed solutions can effectively exploit existing group recommendation semantics and form groups that lead to high group satisfaction in practice.

\begin{figure*}
\centering
\subfigure[]{
   \includegraphics[height=2.5cm, width=5.5cm]{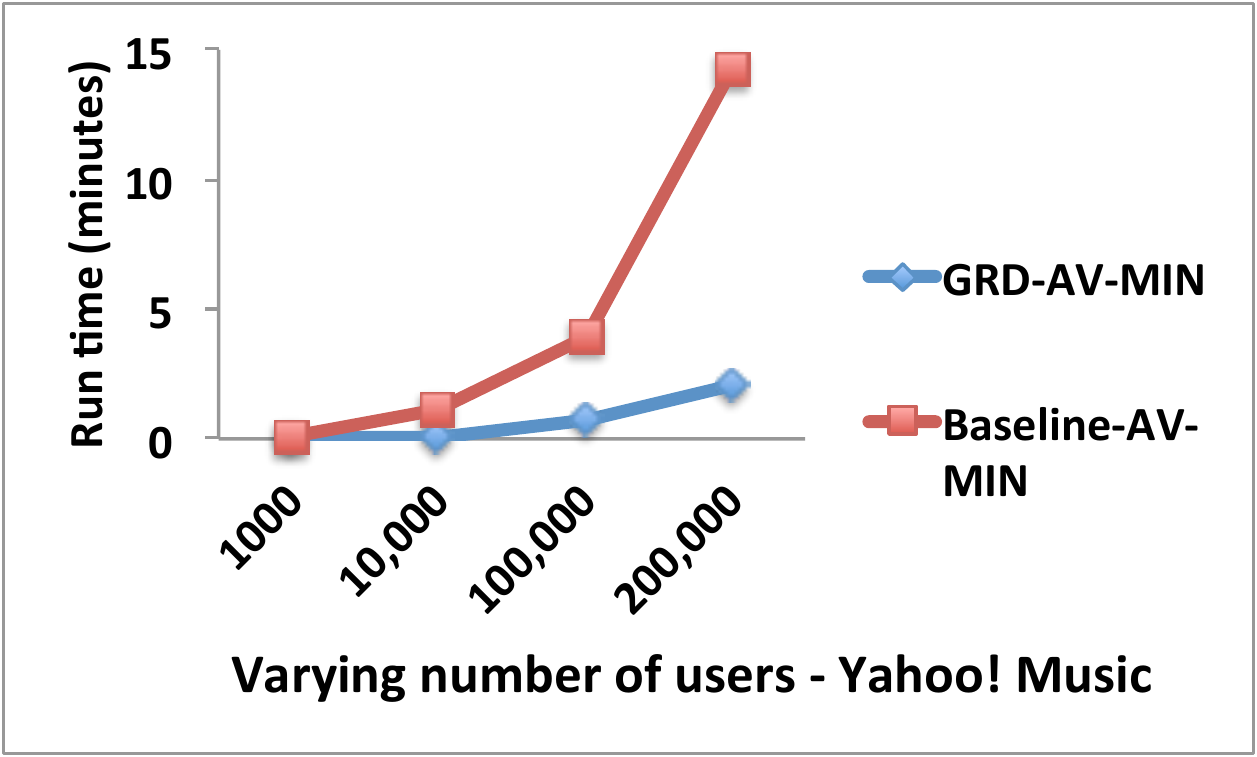}
    %\caption{Run Time in Minutes Varying Number of Users}
    \label{fig:av-user-scale}
%\end{minipage}
}
%\hspace{5mm}
\subfigure[]{
   \includegraphics[height=2.5cm, width=5.5cm]{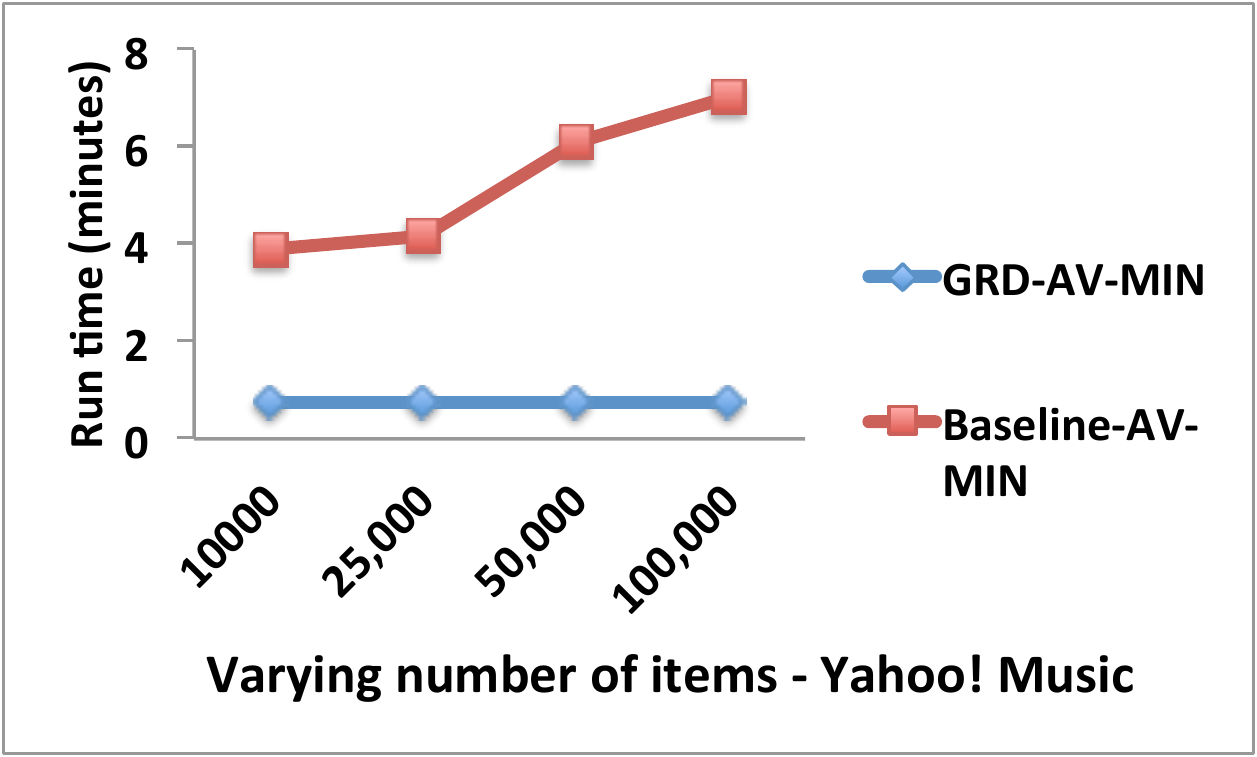}
    %{figures/synthetic/matlab/fig2SimTimeVsThroughput.pdf}
   %\caption{Run Time in Minutes Varying Number of Items}
    \label{fig:av-item-scale}
%\end{minipage}
}
%\hspace{5mm}
\subfigure[]{
   \includegraphics[height=2.5cm, width=5.5cm]{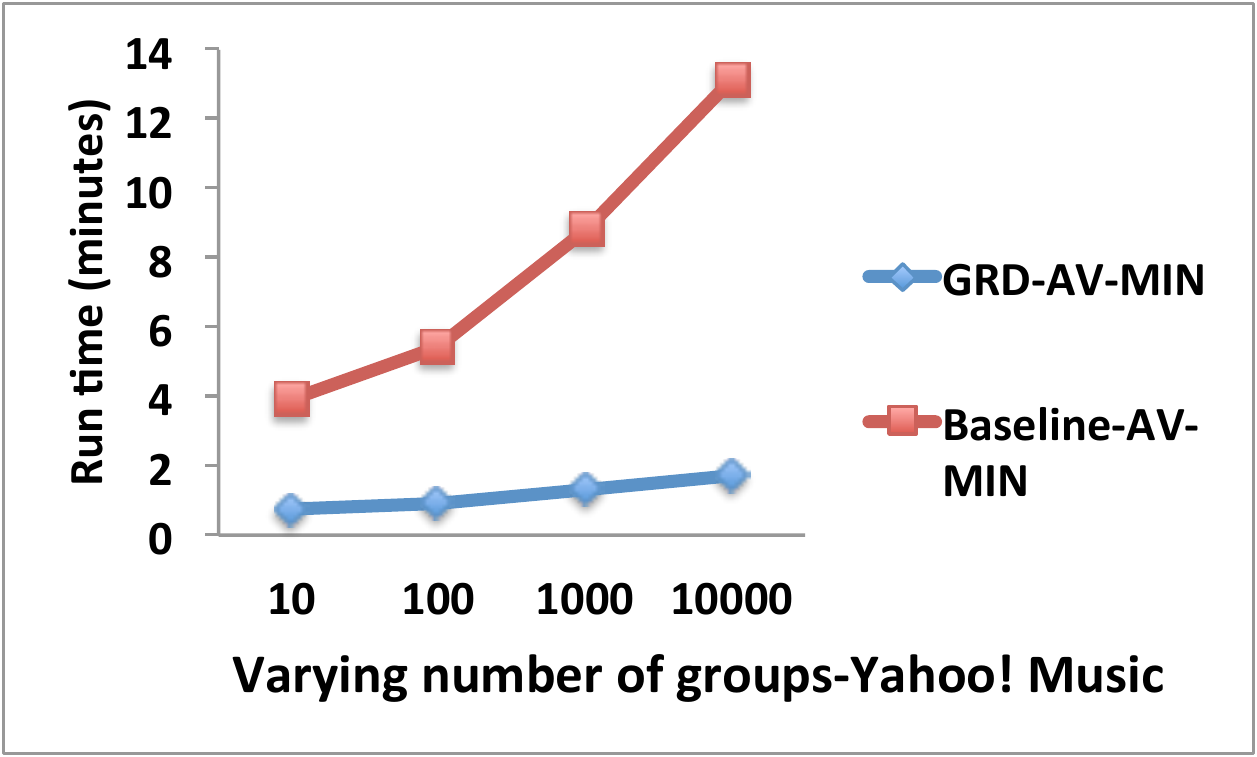}
    %\caption{Run Time in Minutes Varying Number of Groups}
    \label{fig:av-grp-scale}
%\end{minipage}
}
\vspace{-0.1in}
\caption{\small We present the average running time (measured in minutes) of the group formation algorithms under AV semantics  with varying \# users, \# items, \# groups respectively, one at a time. The default parameters are \# users = $100,000$, \# items = $10,000$, \# groups = $10$, $k=5$, considering Min-aggregation. The underlying dataset is Yahoo! Music.}
\end{figure*}
%\vspace{-0.3in}
\begin{figure*}
\centering
\subfigure[]{
   \includegraphics[height=2.5cm, width=5.5cm]{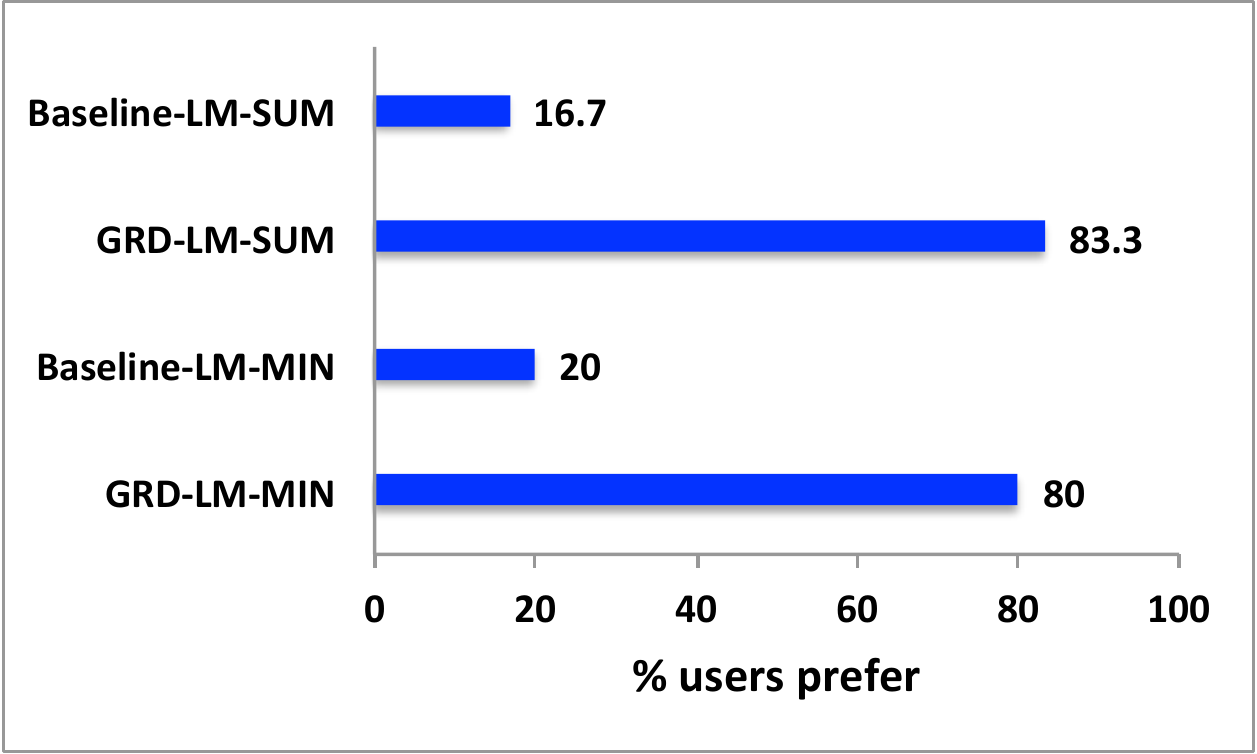}
    %\caption{Run Time in Minutes Varying Number of Users}
    \label{tbl:userstudy}
%\end{minipage}
}
%\hspace{5mm}
\subfigure[]{
   \includegraphics[height=2.5cm, width=5.5cm]{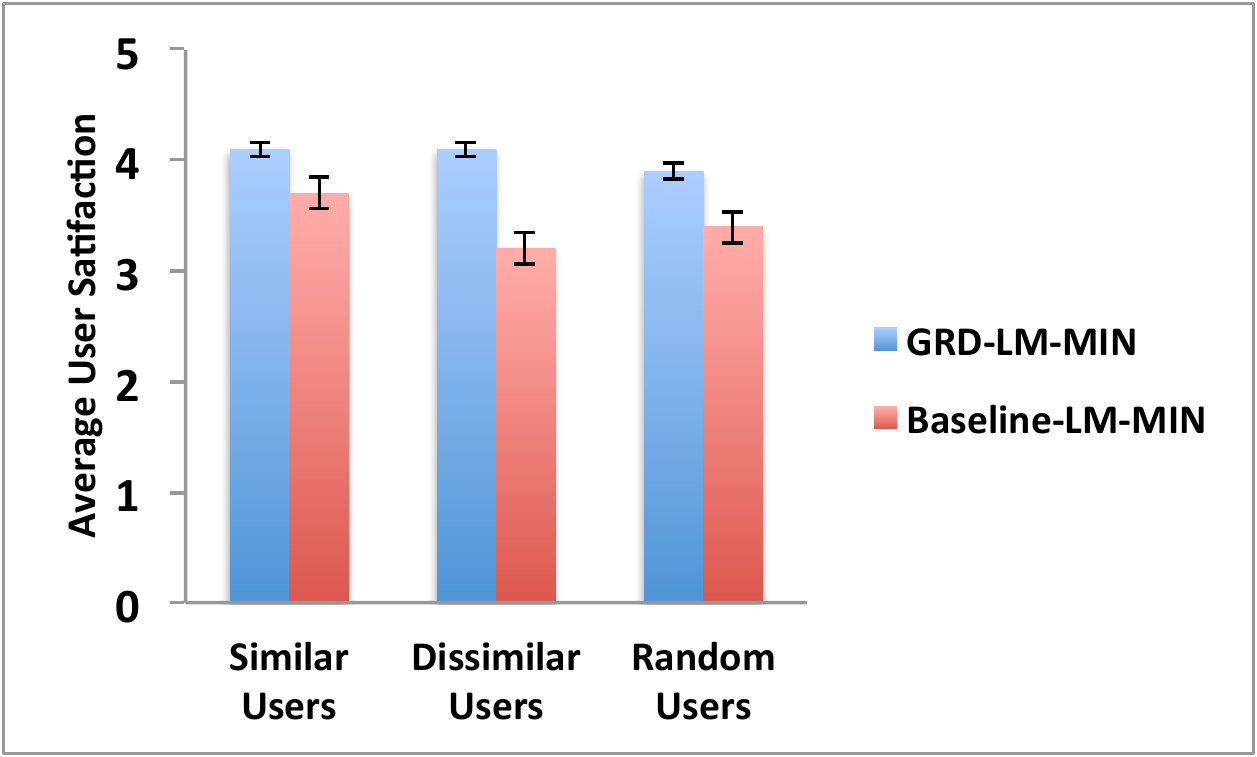}
    %{figures/synthetic/matlab/fig2SimTimeVsThroughput.pdf}
   %\caption{Run Time in Minutes Varying Number of Items}
    \label{fig:userstudy1}
%\end{minipage}
}
\subfigure[]{
   \includegraphics[height=2.5cm, width=5.5cm]{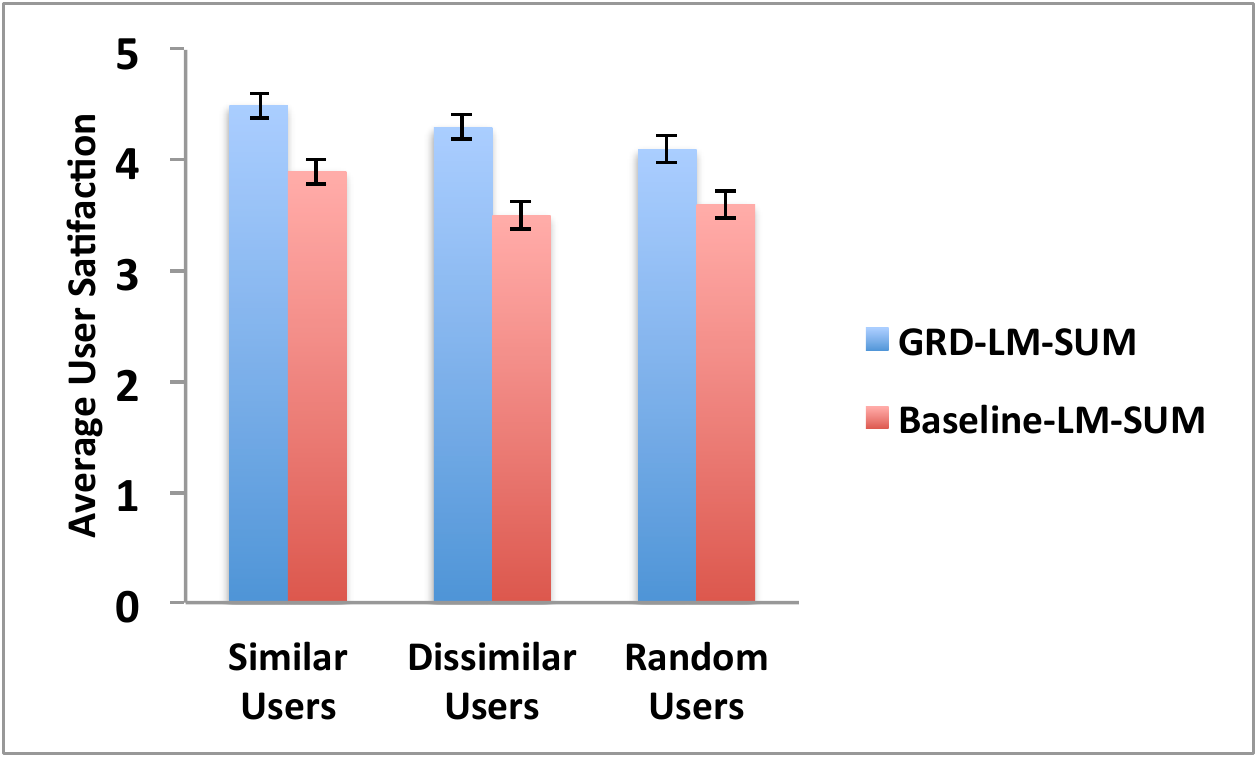}
    %{figures/synthetic/matlab/fig2SimTimeVsThroughput.pdf}
   %\caption{Run Time in Minutes Varying Number of Items}
    \label{fig:userstudy2}
%\end{minipage}
}
\vspace{-0.1in}
\caption{\small User Study Results: Figures~\ref{fig:userstudy1} and~\ref{fig:userstudy2} present the average user satisfaction score of {\tt GRD-LM} and that of {\tt Baseline-LM} for both Min and Sum aggregation with statistical significance analyses. Figure~\ref{tbl:userstudy} indicates about 80\% of the users prefers {\tt GRD} over the {\tt Baseline} algorithms.}
\end{figure*}

\vspace{-0.15in}
\section{Related Work}\label{sec:related}
%In this section, we discuss existing research that are somewhat tangentially related to our studied problem. 
While no prior work has addressed the problem of group formation in the context of recommender systems, we still discuss existing work that appears to be contextually most related. %, even though superficially. 

%Note that, the former  work primarily investigates different group recommendation semantics assuming the groups are given, whereas, the latter research focuses on forming a team of experts to solve a problem together. 

{\bf Group Recommendation:} Group recommendation has been designed for various domains
such as news pages \cite{Pizzutilo:2005:GMP:1366349.1366378}, tourism
\cite{DBLP:journals/eswa/GarciaSO11}, music
\cite{crossen2002flytrap}, book, restaurants~\cite{DBLP:conf/er/NtoutsiSNK12}, and TV programs~\cite{tvrec}. 

There are two dominant strategies for group recommendations
\cite{DBLP:conf/recsys/BerkovskyF10,DBLP:journals/pvldb/Amer-YahiaRCDY09}. The
first approach creates a pseudo-user representing the group and then
makes recommendations to that pseudo-user, while the second strategy
computes a recommendation list for each group member and then combines
them to produce a group's list. For the latter, a widely adopted
approach is to apply an aggregation function to obtain a {\em
  consensus group preference} for a candidate item. Popular aggregation functions, such as, {\em least misery}, {\em aggregate voting} are popularly  used in existing works~\cite{reco1,DBLP:journals/pvldb/Amer-YahiaRCDY09,DBLP:conf/sigmod/RoyICDE14}. ~\cite{DBLP:conf/er/NtoutsiSNK12} pre-clusters the users and the individual recommendations are generated for group members using that member's cluster. After that group aggregation function is applied.

In these works, groups are created beforehand, either by a random set of users with different
interests, or by a number of users  who explicitly choose to be part of a
group.
%, or by clustering users with respect to some
%similarity functions together~\cite{DBLP:conf/er/NtoutsiSNK12}. 

%\note[Laks]{The last one is not accurate: AFAIU, they user clustering for other purposes not for group formation. Groups are given to them.} 

{\bf Market-based Strategies:}  
Existing research on market based strategies~\cite{groupon1,groupon2, groupon3} studies  problems  on daily deals sites, such as Groupon and LivingSocial, that are orthogonal to our group formation problem. These works focus on {\em recommending deals} (i.e., items) to the users and groups %rather than {\em forming groups} that we study. 
For example,~\cite{groupon2} proposes new algorithms for daily deals recommendation based on the explore-then-exploit strategy.~\cite{groupon3} recommends the best deals to the users, among a set of candidate deals, to maximize revenues. Real-time bidding strategy for group-buying deals based on the online optimization of bid values is studied in~\cite{groupon1}.  This body of works essentially relies on price discounts to incentivize group formation around deals, with the objective of maximizing revenue. By contrast, we investigate {\em how to form groups to maximize user satisfaction} under existing group recommendations semantics. Our work is directly deployable as a non-intrusive addition to group recommender systems where explicit incentives may not be present. 
{%\bf Generalized assignment Problems:} Generalized assignment problems are studied in~\cite{gap1,gap2}, where there are resources and tasks

{\bf Team Formation:} Team formation problems ~\cite{lappas2009finding,baykasoglu2007project,zzkarian1999forming,anagnostopoulos2010power,
anagnostopoulos2012online} are often modeled using Integer Programming, or heuristic solutions using Simulated Annealing~\cite{baykasoglu2007project} or Genetic Algorithms~\cite{wi2009team} are designed. These problems are assignment problems.

In general, group formation is not a matching or generalized assignment problem~\cite{gap1,gap2}. There are no resources to match the users to. We need to ``match'' users to one another. In that sense, it's closer in spirit to clustering~\cite{DBLP:books/mk/HanK2000}. However, as we demonstrate in the paper, a clustering algorithm which is agnostic to the group recommendation semantics (LM or AV) is likely to perform poorly for purposes of maximizing group satisfaction.  

{\bf Community Detection:} These problems~\cite{aris} discover communities (a set of users) with common interests. Again, our groups have a more explicit connotation (in the sense of having clearly defined satisfaction scores) than communities. One can potentially generate a graph of users based on a suitable notion of distance or similarity between users in terms of tastes and find communities. Once again, this approach falls short, as it does not take group recommendation semantics into account.

{\bf Multi-way Partition:} Minimization problems over multiway partition functions are studied in ~\cite{part1,part2}, with graphs or hyper-graphs as the underlying abstract model. They attempt to partition the nodes to optimize certain outcome (e.g., variants of $k$-cut problems). While these problems are NP-hard, efficient algorithms with provable approximation factors are known, when the objective function exhibit certain properties~\cite{part1,part2}. If we are to use such a weighted graph, the weight on each edge is local to just two users and does not capture the essence of group recommendation semantics, which renders those solutions to our problem far from ideal. %Therefore, the solutions proposed in the prior work do not extend to our problem.

\vspace{-0.1in}
\section{Conclusion}\label{conc}
\eat{Group recommendations have attracted significant interest in recent years. In this paper, we motivate the related, but complementary problem of group formation. While group recommendation focuses on semantics of group satisfaction and efficient generation of recommendation lists for groups, our group formation problem seeks to find or form groups such that under existing group recommendation semantics, the aggregate satisfaction of the groups is maximized. We motivated this problem with natural applications in various domains. }

We initiate {\em the study of forming groups in the context of group recommender systems}. We consider two popular group recommendation semantics (LM and AV) and formalize the problem of creating a set of non-overlapping  groups over an underlying user population, such that the aggregate satisfaction of the formed groups, with their recommended top-$k$ lists, is maximized. We prove that optimal group formation is computationally intractable under both group recommendation semantics. 
\eat{We provide neat integer programming formulations of the optimal solutions for both problems. Besides their obvious theoretical interest, these formulations are helpful in calibrating approximate or heuristic group formation algorithms.} 
We present efficient greedy group formation algorithms and show that they achieve absolute error guarantees for LM. We present a comprehensive experimental analysis and user studies that demonstrates the effectiveness as well scalability  of our proposed solutions. 
%Last, but not the least, we conduct a user study on Amazon Mechanical Turk. Our results convincingly show that  users by far are more satisfied with the groups formed by our greedy algorithms than the ones formed by natural baselines. 

The approximability of group formation under AV semantics remains open although we conjecture that it may be hard to approximate. Identifying natural special cases that are tractable is an interesting open problem. Forming groups where the individual members are not treated equally, or groups that are possibly overlapping are also worthy of study. 
%in the context of recommender systems 
%group formation problem under other group recommendation semantics  are interesting issues. 
%Finally, the problem we studied in this paper implicitly assumes that user interests are static. In practice, group members' interests could drift and start to diverge over time. How to detect and correct for this in an efficient manner is an important problem.

\newpage
\appendix
%\note[aks]{IMPORTANT: Would it eb better t first give the optimal algorithms, point out they are not practical and then give the ain algorithms section?} 

\section{Optimal Algorithms}\label{sec:optimal}
Despite the fact that the  optimal group formation problem is computationally intractable, we describe optimal algorithms under both LM and AV semantics by formulating them as integer programming problems. We can make use of existing integer programming solvers (such as CPLEX) to solve these problems. Since IP is also NP-hard~\cite{wolsey1998integer} and could be exponential in the worst case, the proposed solutions are not scalable. Nevertheless, the formulation is useful when the numbers of users and items are fairly small. 
%For convenience, the interpretations of the decision variables are described in Table~\ref{tab:notations}.

For both formulations the following Boolean decision variables are defined: $u_{ig}$ captures whether user $u_i$ is part of group $g$. To describe the top-$k$ itemset of each group, an additional Boolean decision variable, $y_{jg}$ checks, if item $j$ is the $k$-th item for group $g$, whereas, $w_{jg}$ is used to check if item $j$ is \emph{one of the top-$(k-1)$ items} for group $g$. The following formulation is provided assuming Min-aggregation for a general value of $k > 1$. To consider Max-aggregation formulation, we no longer need $w_{jp}$ and one has to check if $y_{jg}$ is indeed the top-$1$ item for the recommended itemset. Similarly, Sum aggregation could be performed by modifying the objective function to aggregate over all $k$-items.

\eat{
\begin{table}
\centering
\caption{Notations and Interpretations}
\begin{tabular} {|c|p{5cm}|}
\hline
  {\bf Notation} & {\bf Interpretation} \\
  $y_{jg}$ & a decision variable to denote if item $j$ is the $k$-th item of group $g$ \\
  \hline
  $w_{jp}$ & a decision variable to denote if item $j$ is in the $(k-1)$-th itemset of group $g$ \\
  \hline
  $u_{ig}$ & a decision variable to denote if user  $u_i$ is part of group $g$ \\
  \hline
%  $z_{ip}$ & a decision variable which is same as $w_{ip}$\\
%\hline
\end{tabular}\label{tab:notations}
\end{table}}
%\vspace{-0.2in}
\subsection{IP Formulation for Least Misery}
The objective function for LM aims to form $\ell$ groups such that their sum of scores  is maximized. The  first two constraints capture the fact that the score of an item $j$ is to be computed by considering the minimum score of that item over all users. The third and fourth constraints are used to capture the top-$(k-1)$ items. The rest of the constraints simply state that the $k-$th item is a single item for every group, whereas, there should be a total of $k-1$ additional items whose score is higher than that of the $k$-th item. Finally, we assert that only $\ell$ groups are to be formed, whereas, a user can belong to only one of these groups (satisfying disjointness).
\begin{equation}
\text{Maximize } \Sigma_{g=1}^{\ell} \{ y_{jp} \times sc(g,j) \}
\end{equation}
%y_{jp} \times rel(j,p) \leq w_{ip} \times rel(i,p) \\
s.t. 
\vspace*{-0.2in}
\begin{align*}
\left.\begin{aligned}
sc(j,g) = r \\
r \leq {{\{\forall u_{ig} = 1\}}} u_{ig} \times sc(u_i,j) \\
w_{ig} \times sc(g,i)  \geq  y_{jg} \times sc(g,j) \times w_{ig} \\
(1 - w_{jg}) \times sc(g,j) \leq y_{jg} \times  sc(g,j) \\
\end{aligned} \right\} \\
\left.\begin{aligned}
\Sigma_{j=1}^{m} y_{jg} = 1  \\
\Sigma_{i=1}^{m} w_{ig} = k-1 \\
 \ell \leq n \\
\end{aligned}
 \right\} \\
\left.\begin{aligned}
y_{jg} =1/0 , \forall j=1,2,\ldots,m \\
w_{jg} =1/0 , \forall j=1,2,\ldots,m \\
\Sigma_{g=1}^{\ell} u_{ig} =1, \forall i=1,2,\ldots,n \\
u_{ig} =1/0, \forall i=1,2,\ldots,n, \forall g=1,2,\ldots,\ell  \\
\end{aligned}
 \right\}
%\qquad \text{boolean decision variables}
\end{align*}

When this IP is run on Example~\ref{ex1}, considering $k=1$, the following groups are produced:  $\{u_1,u_3,u_4\}$, $\{u_2, u_6\}$, $\{u_5\}$ with an overall $Obj$ value of $4 + 5 + 3 = 12$. 

%\vspace*{-1ex}
\subsection{IP Formulation for Aggregate Voting}
The formulation of optimal group formation under aggregate voting is similar to that of LM, except for the fact that the score of an item $j$ for a group $g$ is the summation of scores of $j$ over all users in $g$.

\begin{equation}
\text{Maximize } \Sigma_{p=1}^{\ell} \{ y_{jg} \times sc(g,j) \}
\end{equation}
%y_{jp} \times rel(j,p) \leq w_{ip} \times rel(i,p) \\
s.t. 
\begin{align*}
\left.\begin{aligned}
sc(g,j) = \Sigma_{{\{\forall u_{ig} = 1\}}} u_{ig} \times sc(u_i,j) \\
w_{ig} \times sc(g,i)  \geq  y_{jg} \times sc(g,j) \times w_{ig} \\
(1 - w_{jg}) \times sc(g,j) \leq y_{jp} \times  sc(g,j) \\
\end{aligned} \right\} \\
\left.\begin{aligned}
\Sigma_{j=1}^{m} y_{jg} = 1  \\
\Sigma_{i=1}^{m} w_{ig} = k-1 \\
 \ell \leq n \\
\end{aligned}
 \right\} \\
\left.\begin{aligned}
y_{jg} =1/0 , \forall j=1,2,\ldots,m \\
w_{jg} =1/0 , \forall j=1,2,\ldots,m \\
\Sigma_{g=1}^{\ell} u_{ig} =1, \forall i=1,2,\ldots,n \\
u_{ig} =1/0, \forall i=1,2,\ldots,n, \forall g=1,2,\ldots,\ell  \\
\end{aligned}
 \right\}
%\qquad \text{boolean decision variables}
\end{align*}

When run on Example~\ref{ex2}, the optimal grouping with two groups consists of the following groups:  $\{u_1,u_3,u_4\}$, and $\{u_2,u_5,u_6\}$, with the overall objective function value of $14$.
\eat{
\begin{equation}
\text{Maximize } \Sigma_{p=1}^{x} \{ y_{jp} \times rel(j,p) \}
\end{equation}
%y_{jp} \times rel(j,p) \leq w_{ip} \times rel(i,p) \\
s.t. 
\begin{align*}
w_{ip} \times rel(i,p)  \geq  y_{jp} \times rel(j,p) \times z_{ip} \\
w_{ip} = z_{ip} \\
\Sigma_{j=1}^{m} y_{jp} = 1  \\
\Sigma_{j=1}^{m} w_{jp} = k-1 \\
rel(j,p) = Minimum_{\{\forall u_{ip} = 1\}} u_{ip} \times r(i,j) \\
x \leq n \\
\text{top-k-value}(p) \neq \text{top-k-value}(p'), \forall p \neq p' \\
\text{top-k-value}(p) = \Sigma_{j=1}^{m} \{ w_{jp} \times 2^j \} + \Sigma_{j=1}^{m} \{ y_{jp} \times 2^j \} \\
(1 - w_{jp}) \times rel(j,p) \leq y_{jp} \times  rel(j,p) \\
y_{jp} =1/0 , \forall j=1,2,\ldots,m \\
z_{ip} =1/0 , \forall i=1,2,\ldots,m \\
w_{jp} =1/0 , \forall j=1,2,\ldots,m \\
\Sigma_{p=1}^{x} u_{ip} =1, \forall i=1,2,\ldots,n \\
u_{ip} =1/0, \forall i=1,2,\ldots,n, \forall p=1,2,\ldots,x  \\
\end{align*}}

\section{Suboptimal Group Formation by GRD-LM-SUM}\label{addex}
%Imagine that the user set \\
%$\mathcal{U}= \{u_1,u_2,u_3,u_4,u_5,u_6\}$ contains $6$ members and the itemset  $\mathcal{I} = \{i_1,i_2, i_3\}$ has $3$-items. The user's preference for the itemset is given (or predicted) as in Table~\ref{tab:ex5}. Imagine that the user set needs to be partitioned into at most $3$ groups ($\ell \leq 3$), where each group has to be recommended $k=2$ items.

\begin{example}\label{ex5}
{\em 
Consider the user set \\
$\mathcal{U}= \{u_1,u_2,u_3,u_4,u_5,u_6\}$  and the itemset  $\mathcal{I} = \{i_1,i_2, i_3\}$. The users'  preferences for the itemset are given (or predicted) as in Table~\ref{tab:ex5}. Suppose the users need to be partitioned into at most $3$ groups ($\ell \leq 3$), where each group has to be recommended $k=2$ items. \qed }  
\end{example}

Using {\tt GRD-LM-SUM}, this will form the following $3$ groups at the end: $\{u_2\}$, $\{u_3,u_4\}$, $\{u_1,u_5,u_6\}$ with the overall objective function value as $(5+3)+(5+2)+(3+2)=20$, whereas, the optimal grouping will give a different solutions to form $3$ groups, $\{u_2,u_6\}$, $\{u_3,u_4\}$, $\{u_1,u_5\}$ with the overall objective function value of $(5+2)+(5+2)+(4+3)=21$. 

\begin{table}[!htb]
\begin{tabular}{l*{6}{c}r}
User-item Ratings & $u_1$ & $u_2$ & $u_3$ & $u_4$ & $u_5$  & $u_6$  \\
\hline
$i_1$ 	    & 1 & 2 & 2 & 2 & 2 & 1   \\
$i_2$       & 4 & 3 & 5 & 5 & 4 & 2   \\
$i_3$       & 3 & 5 & 1 & 1 & 3 & 5  \\
\end{tabular}
\caption{User Item Preference Rating for Example~\ref{ex5}\label{tab:ex5}}
\end{table}

\bibliographystyle{abbrv}
\bibliography{paperbib,paperbib1,paperbib2}    % sigproc.bib is the name of the Bibliography in this case

\begin{thebibliography}{10}

\bibitem{DBLP:journals/pvldb/Amer-YahiaRCDY09}
S.~Amer-Yahia, S.~B. Roy, A.~Chawla, G.~Das, and C.~Yu.
\newblock Group recommendation: Semantics and efficiency.
\newblock {\em PVLDB}, 2(1):754--765, 2009.

\bibitem{part2}
O.~Amini, F.~Mazoit, N.~Nisse, and S.~Thomass{\'e}.
\newblock Submodular partition functions.
\newblock {\em Discrete Mathematics}, 309(20):6000--6008, 2009.

\bibitem{anagnostopoulos2010power}
A.~Anagnostopoulos, L.~Becchetti, C.~Castillo, A.~Gionis, and S.~Leonardi.
\newblock Power in unity: forming teams in large-scale community systems.
\newblock In {\em Proceedings of the 19th ACM international conference on
  Information and knowledge management}, pages 599--608. ACM, 2010.

\bibitem{anagnostopoulos2012online}
A.~Anagnostopoulos, L.~Becchetti, C.~Castillo, A.~Gionis, and S.~Leonardi.
\newblock Online team formation in social networks.
\newblock In {\em Proceedings of the 21st international conference on World
  Wide Web}, pages 839--848. ACM, 2012.

\bibitem{baeza1999modern}
R.~Baeza-Yates, B.~Ribeiro-Neto, et~al.
\newblock {\em Modern information retrieval}, volume 463.
\newblock ACM press New York, 1999.

\bibitem{groupon1}
R.~Balakrishnan and R.~P. Bhatt.
\newblock Real-time bid optimization for group-buying ads.
\newblock In {\em Proceedings of the 21st ACM international conference on
  Information and knowledge management}, pages 1707--1711. ACM, 2012.

\bibitem{baykasoglu2007project}
A.~Baykasoglu, T.~Dereli, and S.~Das.
\newblock Project team selection using fuzzy optimization approach.
\newblock {\em Cybernetics and Systems: An International Journal},
  38(2):155--185, 2007.

\bibitem{DBLP:conf/recsys/BerkovskyF10}
S.~Berkovsky and J.~Freyne.
\newblock Group-based recipe recommendations: analysis of data aggregation
  strategies.
\newblock In {\em RecSys}, pages 111--118, 2010.

\bibitem{crossen2002flytrap}
A.~Crossen, J.~Budzik, and K.~J. Hammond.
\newblock Flytrap: intelligent group music recommendation.
\newblock In {\em Proceedings of the 7th international conference on
  Intelligent user interfaces}, pages 184--185. ACM, 2002.

\bibitem{gap2}
L.~Fleischer, M.~X. Goemans, V.~S. Mirrokni, and M.~Sviridenko.
\newblock Tight approximation algorithms for maximum general assignment
  problems.
\newblock In {\em Proceedings of the seventeenth annual ACM-SIAM symposium on
  Discrete algorithm}, pages 611--620. ACM, 2006.

\bibitem{DBLP:journals/eswa/GarciaSO11}
I.~Garcia, L.~Sebastia, and E.~Onaindia.
\newblock On the design of individual and group recommender systems for
  tourism.
\newblock {\em Expert Syst. Appl.}, 38(6):7683--7692, 2011.

\bibitem{DBLP:books/fm/GareyJ79}
M.~R. Garey and D.~S. Johnson.
\newblock {\em Computers and Intractability: A Guide to the Theory of
  NP-Completeness}.
\newblock 1979.

\bibitem{DBLP:books/mk/HanK2000}
J.~Han and M.~Kamber.
\newblock {\em Data Mining: Concepts and Techniques}.
\newblock Morgan Kaufmann, 2000.

\bibitem{reco1}
A.~Jameson and B.~Smyth.
\newblock Recommendation to groups.
\newblock In {\em The adaptive web}, pages 596--627. Springer, 2007.

\bibitem{kann1991maximum}
V.~Kann.
\newblock Maximum bounded 3-dimensional matching is max snp-complete.
\newblock {\em Information Processing Letters}, 37(1):27--35, 1991.

\bibitem{kann1992approximability}
V.~Kann.
\newblock {\em On the approximability of NP-complete optimization problems}.
\newblock PhD thesis, Royal Institute of Technology Stockholm, 1992.

\bibitem{gap1}
S.~O. Krumke and C.~Thielen.
\newblock The generalized assignment problem with minimum quantities.
\newblock {\em European Journal of Operational Research}, 228(1):46--55, 2013.

\bibitem{groupon2}
A.~Lacerda, A.~Veloso, and N.~Ziviani.
\newblock Exploratory and interactive daily deals recommendation.
\newblock In {\em RecSys}, 2013.

\bibitem{lappas2009finding}
T.~Lappas, K.~Liu, and E.~Terzi.
\newblock Finding a team of experts in social networks.
\newblock In {\em SIGKDD}, pages 467--476, 2009.

\bibitem{groupon3}
T.~Lappas and E.~Terzi.
\newblock Daily-deal selection for revenue maximization.
\newblock In {\em 21st {ACM} International Conference on Information and
  Knowledge Management, CIKM'12, Maui, HI, USA, October 29 - November 02,
  2012}, pages 565--574, 2012.

\bibitem{lieberman1999let}
H.~Lieberman, N.~Van~Dyke, and A.~Vivacqua.
\newblock Let's browse: a collaborative browsing agent.
\newblock {\em Knowledge-Based Systems}, 12(8):427--431, 1999.

\bibitem{DBLP:conf/er/NtoutsiSNK12}
E.~Ntoutsi, K.~Stefanidis, K.~N{\o}rv{\aa}g, and H.-P. Kriegel.
\newblock Fast group recommendations by applying user clustering.
\newblock In {\em ER}, pages 126--140, 2012.

\bibitem{polylens}
M.~O’connor, D.~Cosley, J.~A. Konstan, and J.~Riedl.
\newblock Polylens: a recommender system for groups of users.
\newblock In {\em ECSCW 2001}, pages 199--218. Springer, 2001.

\bibitem{Pizzutilo:2005:GMP:1366349.1366378}
S.~Pizzutilo, B.~De~Carolis, G.~Cozzolongo, and F.~Ambruoso.
\newblock Group modeling in a public space: Methods, techniques, experiences.
\newblock In {\em Proceedings of the 5th WSEAS International Conference on
  Applied Informatics and Communications}, AIC'05, pages 175--180, Stevens
  Point, Wisconsin, USA, 2005. World Scientific and Engineering Academy and
  Society (WSEAS).

\bibitem{romesburg2004cluster}
C.~Romesburg.
\newblock {\em Cluster analysis for researchers}.
\newblock Lulu. com, 2004.

\bibitem{DBLP:conf/sigmod/RoyICDE14}
S.~B. Roy, S.~Thirumuruganathan, S.~Amer-Yahia, G.~Das, and C.~Yu.
\newblock Exploiting group recommendation functions for flexible preferences.
\newblock In {\em ICDE Conference}, 2014.

\bibitem{senot2010analysis}
C.~Senot, D.~Kostadinov, M.~Bouzid, J.~Picault, A.~Aghasaryan, and C.~Bernier.
\newblock Analysis of strategies for building group profiles.
\newblock In {\em User Modeling, Adaptation, and Personalization}, pages
  40--51. Springer, 2010.

\bibitem{aris}
M.~Sozio and A.~Gionis.
\newblock The community-search problem and how to plan a successful cocktail
  party.
\newblock In {\em Proceedings of the 16th ACM SIGKDD international conference
  on Knowledge discovery and data mining}, pages 939--948. ACM, 2010.

\bibitem{wi2009team}
H.~Wi, S.~Oh, J.~Mun, and M.~Jung.
\newblock A team formation model based on knowledge and collaboration.
\newblock {\em Expert Systems with Applications}, 36(5):9121--9134, 2009.

\bibitem{wolsey1998integer}
L.~A. Wolsey.
\newblock {\em Integer programming}, volume~42.
\newblock Wiley New York, 1998.

\bibitem{tvrec}
Z.~Yu, X.~Zhou, Y.~Hao, and J.~Gu.
\newblock {TV Program Recommendation for Multiple Viewers Based on user Profile
  Merging}.
\newblock {\em User Modeling and User-adapted Interaction}, 16:63--82, 2006.

\bibitem{part1}
L.~Zhao, H.~Nagamochi, and T.~Ibaraki.
\newblock Greedy splitting algorithms for approximating multiway partition
  problems.
\newblock {\em Mathematical Programming}, 102(1):167--183, 2005.

\bibitem{zzkarian1999forming}
A.~Zzkarian and A.~Kusiak.
\newblock Forming teams: an analytical approach.
\newblock {\em IIE transactions}, 31(1):85--97, 1999.

\end{thebibliography}
%\newpage
%\appendix

%\input{proof}

\end{document}